\documentclass[11pt]{article}

\usepackage[english]{babel}

\usepackage[utf8]{inputenc}
\usepackage{amsmath}
\usepackage{amssymb}
\usepackage{amsthm}
\usepackage{bm}
\usepackage{bbm}
\usepackage{cancel}
\usepackage{centernot}
\usepackage{enumerate}
\usepackage{enumitem}
 \usepackage{url}
\usepackage[margin=1.0in]{geometry}

\usepackage{graphicx}
\usepackage{hyperref}
\usepackage{listings}
\usepackage{mathrsfs}
\usepackage{mathtools}
\usepackage{thm-restate}
\usepackage[round]{natbib}

\newtheorem{theorem}{Theorem}[section]
\newtheorem{lemma}[theorem]{Lemma}
\newtheorem{proposition}[theorem]{Proposition}
\newtheorem{corollary}[theorem]{Corollary}
\newtheorem{definition}[theorem]{Definition}
\newtheorem{example}[theorem]{Example}
\newtheorem{claim}[theorem]{Claim}
\newtheorem{assumption}[theorem]{Assumption}
\newtheorem{remark}[theorem]{Remark}

\DeclareMathOperator*{\argmin}{argmin}

\def\Var{{\rm Var}}

\DeclareMathOperator{\Unif}{Unif}
\DeclareMathOperator{\Binom}{Bin}
\DeclareMathOperator{\Bern}{Bern}
\DeclareMathOperator{\Poi}{Poi}
\DeclareMathOperator{\Exp}{Exp}

\DeclareMathOperator{\Perm}{Perm}

\DeclareMathOperator{\diag}{diag}
\newcommand{\bO}{\ensuremath{O}}
\newcommand{\Pp}{\ensuremath{\mathbb{P}}}
\newcommand{\E}{\ensuremath{\mathbb{E}}}
\newcommand{\R}{\ensuremath{\mathbb{R}}}
\newcommand{\Z}{\ensuremath{\mathbb{Z}}}
\newcommand{\N}{\ensuremath{\mathbb{N}}}

\newcommand{\eps}{\ensuremath{\epsilon}}

\newcommand{\floor}[1]{\lfloor #1 \rfloor}
\newcommand{\dd}{\ensuremath{\,d}}

\newcommand{\cM}{\mathcal{M}}
\newcommand{\cW}{\mathcal{W}}
\newcommand{\vu}{\mathbf{u}}
\newcommand{\vv}{\mathbf{v}}
\newcommand{\vw}{\mathbf{w}}
\newcommand{\vU}{\mathbf{U}}
\newcommand{\vV}{\mathbf{V}}
\newcommand{\vW}{\mathbf{W}}

\newcommand{\va}{\mathbf{a}}
\newcommand{\vb}{\mathbf{b}}
\newcommand{\vd}{\mathbf{d}}
\newcommand{\ve}{\mathbf{e}}
\newcommand{\vA}{\mathbf{A}}
\newcommand{\vB}{\mathbf{B}}
\newcommand{\vD}{\mathbf{D}}
\newcommand{\vE}{\mathbf{E}}
\newcommand{\vM}{\mathbf{M}}
\newcommand{\vx}{\mathbf{x}}
\newcommand{\vy}{\mathbf{y}}

\newcommand{\vX}{\mathbf{X}}
\newcommand{\vY}{\mathbf{Y}}
\newcommand{\vZ}{\mathbf{Z}}
\newcommand{\pxy}{p(\vx,\vy)}
\newcommand{\qxy}{q(\vx,\vy)}

\newcommand{\sumiton}{\sum_{i=1}^n}

\newcommand{\sumjton}{\sum_{j=1}^n}

\newcommand{\m}{\ensuremath{\mathsf{m}}}
\newcommand{\w}{\ensuremath{\mathsf{w}}}
\newcommand{\Oeig}{\Omega_{\text{eig}}}
\newcommand{\Oeigz}{\Omega_{\text{eig}}(\zeta)}

\newcommand{\Oempe}{\Omega_{\text{emp}}(\eps)}
\newcommand{\tOemp}{\tilde{\Omega}_{\text{emp}}}

\newcommand{\Otailab}{\Omega_{\text{tail}}(\xi,\rho)}
\newcommand{\tOtailab}{\tilde{\Omega}_{\text{tail}}(\xi,\rho)}
\newcommand{\Femp}{\hat{\mathcal{F}}}
\newcommand{\Eratio}{\mathcal{E}_{\text{ratio}}}

\usepackage[colorinlistoftodos,textsize=footnotesize]{todonotes}

\newcounter{relctr} %
\everydisplay\expandafter{\the\everydisplay\setcounter{relctr}{0}} %
\newcommand\labelrel[2]{%
  \begingroup
    \refstepcounter{relctr}%
    \stackrel{\textnormal{(\roman{relctr})}}{\mathstrut{#1}}%
    \originallabel{#2}%
  \endgroup
}
\AtBeginDocument{\let\originallabel\label} %

\title{Welfare Distribution in Two-sided Random Matching Markets}

\author{Itai Ashlagi\;\!\thanks{Department of Management Science and Engineering, Stanford University} \and Mark Braverman\;\!\thanks{Department of Computer Science, Princeton University} \and Geng Zhao\;\!\thanks{Department of Electrical Engineering and Computer Sciences, University of California, Berkeley}}

\date{\today}

\begin{document}

\maketitle

\begin{abstract}
We study the welfare structure in two-sided large random matching markets. In the model, each agent has a latent personal score for every agent on the other side of the market and her preferences follow a logit model based on these scores. Under a contiguity condition, we provide a tight description of stable outcomes. 

First, we identify an intrinsic fitness for each agent that represents her relative competitiveness in the market, independent of the realized stable outcome.  The intrinsic fitness values correspond to scaling coefficients  needed to make a mutual latent matrix bi-stochastic, where the latent scores can be interpreted as a-priori probabilities of a pair being matched.

Second, in every stable (or even approximately stable) matching, the welfare or the ranks of the agents on each side of the market, when scaled by their intrinsic fitness, have an approximately exponential empirical distribution. Moreover, the average welfare of agents on one side of the market is sufficient to determine the average on the other side.

Overall, each agent's welfare is determined by a global parameter, her intrinsic fitness, and an extrinsic factor with exponential distribution across the population.
\end{abstract}

\section{Introduction}

This paper is concerned with the welfare in random two-sided matching markets. In a two-sided matching market there are two kinds of agents, where each agent has preferences over potential partners of the other kind. We assume that the outcome  is stable  \citep{gale1962college}, meaning that there are no blocking pairs of agents who would rather match to each other over the their assigned partners.

A large literature initiated by \citep{gale1962college} has deepened our understanding of two-sided matching markets, generating a blend of rich theory and market designs.\footnote{See, e.g., \citet{roth1992two,roth2018marketplaces}.} Less understood, however are welfare properties in typical markets. We study the welfare structure in matching markets when agents have latent preferences generated  according to observed characteristics. Specifically we are interested in the empirical welfare distribution of agents on each side of the market under stable outcomes as well as the relation between the outcomes of each side of the market.%

We study this question in large randomly generated markets, which allow for both vertical and horizontal differentiation. The model assumes that every agent has an observed personal score for every other agent in the market, and her preferences follows a Logit model based on these scores. We impose that no agent is a-priori overwhelmingly more desirable than any other agent. We find that the observed characteristics alone determine the empirical welfare distribution on each side of the market. Moreover, the joint surplus in the market is fixed, and the average welfare of one side of the market is a sufficient statistic to determine the empirical welfare distribution on both sides of the market.

The model we consider has an equal number of men and women. For every man $\m_i$ and every woman $\w_j$, we are given non-negative scores $a_{ij}$ and $b_{ji}$, which %
can be viewed as generated from observed characteristics. 
Each man and each woman have strict  preference rankings  generated independently and proportionally to these latent scores, as in the Logit model.\footnote{Numerous empirical papers that study two-sided matching market assume agents' preferences follow a logit model (see e.g., \cite{agarwal2018demand,hitsch2010matching}).} Equivalently, each man $\m_i$ has a latent value from matching with woman $\w_j$ that is distributed exponentially with rate   $a_{ij}$  (smaller values are considered  better).\footnote{One can view the utility of an agent for her match to be the negative of the corresponding latent value.}  Women's latent values for men are generated similarly.\footnote{Special cases of this general model are markets with  uniformly random preferences \citep{knuth1990stable,pittel1989average,knuth1997stable,pittel1992likely, ashlagi2017unbalanced} or when agents have common public scores \citep{mauras2021two,ashlagi2020tiered}.}

We identify an intrinsic fitness for each agent that represents her relative competitiveness in the market, independent of the realized stable outcome. For every pair of agents on opposing sides of the market, we can obtain a mutual score of the pair’s match. If we write these scores in a matrix, the intrinsic fitness values
correspond to scaling coefficients that make the mutual matrix bi-stochastic.\footnote{This representation is valid since preferences are invariant under such transformations.} Intuitively, this bi-stochastic mutual matrix can be thought of as consisting of {\em a-priori} probabilities of each pair matching. In particular, this  representation captures the interactions between the sides of the market. We exploit this  representation to further  analyze typical realized outcomes in the market.

We find that the welfare, or the ranks of the agents, when scaled by their intrinsic fitness, have an approximately exponential empirical distribution on each side of the market. Moreover, the average welfare of agents on one side of the market is sufficient to determine the average on the other side. Overall, each agent’s welfare can be seen as determined by a global parameter, her intrinsic fitness, and an extrinsic factor with exponential distribution across the population. This characterization holds with high probability in every stable  matching. In fact, this structure extends to matchings that are only approximately stable, which can tolerate a vanishing fraction of blocking pairs.

At its core, since our proof needs to apply to all stable matchings (and even to nearly-stable matchings), it is a union bound argument. We use inequalities derived from the integral formula for the probability that a given matching is stable, first introduced by \citet{knuth1976mariages}. The heterogeneous preferences brings great difficulty, which we overcome with a truncation technique to accommodate heavy tails of agents' outcomes and a fixed-point argument on the eigenspace of the characterizing matrix of the market. The exponential empirical distribution part of the result holds intuitively because there are not too many stable matchings in expectation, and the exponential distribution has the highest entropy of all non-negative distributions with a given mean.

Closely related to our work is the remarkable paper    \cite{menzel2015large}, which  finds that the joint surplus in  the market is unique. The focus in  \cite{menzel2015large} is  on analyzing the matching rates between agents of different types, rather than the rankings and agents' welfare. Menzel's preference model is more general.\footnote{We note that both his and our model assume that the ratio between any two systematic scores is bounded.} Menzel establishes that, at the limit, agents choose partners according to a logit model from opportunity sets, while we consider large markets and assume agents' preferences are  logit based. There are several other key  differences. First, his model requires many agents of each type (with the same characteristics), while every agent in our model may have different characteristics.  Second, while in our model every agent is matched, he assumes agents have a non-negligible outside option resulting in a large number of unmatched agents\footnote{\citet{menzel2015large} identifies how to scale the market under this assumption to capture realistic outcomes.}; this assumption allows him to apply a fixed point contraction argument and establish the uniqueness and characterization result.\footnote{Technically, such substantial outside options keep rejection chains short and prevent them from cycling.}

\subsection{Literature}

The analysis of random  two-sided markets goes back to \citet{knuth1990stable,pittel1989average,pittel1992likely}, who consider markets with  uniformly random complete preference lists.  These papers establish the number of stable matchings as well as the average ranks on each side. A key finding is that the product of there average rank of agents on each side of the market is approximately the size of the market \citep{pittel1992likely}, implying that  stable matchings  essentially  lie on a parabola. Our findings generalize these findings to markets to random logit markets. We also expand these findings to describe the distributional outcomes in the market.

Several papers consider markets with uniformly drawn preferences with an unequal number of agents on each side of market \citep{ashlagi2017unbalanced,pittel2019likely,cai2019short}.
A key finding is there is an essentially unique stable matching and agents on the short side have a substantial advantage. We believe that similar findings hold in random logit markets. Since our results hold for approximately stable matches, our findings  extend to the imbalanced case as long as the imbalance is not too large.%

Our paper further contributes to the above literature by considering also outcomes that are  approximately stable outcomes.

Several papers study markets random markets when (at least on one side) agents'  preferences are  generated proportionally to public scores. \citep{immorlica2015incentives,kojima2009incentives,ashlagi2014stability} look at the size of the core.\footnote{They further consider the related issue of strategizing under stable matching mechanisms.} Their analysis relies on a certain market structure (keeping preference lists short), which leaves many agents unmatched. %
 \cite{gimbert2019popularity}  and \citet{ashlagi2020tiered}  assume agents have complete preference lists and their focus is on the  size of the core or agents' average rank.

\subsection{Notations}

Denote $[n] = \{1,\ldots,n\}$. Boldface letters denote vectors (lower case) and matrices (upper case), e.g., $\vx = (x_i)_{i\in[n]}$ and $\vA = (a_{ij})_{i\in[n],j\in [m]}$, and capital letters denote random variables.

For two identically shaped matrices (or vectors) $\vM$ and $\mathbf{N}$, $\vM\circ \mathbf{N}$ denotes their Hadamard (entry-wise) product. For a vector $\vx\in\R^n$ with non-zero entries, denote its coordinate-wise inverse by $\vx^{-1}$. $\diag(\vx)$ denotes the diagonal matrix whose $i$-th entry on the diagonal is $x_i$.

$\Exp(\lambda)$ and $\Poi(\lambda)$ denote, respectively, the exponential distribution and the Poisson distribution with rate $\lambda$. We  denote the probability density function (pdf) and cumulative distribution function (CDF) of $\Exp(\lambda)$ by $f_\lambda$ and $F_\lambda$, respectively. %
$\Bern(p)$ denotes the Bernoulli distribution with success probability $p\in[0,1]$. %
For distributions $\mathcal{D}_1$ and $\mathcal{D}_2$ over space $\mathcal{X}$, $\mathcal{D}_1\otimes \mathcal{D}_2$ denotes their product distribution over $\mathcal{X}^2$.
$\Femp(\vx)$ denotes the empirical distribution function for the components of a vector $\vx$, treated as a function from $\R$ to $[0,1]$. $\mathcal{F}(\mathcal{D})$ denotes the CDF of a distribution $\mathcal{D}$ on $\R$.
For real-valued random variables $X$ and $Y$, $X\preceq Y$ denotes stochastic domination of $X$ by $Y$.

We use the standard $O(\cdot)$, $o(\cdot)$, $\Omega(\cdot)$, and $\Theta(\cdot)$ notations to hide constant factors. For functions $f,g:\N\to\R_+$, we say $f = O(g)$ (resp. $\Omega(g)$) if there exists an absolute constant $K\in(0,\infty)$ such that $f \le K g$ (resp. $f \ge K g$) for $n$ sufficiently large; $f=o(g)$ if $f/g \to 0$ as $n\to\infty$; and $f=\Theta(g)$ if $f=O(g)$ and $f=\Omega(g)$. We say $f = o_\alpha(g)$ if $f/g\to 0$ as $\alpha\to 0$ (uniformly over all other parameters, such as $n$). For example, $\sqrt{\eps} = o_\eps(1)$.

\section{Model}

We study two-sided matching markets with randomly generated preferences. Next we formalize the model, how preferences are generated and key assumptions. 

\paragraph{Setup.}  A matching market consists of two sets of agents, referred to as men $\mathcal{M}$ and women $\mathcal{W}$. Unless specified otherwise, we assume that  $|\mathcal{M}| = |\mathcal{W}| = n$, men are labeled $\m_1,\ldots, \m_n$ and  women are labeled $\w_1,\ldots, \w_n$. 
Each man $\m_i$ has a complete strict preference list $\succ_{\m_i}$ over the the set of women and each woman  $\w_j$ has a complete strict preference list $\succ_{\w_j}$  over the set of  men.   A \emph{matching} is a bijection $\mu : \mathcal{M}\to\mathcal{W}$. To simplify the notation, men and women will  be presented using the set of integers $[n]=\{1,2,\ldots,n\}$ and we write $\mu:[n]\to[n]$ so that $\mu(i)=j$ and $\mu^{-1}(j)=i$ means that  $\m_i$ is matched with  $\w_j$ in  $\mu$. The \emph{rank} for man $\m_i$, denoted by $R_i(\mu)$,  is the position of $\mu(i)$ on $\m_i$'s preference list (e.g., if an agent is matched to the second agent on her list,  her rank is two). Write $\mathbf{R}(\mu):=(R_i(\mu))_{i\in[n]}$ for the men's rank vector in matching $\mu$. 

The matching $\mu$ is \emph{unstable} if there is a pair of  man $\m_i$ and woman $\w_j$ such that $\w_j \succ_{\m_i} \w_{\mu(i)}$ and $\m_i \succ_{\w_j} \w_{\mu^{-1}(j)}$. A matching is said to  \emph{stable} otherwise. It is well-known that the set of stable matchings is not empty.

\paragraph{Logit-based random markets: the canonical form.} We consider markets in which complete preferences are randomly generated as follows. %
For each man $\m_i$, we are given a stochastic vector $\hat{\mathbf{a}}_i = (\hat{a}_{ij})_{j\in[n]} \in\R^n_+$. Then, $\m_i$'s preference list  is generated from a logit model based on $\hat{\mathbf{a}}_i$. In  particular, let $\mathcal{D}_i$ be the distribution on $\cW$ that places on $\w_j$ a probability proportional to $\hat{a}_{ij}$; then $\m_i$ samples from $\mathcal{D}_i$ for his favorite partner, and repeatedly sample from it without replacement for his next favorite partner until completing his list.
Similarly, each woman $\w_j$ preference list is generated from a logit model based on a given stochastic vector $\hat{\mathbf{b}}_j = (\hat{b}_{ji})_{i\in[n]}$.
Denote by  $\hat{\vA} = (\hat{a}_{ij})_{i,j\in[n]}$ and $\hat{\vB} = (\hat{b}_{ji})_{j,i\in[n]}$ the  row-stochastic matrices. We refer to this matrix representation of the preference model as the  \emph{canonical form} and to $\hat{a}_{ij}$ (resp. $\hat{b}_{ji}$) as the \emph{canonical score} that $\m_i$ (resp. $\w_j$) assigns to $\w_j$ (resp. $\m_i$).

This model captures the multinomial logit (MNL) choice model, in which  scores are closely related to the systematic utilities for agents over matches. The special case in which $\hat{a}_{ij} = \hat{b}_{ji} = 1/n$ for all $i,j\in[n]$ corresponds to the uniformly random preference model.

\paragraph{Mutual matrix and intrinsic fitness: the balanced form.}
While the canonical form is a useful way to describe the market, it will be helpful for the analysis to describe it using an  alternative scaling scheme, which we refer to as the \emph{balanced form}.

Observe that multiplying any row of $\hat{A}$ and $\hat{B}$ by a constant does not change the behavior of the market.
We look for scaling vectors $\bm{\phi},\bm{\psi}\in\R^n_+$ for the rows of $\hat\vA$ and $\hat\vB$ such that $\vM = n^{-1} \vA \circ \vB$ is bistochastic\footnote{The nonnegative matrix $\vM$ is bistochastic if the sum of entries in each row and each column is one.}, where $\vA = \diag(\bm{\phi}) \hat\vA$ and $\vB = \diag(\bm{\psi}) \hat\vB$. As is shown by \citet[Theorem~1]{sinkhorn1964relationship}, such a bistochastic matrix $\vM$ always uniquely exists, and the scaling vectors $\bm{\phi}$ and $\bm{\psi}$ are unique up to constant rescaling. That is, $\bm{\phi}$ and $\bm{\psi}$ jointly solve
\begin{equation}
    \frac{1}{n} \diag(\bm{\phi}) (\hat{\vA} \circ \hat{\vB}^\top) \diag(\bm{\psi}) \mathbf{1} = \mathbf{1} \quad\text{ and }\quad \frac{1}{n} \diag(\bm{\psi}) (\hat{\vB} \circ \hat{\vA}^\top) \diag(\bm{\phi}) \mathbf{1} = \mathbf{1},
\end{equation}
where $\mathbf{1}$ is the vector consisting of all $1$'s.
The matrix $\vM$ will be referred to as the \emph{mutual matrix}. %

In the remainder of the paper we assume without loss of generality that the market is described in the balanced form, using $\vA,\vB$  and the mutual matrix $\vM$.

The bistochasticity constraint incurs the following relationship: if $\hat{b}_{ji}$'s increase (resp. decrease) by a factor of $\alpha$ simultaneously for all $j\in[n]$, the scaling factor $\phi_i$, and hence all $a_{ij}$'s for $j\in[n]$, must decrease (resp. increase) by the same factor to maintain bistochasticity of $\vM$. In other words, a uniform increase (resp. decrease) of $\m_i$'s popularity among the women will lead to a proportional decrease (resp. increase) in $\phi_i$. Thus, we can view the $\phi_i$ as reflecting the ``average popularity'' of man $\m_i$ among the women: Loosely speaking, the smaller $\sumjton a_{ij}$ is, the more popular $\m_i$ is (reflected by larger values of $b_{ji}$'s). 

We  refer to the vector $\bm{\phi}$ and $\bm{\psi}$ as the men's and women's \emph{intrinsic fitness} vector, respectively (and note that a smaller intrinsic fitness value means the agent is more competitive). Note that since $\hat{\vA} = \diag(\bm{\phi})^{-1} \vA$ is row-stochastic, we conveniently have $\phi_i = \sumjton a_{ij}$, and similarly $\psi_j = \sumiton b_{ji}$ in the balanced form.

\begin{example}
[Markets with  public  scores]\label{Ex_public_scores}
We say a matching market has public scores when $\hat{\mathbf{a}}_i=\hat{\mathbf{a}}\in\R_+^n$ for all $i\in[n]$ and $\hat{\mathbf{b}}_j=\hat{\mathbf{b}}\in\R_+^n$ for all $j\in[n]$. In other words, agents on the same side of the market share an identical preference distribution. %
The fitness vectors are simply $\bm{\phi} = \hat{\mathbf{b}}^{-1}$ and $\bm{\psi} = \hat{\mathbf{a}}^{-1}$, where the inverse is taken component-wise. The mutual matrix $\vM = \mathbf{J} := (n^{-1})_{i,j,\in[n]}$ in this case.
\end{example}

\paragraph{Latent values.}
The logit-based preference model can be generated equivalently in the following way.
Let $\vX,\vY\in\R_+^{n\times n}$ be two random matrices with independent entries $X_{ij}$ (resp. $Y_{ji}$) sampled from   $\Exp(a_{ij})$ (resp. $\Exp(b_{ji})$). The preference profile is then derived from $\vX$ and $\vY$ as follows:
\[
    \w_{j_1}\succeq_{\m_i}\w_{j_2} \quad\Longleftrightarrow\quad X_{ij_1} < X_{ij_2},
\]
\[
    \m_{i_1} \succeq_{\w_j} \m_{i_2} \quad\Longleftrightarrow\quad Y_{ji_1} < Y_{ji_2}.
\]
We  refer to each $X_{ij}$ (resp. $Y_{ji}$) for $i,j\in[n]$ as the \emph{latent value} (or simply {\em value}) of $\m_i$ (resp. $\w_j$) if matched with $\w_j$ (resp. $\m_i$). 

Note that for every agent, a lower rank implies a lower latent value (and therefore lower values of rank and latent value are better).

\paragraph{Regularity assumption.}

We  study the asymptotic behavior of two-sided matching markets as the market size grows large. Informally, we restrict attention to contiguous markets, in the sense that,  ex ante, no agent finds any other agent (on the opposite side of the market) disproportionately favorable or unfavorable to  other  agents. The  condition is formalized as follows.

A matrix $\mathbf{L}\in\R^{n\times n}$ with non-negative entries is called \emph{$C$-bounded} for some constant $C\ge 1$ if $\ell_{ij}\in[1/C,C]$ for all $1\le i,j\le n$. When $\mathbf{L}$ is (bi-)stochastic, we will abuse notation and say $\mathbf{L}$ is $C$-bounded if $n\mathbf{L}$ satisfies the definition above.

\begin{assumption}[Contiguity]\label{Assumption_C_bounded}
We assume that, by choosing an appropriate scaling of $\bm{\phi}$ and $\bm{\psi}$ in the balanced form,
there exist absolute constants $C\in[1,\infty)$ and $n_0<\infty$ such that $\vA$, $\vB$, and $n\vM=\vA\circ \vB^\top$ are all $C$-bounded for all $n\ge n_0$; that is, there exists $C\in[1,\infty)$ such that
\begin{equation}\label{Eqn_assumpt_C_bound_main}
    \frac{1}{C} \le \min_{i,j\in[n]} \min\{a_{ij}, b_{ji}, nm_{ij}\} \le \max_{i,j\in[n]} \max\{a_{ij}, b_{ji}, nm_{ij}\} \le C \quad\text{ for all }\; n > n_0.
\end{equation}

\end{assumption}

\begin{remark}
It is easy to verify that Assumption~\ref{Assumption_C_bounded} holds when no agent finds any potential partner disproportionately favorable or unfavorable based on their canonical scores:
    If $\hat\vA$ and $\hat\vB$ are $C$-bounded, then there exists a choice of $\bm{\phi}$ and $\bm{\psi}$ with all entries in $[n/C^2, nC^2]$ in the balanced form; further, $\vM $ is $C^4$-bounded.
Thus, Assumption~\ref{Assumption_C_bounded} is equivalent to the existence of an absolute upper bound on the ratio between pairs of entries within the same row of $\vA$ or $\vB$; that is
\begin{equation}
    \limsup_{n\to\infty} \max_{i,j_1,j_2\in[n]} \frac{a_{ij_1}}{a_{ij_2}} < \infty \qquad\text{ and }\qquad \limsup_{n\to\infty} \max_{j,i_1,i_2\in[n]} \frac{b_{ji_1}}{b_{ji_2}} < \infty.
\end{equation}
This condition is agnostic to scaling of the matrices and hence easy to certify. However, the lower and upper bounds in \eqref{Eqn_assumpt_C_bound_main} are more convenient in our later analysis, where the constant $C$ will make an appearance (although often made implicit in the results).
\end{remark}

\begin{remark}
Assumption~\ref{Assumption_C_bounded} offers a strong contiguity condition on the market, in that the attractiveness among all pairs of men and women vary at most by an (arbitrarily large) constant factor as the market grows.   We expect  the results to  hold under a  weaker assumption, which can be described through the spectral gap of the  matrix $\vM$. Recall that, as a bistochastic matrix, $\vM$ has a largest eigenvalue of $1$ and all other eigenvalues of magnitude at most $1$. We may think of the market as contiguous in this weaker sense if the spectral gap of $\vM$, given by $1-|\lambda_{\max}(\vM-\mathbf{J})|$, is bounded away from zero as the market grows. The spectral gap is a common and powerful notion when studying the structure of networks and communities.\footnote{In our model of the matching market, the spectral gap of $\vM$ describes the extent to which the market interconnects globally (contiguity) or decomposes into multiple sub-markets (modularity). A larger spectral gap means that the market is more cohesive, with more uniform or homogeneous preferences. For instance, the uniform market with $\vM=\mathbf{J}$ has a unit spectral gap, the maximum possible value. On the other hand, a smaller spectral gap means that the market is more clustered, with a clearer boundary between communities and poorly mixed preferences. For instance, any block-diagonal bistochastic matrix (with more than one blocks) has a zero spectral gap, and corresponds to a market that decomposes into two or more independent sub-markets --- one cannot hope to have a uniform structure result in such markets.} We impose Assumption \ref{Assumption_C_bounded} as it  simplifies  substantially the analysis and exposition. 
\end{remark}

\section{Main results}

We  denote the (random) set of stable matchings by $\mathcal{S}$. Recall that for a matching $\mu$, $\vX(\mu)$ and $\mathbf{R}(\mu)$ denote men's value and rank vectors, respectively,  under $\mu$. Denote by $\Femp(\vv)$  the empirical distribution of the components of a vector $\vv$ (viewed as a function from $\R$ to $[0,1]$), and $F_\lambda$ denotes the CDF of $\Exp(\lambda)$.

\begin{theorem}[Empirical distribution of values]\label{Thm_main_happiness_dist}
For any fixed $\eps>0$,
\begin{equation}\label{Eqn_happiness_dist_main_thm_whp}
    \Pp\bigg(\max_{\mu\in\mathcal{S}} \inf_{\lambda\in\R_+} \|\Femp(\vX(\mu))-F_\lambda\|_\infty \le \eps\bigg) \to 1 \text{ as } n\to\infty.
\end{equation}
That is, with high probability, in all stable matchings simultaneously, the empirical distribution of the men's values  is arbitrarily close to some exponential distribution $\Exp(\lambda)$ in Kolmogorov-Smirnov norm, where the parameter $\lambda$ depends on the specific stable matching. In particular, the infimum over $\lambda$ in \eqref{Eqn_happiness_dist_main_thm_whp} can be replaced with the choice of $\lambda$ that can be computed from the women's value vector $\vY(\mu)$.
\end{theorem}

\begin{theorem}
[Empirical distribution of ranks]\label{Thm_main_rank_dist}
For any fixed $\eps>0$,
\begin{equation}\label{Eqn_rank_dist_main_thm_whp}
    \Pp\bigg(\max_{\mu\in\mathcal{S}} \inf_{\lambda\in\R_+} \|\Femp(\bm{\phi}^{-1}\circ\mathbf{R}(\mu))-F_\lambda\|_\infty \le \eps\bigg) \to 1 \text{ as } n\to\infty,
\end{equation}
where $\bm{\phi}$ is the fitness vector (and can be computed as $\phi_i = \sumjton a_{ij}$).
That is, with high probability, in all stable matchings simultaneously, the empirical distribution of rescaled ranks of the men is arbitrarily close to some exponential distribution $\Exp(\lambda)$ in Kolmogorov-Smirnov norm, where the parameter $\lambda$ depends on the specific stable matching (yet the scaling doesn't). Again, we may replace the infimum with the choice of $\lambda$ that can be computed from the women's latent value vector $\vY(\mu)$.%
\end{theorem}

\subsection{Discussion}

The results characterize outcomes of all stable matchings. A slight refinement of Theorem~\ref{Thm_main_happiness_dist} will imply   that the average value of one side of the market is essentially sufficient to  determine the average value of the  other side.  Roughly, for a given stable matching $\mu$, the value of $\lambda$ in \eqref{Eqn_happiness_dist_main_thm_whp} and \eqref{Eqn_rank_dist_main_thm_whp}  is approximately the sum of the women's values in $\mu$.\footnote{Technically, the choice of $\lambda$ can be taken as the sum of the values of women after excluding a small fraction $\delta$  of the women who are the least satisfied (those with the highest latent values) under the matching $\mu$. This truncation, which is also done for technical reasons,  avoid outliers and in fact shows that the predictions still hold under even weaker notions of stability. We believe that such trimming is unnecessary with a more careful analysis.} This suggests that the average value of men is approximately $1/\lambda \approx 1/\|\vY(\mu)\|_1$. Therefore multiplying the  average values  of the two sides of the market gives approximately  $1/n$ simultaneously in all stable matchings with high probability. While we will establish such an approximation, we believe that, with a refined analysis, one should be able to show $\sup_{\mu\in\mathcal{S}}\big|n^{-1}\|\vX(\mu)\|_1\|\vY(\mu)\|_1 - 1\big| \overset{p}{\to} 0$. %

Moreover, the average value of men is also sufficient to predict the empirical value distribution on each side of the market.  For example, if we find  that $30\%$ of the men have value $h$ or higher, then we should expect $9\%$ to have value $2h$ or higher.  %

Theorem \ref{Thm_main_rank_dist} is similar  but with respect to ranks; it implies that the product of the average scaled ranks  of men and women should be asymptotically  $n$, and the the average rank on each side determines the empirical rank distributions.

Observe that the scaling in \eqref{Eqn_rank_dist_main_thm_whp} is consistent with the intuition of $\phi_i=\sumjton a_{ij}$ being the average fitness of $\m_i$. Within a stable matching, a more popular man should, on average, achieve a better (smaller) rank than a less popular one.
For instance, in a market with bounded public scores (Example~\ref{Ex_public_scores}), each man receives a number of proposals roughly inversely proportional to his fitness during the woman-proposing deferred acceptance algorithm,
implying that his rank is proportional to $\phi_i$ in the woman optimal stable matching.

The proof of Theorem~\ref{Thm_main_happiness_dist} also offers evidence that the number of stable matchings should essentially be sub-exponential. This is formally stated in Corollary~\ref{Cor_subexp_num_stable_match}.

\subsection{Results for approximately stable matchings}

The proof suggest that the characterization further extends to matchings that are only approximately  stable in the following sense.

\begin{definition}
We say a matching $\mu$ between $\cM$ and $\cW$ is $\alpha$-stable for some $0 < \alpha < 1$ if there exists a sub-market of size at least $(1-\alpha) n$ on which $\mu$ is stable; that is, there exist subsets $\cM'\subseteq \cM$ and $\cW' \subseteq \cW$ both with cardinality $|\cM'|=|\cW'| \ge (1-\alpha) n$ such that $\mu(\cM') = \mu(\cW')$ and the partial matching induced by $\mu$ between $\cM'$ and $\cW'$ is stable (within this sub-market). We refer to the stable sub-matching between $\cM'$ and $\cW'$ as the \emph{stable part} of $\mu$.
Denote the set of $\alpha$-stable matchings by $\mathcal{S}_\alpha$.
\end{definition}

The following Theorem can be derived from the quantitative versions of Theorem~\ref{Thm_main_happiness_dist} and \ref{Thm_main_rank_dist}, which will be presented in Section~\ref{sec_distribution}.

\begin{restatable}{theorem}{ThmMainApproxStable}\label{Thm_dist_body_approx_stable}
    Assume $\alpha < n^{-\eta}$ for some constant $\eta > 1/2$. Then, as $n\to\infty$,
    \begin{equation}\label{Eqn_happiness_dist_approx_stable}
        \max_{\mu\in\mathcal{S}_\alpha} \inf_{\lambda\in\R_+} \|\Femp(\vX(\mu))-F_\lambda\|_\infty \overset{p}{\to} 0 \quad\text{ and }\quad \max_{\mu\in\mathcal{S}_\alpha} \inf_{\lambda\in\R_+} \|\Femp(\bm{\phi}^{-1}\circ\mathbf{R}(\mu))-F_\lambda\|_\infty \overset{p}{\to} 0.
    \end{equation}
\end{restatable}

The approximately exponential empirical distribution applies to any matching that is stable except for $o(\sqrt{n})$ agents. The key observation is that, if the entire market satisfies the contiguity assumption, then each sub-market of it is also contiguous, and we can apply union bound over all sub-markets of size $(1-\alpha)n$. 
As a corollary, we have the following result for slightly imbalanced markets.

\begin{restatable}{corollary}{CorImbalanceMarket}\label{Cor_body_imbalance}
    Consider a market consisting of $n-k$ men and $n$ women, where $k < n^{\beta}$ for some constant $\beta < 1/2$. Assume that the contiguity condition holds as in Assumption~\ref{Assumption_C_bounded}. Then, as $n\to\infty$,
    \begin{equation}\label{Eqn_happiness_dist_imbalance}
        \max_{\mu\in\mathcal{S}} \inf_{\lambda\in\R_+} \|\Femp(\vX(\mu))-F_\lambda\|_\infty \overset{p}{\to} 0 \quad\text{ and }\quad \max_{\mu\in\mathcal{S}} \inf_{\lambda\in\R_+} \|\Femp(\bm{\phi}^{-1}\circ\mathbf{R}(\mu))-F_\lambda\|_\infty \overset{p}{\to} 0.
    \end{equation}
\end{restatable}

\begin{remark}
    The results will not hold if there is a linear imbalance in the market, i.e., $k \propto n$. This is because in such markets the men achieve a constant average rank \cite{ashlagi2017unbalanced}, and therefore the convergence to the exponential distribution is impossible. 
\end{remark}

These results are not necessarily tight and one may weaken the constraints on $\alpha$ and $k$ with a more careful analysis. Exploring other notions of approximate stability is left for future work. %

\section{Intuition and Proof Ideas \label{sec_prep_proof}}

This section offers intuition and the  key ideas behind the proofs.

\subsection{Intuition}

Let us start with providing a high-level intuition of both why the result is true, and how we should expect the proof to go. The actual proof does not follow the intuition exactly due to some technical difficulties that need to be overcome. It is possible that one can find a proof that follows the intuition below more directly. 

At a very high level, the result follows from a {\em union bound}. There are $n!$ potential matchings. Based on {\em a-priori} preferences, for each matching $\mu$, one could compute the probability $P_\mu$ that $\mu$ is stable under the realized preferences. A union bound argument just establishes that 
\begin{equation}\label{eq:int1}
    \sum_{\text{$\mu$ does not satisfy the conditions of \eqref{Eqn_happiness_dist_main_thm_whp}}} P_\mu = \exp(-\Omega_\epsilon(n))
\end{equation}

Establishing \eqref{eq:int1} directly appears to be difficult. A more approachable statement is a universal bound on the number of stable matchings overall:
\begin{equation}\label{eq:int2}
    \sum_{\mu} P_\mu = \exp(o(n)).
\end{equation}
That is, in expectation, there are not so many stable matchings. 

Consider the triplet of random variable $(X,Y,\mu)$, where $X$ and $Y$ are preferences sampled according to the model, and $\mu$ is a uniformly random matching. Let $\mathcal{S}$ be the event that $\mu$ is stable under preference profiles $(X,Y)$. If there are few stable matchings overall, as \eqref{eq:int2} implies, then we have
\begin{equation}
    \label{eq:int3}
\Pp(\mathcal{S})=\exp(o(1))\cdot n!^{-1}
\end{equation}
Another way of uniformly sampling from $\mathcal{S}$ is as follows. 
First, sample $(X_1,Y,\mu)\in_U \mathcal{S}$. Then resample $X_2$ conditioned on $(X_2,Y,\mu)\in \mathcal{S}$. The triple $(X_2,Y,\mu)$ is a uniform element of $\mathcal{S}$. Note that for a fixed $(Y,\mu)$ the marginal distribution of $X_2$ conditioned on 
$(X_2,Y,\mu)\in\mathcal{S}$ is fairly simple to reason about: each member should prefer the pairing assigned to them by $\mu$ to all other potential blocking matches. In such a resampling, as we shall see, the empirical exponential distribution appears naturally from large deviations theory.  

Suppose we prove that for all $(Y,\mu)$,
\begin{equation}
    \label{eq:int4}
\Pp_{X_2:(X_2,Y,\mu)\in\mathcal{S}}(\text{$X_2$ does not satisfy the conditions of \eqref{Eqn_happiness_dist_main_thm_whp}} ) = \exp(-\Omega_\epsilon(n))
\end{equation}
Putting these together, we would get 
\begin{multline*}
    \Pp\big((X_2,Y,\mu)\in \mathcal{S} \wedge(\text{$X_2$ does not satisfy the conditions of \eqref{Eqn_happiness_dist_main_thm_whp}} )\big) \\ = \Pp((X_1,Y,\mu)\in \mathcal{S}) \cdot \exp(-\Omega_\epsilon(n)) = n!^{-1}\cdot  \exp(-\Omega_\epsilon(n)),
\end{multline*}
implying \eqref{eq:int1} (together with a similar statement about $Y$).

A certain amount of technical work is needed to make the above blueprint go through. In particular, since our bounds need to be 
pretty tight, we need to worry about tail events. We end up having to perform the above resampling trick multiple times. 

\paragraph{Why do we need the boundedness assumption \ref{Assumption_C_bounded}?} It is worth noting that while the boundedness assumption might not be the weakest assumption under which our results hold, some assumptions on the market are inevitable. In particular, if the market can be split into two independent balanced markets $A$ and $B$, then there is no connection between the fortunes of men in market $A$ and their fortunes in market $B$, and the empirical distribution of values on each side will be a mixture of two exponential distributions. 
Things will get even more complicated if markets $A$ and $B$ are not entirely independent, but are connected by a small number of agents. It is still possible that some version of Theorem~\ref{Thm_main_happiness_dist} holds, but it will need to depend on the eigenspaces corresponding to large eigenvectors of the matrix. 

It is worth noting that even \eqref{eq:int2} fails to hold when we do not have the boundedness assumption. Consider a market consisting of $n/2$ small markets with just $2$ men and $2$ women in each. Under the uniform preferences within each market, the expected number of stable matchings is $9/8$, thus 
\begin{equation*} 
    \sum_{\mu} P_\mu = (9/8)^{n/2} \neq \exp(o(n)).
\end{equation*}

\subsection{Proof sketch}

\paragraph{Case with uniformly random preferences.} 
Let us first look at the classic case where agents' preferences are generated independently and uniformly at random, i.e., 
all canonical scores equal  $1/n$. %
The proof in this case is more straightforward due to the symmetry among agents and the established high probability bound on the number of stable matchings \citep{pittel1989average}. We keep however the  discussion informal in this section.

For a given matching $\mu$, we study the conditional distribution of the value vectors $\vX(\mu),\vY(\mu)\in\R^n$ conditional on stability of $\mu$. Conditional also on women's value vector $\vY(\mu)=\vy$, man $\m_i$'s value $X_i(\mu)$ must satisfy $X_i(\mu) < X_{ij}$ for each $j\ne \mu(i)$ with $Y_{ji} < y_j$ in order not to form a blocking pair. Since $X_{ij}$ and $Y_{ji}$ are i.i.d. samples from $\Exp(1)$ for each $j\ne \mu(i)$, one should expect $X_i(\mu)$ to be effectively less than the minimum of about $\sum_{j\ne \mu(i)} 1 - e^{-y_j} \approx \|\vy\|_1$ number of $\Exp(1)$ random variables. Such a constraint acts independently on each $X_i$ (conditional on $\vY(\mu)=\vy$), and therefore in the posterior distribution one should expect $\vX(\mu)$ to behave like i.i.d. samples from $\Exp(\|\vy\|_1)$.

Concretely, conditional on $\vU(\mu)=\vu=1-F(\vx)$ and $\vV(\mu)=\vv=1-F(\vy)$, where we recall that $F(z) = 1-e^{-z}$ is the CDF of $\Exp(1)$ and is applied component-wise on the value vectors. The likelihood of $(\vx,\vy)$ when $\mu$ is stable is $\prod_{j\ne \mu(i)} (1-u_iv_j)$ \citep{pittel1989average}. With some crude first order approximation (namely, $1-z\approx e^{-z}$), this expression can be approximated by
\begin{equation}\label{Eqn_approx_intuition}
    \prod_{j\ne \mu(i)} (1-u_iv_j)\approx \exp\left(-\sum_{i,j\in[n]}x_i y_j\right)=\exp(-\|\vx\|_1\|\vy\|_1),
\end{equation}
where we also put in the terms $x_i y_{\mu(i)}$ for $i\in [n]$ despite their absence in the original product. By the Bayes rule, conditional on $\vY(\mu)=\vy$ and that $\mu$ is stable, the distribution of $\vX(\mu)$ is approximately $p(\vx|\mu\in\mathcal{S}, \vY=\vy)\propto \exp(-\|\vx\|_1\|\vy\|_1)\cdot\prod_{i=1}^n e^{-x_i} = \prod_{i=1}^n \exp(-(1+\|\vy\|_1)x_i)$. Note that this is the joint density of the $n$-fold product of $\Exp(1+\|\vy\|_1)$. Our main theorems in the case with uniformly random preferences follow directly from the the convergence of empirical measures.

\paragraph{The general case.}

The entire result can be viewed abstractly from the following lens: For any matching $\mu$, we expect the value vectors $(\vX(\mu),\vY(\mu))$ to behave ``nicely'' with very high probability conditional on stability of $\mu$, so that even if we apply union bound on all stable matchings (which we will show to be ``rare'' separately) it is still unlikely to see any ``bad'' value vectors. To do so requires a careful analysis on the conditional distribution of value (given stability) $\mathcal{D}_\mu\in\Delta(\R^n_+\times\R^n_+)$, which depends on both the (unconditional) preference distribution (the ``prior'') and the conditional probability that $\mu$ is stable given a pair of value vector $\vX(\mu)=\vx$ and $\vY(\mu)=\vy$, which we will denote by $p_\mu(\vx,\vy)$. We will define the ultimate ``nice'' event to be the $\eps$-proximity of empirical distribution of value (or rescaled rank) to some exponential distribution, but unsurprisingly it is hard to analyze this event directly from $\mathcal{D}_\mu$, which is complicated itself, in one single step. Instead, we will follow a ``layer-by-layer peeling'' of the desirable events. Namely, we will find a nested sequence of subsets $\R^n_+\times \R^n_+ = \Omega_0 \supseteq \Omega_1 \supseteq \cdots\supseteq \Omega_K$ representing events on the joint value vectors of a stable matching, with $\Omega_K$ the desired event that the empirical distribution of value for men is $\eps$-close to some exponential distribution. Step by step, we will show that a stable matching, conditional on its value vectors in $\Omega_i$, must have value vectors in $\Omega_{i+1}$ with very high probability. Here is the roadmap to establishing these increasingly ``nice'' events:
\begin{enumerate}[listparindent=1cm, label=(\alph*)]   \item\label{Step_1_in_proof_sketch} As a first step, we approximate $p_\mu(\vx,\vy)$, the likelihood of value vectors $\vx$ and $\vy$ in a stable matching $\mu$, by the function $q(\vx,\vy)=\exp(-n\vx^\top\vM\vy)$. That is, the log likelihood of value vectors in a stable matching is approximately bilinear in the two vectors. To establish this, we identify a weak regularity condition on the value vectors of all stable matchings in terms of the first and second moments and extremal quantiles of the value vectors, under which the approximation holds (see Section~\ref{subsec_reg_of_happiness_moment_conditions}). Such a condition is met by all stable matchings with high probability (see Appendix~\ref{Append_weak_regular_scores} for details). The proof primarily consists of standard analysis of the deferred acceptance algorithm and careful use of first- and second-order approximation of $p_\mu(\vx,\vy)$. Here we use the fact that the men-proposing and women-proposing deferred acceptance algorithms output the extremal outcomes with respect to the two sides' values among all possible stable matchings.
    \item\label{Step_2_in_proof_sketch} In the expression for $\qxy$, the value vectors relate through the matching matrix $\vM$. However, we show next that, in stable matchings, we can further simplify things by approximately factoring $n\vX(\mu)^\top\vM\vY(\mu)$ into a product $\|\vX(\mu)\|_1\|\vY(\mu)\|_1$ of sums of values on the two sides. More specifically, both $\vM\vY(\mu)$ and $\vM^\top\vX(\mu)$ lie near the maximal eigenspace of $\vM$, which is the span of $\mathbf{1}$ under Assumption~\ref{Assumption_C_bounded} (see Section~\ref{Subsec_proximity_eigensubsp}). The proof uses a fixed point argument to deduce that $\vM\vY(\mu)$ depends almost deterministically on $\vM^\top\vX(\mu)$ and, symmetrically, $\vM^\top\vX(\mu)$ on $\vM\vY(\mu)$, which forces both quantities to lie near the eigenspace.
    Along the way, we also deduce an upper bound for the (unconditional) probability for $\mu$ to be stable (see Section~\ref{Subsec_uncond_stable_prob}), suggesting an sub-exponential upper bound on the typical number of stable matchings (Corollary~\ref{Cor_subexp_num_stable_match}).
    \item\label{Step_3_in_proof_sketch} Under the previous event, the men's values behaves approximately like i.i.d. exponential samples with rate $\|\vY(\mu)\|_1$ conditional on stability of $\mu$ and $\vY(\mu)$ -- in fact, they are conditionally independent and nearly identically distributed. The result on the empirical distribution of men's values
    follows immediately from a concentration inequality of Dvoretzky–Kiefer–Wolfowitz (DKW) type, generalized for nearly identically distributed independent random variables (Lemma~\ref{Lemma_dkw_non_identical}).%
    \item Finally, we translate values into ranks. Using the classic first- and second-moment method, we show that for majority of the agents in the market, the rescaled rank (based one's own scores) lies close to the value. This implies Theorem~\ref{Thm_main_rank_dist}.
\end{enumerate}

There is one caveat, however: In \ref{Step_1_in_proof_sketch}, a second order expansion of $p_\mu(\vx,\vy)$ is required in order to justify the approximation with $q(\vx,\vy)$. As a result, we need to control second order behavior of value, i.e., $\|\vX(\mu)\|_2^2=\sumiton X_i(\mu)^2$, in any stable matching $\mu$. However, the second moment cannot be easily controlled due to the heavy tail of $\Exp(1)$ (indeed, the moment generating function for $X^2$ does not exist for $X\sim\Exp(1)$). To resolve this issue, we perform a truncation in the upper $\delta/2$-quantile of the values on each side. By choosing $\delta$ sufficiently small, we can ensure that the truncation only affects the empirical distribution by an arbitrarily small amount in $\ell^\infty$ norm. As a price to pay, in \ref{Step_2_in_proof_sketch} and \ref{Step_3_in_proof_sketch}, we will have to deal with not just all stable matchings, but all \emph{partial} matchings on any $(1-\delta)$-fraction of the market that is stable. See Section~\ref{Subsec_partial_match_and_truncation} for the technical definition of truncated and partial matchings.

\section{Preliminaries} 

\subsection{Probability of stability and its approximation}

For each matching $\mu$, define the function $p_\mu: \R^n_+\times\R^n_+ \to [0,1]$ to be probability that $\mu$ is stable given values of men and women in $\mu$. That is 
\begin{equation}\label{Eqn_def_pmu}
    p_\mu(\vx,\vy) = \Pp(\mu\in\mathcal{S} | \vX(\mu)=\vx,\vY(\mu)=\vy).
\end{equation}
Just like the integral formula used in \citet{knuth1976mariages} and \citet{pittel1989average,pittel1992likely} to study matching markets with uniformly random preferences, the probability of a matching $\mu$ being stable can be similarly characterized by an integral
\begin{multline}\label{Eqn_integral_formula_orig}
   \Pp(\mu\in\mathcal{S}) = \E_{\vX\sim\bigotimes_{i=1}^n\Exp(a_{i,\mu(i)}),\vY\sim\prod_{i=1}^n\Exp(b_{i,\mu^{-1}(i)})}[p_\mu(\vX,\vY)] \\
    = \int_{\R_+^n\times \R_+^n} p_\mu(\vx,\vy) \prod_{i=1}^n f_{a_{i,\mu(i)}}(x_i)f_{b_{i,\mu^{-1}(i)}}(y_i) \dd\vx \dd\vy.
\end{multline}

The function $p_\mu$ can be further expressed in closed form. Condition on the value vector $\vX(\mu)
= \vx$ and $\vY(\mu)
= \vy$ and sample the rest of the values $X_{ij}$ and $Y_{ji}$ for all $j\ne \mu(i)$. Each pair of $i,j\in[n]$ with $j\ne \mu(i)$ may form a blocking pair when $X_{ij} < x_i$ and $Y_{ji} < y_j$, which event happens with probability $(1-\exp(-a_{ij}x_i))(1-\exp(-b_{ji}y_j))$. For $\mu$ to be stable, there must be no blocking pairs and thus
\begin{equation}
    p_\mu(\vx,\vy) = \prod_{\substack{i,j\in[n]\\j\ne \mu(i)}} \left(1 -  \big(1-e^{-a_{ij}x_i}\big)\big(1-e^{-b_{ji}y_j}\big) \right).
\end{equation}

Under Assumption~\ref{Assumption_C_bounded}, i.e., $\max_{i,j_1,j_2} a_{ij_1}/a_{ij_2} \le C^2$ and $\max_{j,i_1,i_2} b_{ji_1}/b_{ji_2} \le C^2$, we observe a simple upper bound
\begin{equation}\label{Eqn_naive_bound_pxy}
    p_\mu(\vx,\vy) \le \prod_{\mu(i)\ne j} \Big(1 - \big(1-e^{-\hat{x}_i/C^2}\big)\big(1-e^{-\hat{y}_j/C^2}\big) \Big),
\end{equation}
where $\hat{x}_i = x_i a_{i,\mu(i)}$ and $\hat{y}_j = y_j b_{j,\mu^{-1}(j)}$ for $i,j\in[n]$ are the renormalized values (thus named because they have unit mean).
This bound is fairly conservative and crude for our final purpose, but will prove useful for establishing preliminary results.

To further simplify the analysis, we recognize that, through first order approximation, 
\begin{equation}\label{Eqn_goal_approx_p_with_q}
    p_\mu(\vx,\vy) \approx \prod_{i\neq j}\left(1-a_{ij}b_{ji} x_i y_j\right) \le \exp\Big(-\sum_{i\ne j} a_{ij}b_{ji} x_i y_j\Big) \approx \exp(-n \vx^\top \vM \vy).
\end{equation}
Define the function
\begin{equation*}
    q(\vx,\vy) := \exp(-n \vx^\top \vM \vy) = \exp\bigg(-n \sum_{i,j=1}^n m_{ij}x_i y_j\bigg).
\end{equation*}
In the next section we  discuss conditions under which the function $q(\vx,\vy)$ offers a good  approximation for $p_\mu(\vx,\vy)$.

\subsection{Partial matchings and truncation} \label{Subsec_partial_match_and_truncation}

In order to study also approximately stable matchings, as well as for technical reasons, we need to consider matchings that are stable on a significant subset of the market. We first formalize a general partial matching  then describe a 
 particular way to form stable partial matchings. %

Let $\mathcal{M}'\subseteq \mathcal{M}$ and $\mathcal{W}'\subseteq \mathcal{W}$ be subsets of the men and women with cardinality $|\cM'|=|\cW'|= n'$. A {\em partial matching} $\mu':\cM'\to\cW'$ is a bijection between $\cM'$ and $\cW'$. Denote the values of men among $\cM'$ and women among $\cW'$ in the partial matching $\mu'$ by $\vX_{\cM'}(\mu')$ and $\vY_{\cW'}(\mu')$, respectively. While it may be natural to view $\vX_{\cM'}(\mu')$ and $\vY_{\cW'}(\mu')$ as $n'$-dimensional vector, we choose to view them as $n$-dimensional vectors where components corresponding to men in $\cM\backslash\cM'$ and women in $\cW\backslash\cW'$ are zero (recall that since small is better, zero is the best possible latent value). Therefore, conditional on $\vX_{\cM'}(\mu')=\vx'$ and $\vY_{\cW'}(\mu')=\vy'$ for $\vx',\vy'\in\R^n$ supported on $\cM'$ and $\cW'$, respectively, the probability that $\mu'$ is stable (as a matching between $\cM'$ and $\cW'$) is simply $p_{\mu'}(\vx',\vy')$.%

Given a full stable matching $\mu$ and any $\delta > 0$, we define the following routine to construct a stable partial matching of size $n-\floor{\delta n}$: Let $\bar{\cM}_{\mu,\delta/2}\subseteq \cM$ be the subset of $\floor{\delta n / 2}$ men with the largest value (i.e., the least happy men) in $\mu$, and similarly let $\bar{\cW}_{\mu,\delta/2}\subseteq \cW$ be the set of $\floor{\delta n / 2}$ least happy women. Construct $\cM'_{\mu,\delta}\subseteq \cM \backslash (\bar{\cM}_{\mu,\delta/2} \cup \mu(\bar{\cW}_{\mu,\delta/2}))$ of cardinality $n-\floor{\delta n}$. This is always possible because $|\bar{\cM}_{\mu,\delta/2} \cup \mu(\bar{\cW}_{\mu,\delta/2})| \le 2 \floor{\delta n / 2} \le \floor{\delta n}$ and in fact here can be multiple ways to choose $\cM'_{\mu,\delta}$. The specific way $\cM'_{\mu,\delta}$ is chosen (when $|\cM \backslash (\bar{\cM}_{\mu,\delta/2} \cup \mu(\bar{\cW}_{\mu,\delta/2}))| > n-\floor{\delta n / 2}$) is irrelevant to our discussion, but it may be helpful to assume that the choice is made based on some canonical ordering of the men so there is no extra randomness. Let $\mu_\delta:\cM'_{\mu,\delta}\to\mu(\cM'_{\mu,\delta})$ be the partial matching induced by $\mu$ on $\cM'_{\mu,\delta}$ and their partners. Define the \emph{$\delta$-truncate value} vector for $\mu$ to be $\vX_\delta(\mu):= \vX_{\cM'_{\mu,\delta}}(\mu_\delta)$ and $\vY_\delta(\mu):= \vY_{\mu(\cM'_{\mu,\delta})}(\mu_\delta)$.

\section{Regularity of values in stable matchings}\label{sec_reg_of_happiness}

In this section, we  establish several (high probability) properties of stable matchings.

\subsection{Moment behavior and approximation of the conditional stable probability}\label{subsec_reg_of_happiness_moment_conditions}

We first consider a set of events in the value space, which can be thought of as regularity conditions for the approximation \eqref{Eqn_goal_approx_p_with_q} of $p_\mu(\vx,\vy)$ by $q(\vx,\vy)$. Define
\begin{equation}\label{Eqn_def_underR1}
    \underline{\mathcal{R}}_1 = \{\vu\in\R_+^n : \|\vu\|_1 \ge \underline{c}_1 \log n\},
\end{equation}
\begin{equation}\label{Eqn_def_overR1}
    \overline{\mathcal{R}}_1 = \{\vu\in\R_+^n : \|\vu\|_1 \le \overline{c}_1 n (\log n)^{-7/8}\},
\end{equation}
\begin{equation}\label{Eqn_def_R2}
    \mathcal{R}_{2} = \{(\vu,\vv)\in\R_+^n\times\R_+^n : \vu^\top\vM\vv \le c_2 (\log n)^{1/8}\},
\end{equation}
where $\underline{c}_1,\overline{c}_1,c_2\in\R_+$ are constants to be specified later. 
Let
\begin{equation*}
    \mathcal{R}_1 = \underline{\mathcal{R}}_1\cap\overline{\mathcal{R}}_1 \enspace\text{ and }\enspace \mathcal{R} = (\mathcal{R}_1 \times \R_+^n) \cap (\R_+^n \times \mathcal{R}_1) \cap \mathcal{R}_2 = \{(\vx,\vy)\in\mathcal{R}_2: \vx,\vy\in \mathcal{R}_1 \}.
\end{equation*}
The region $\mathcal{R}$ should capture the typical behavior of $(\vX_\delta(\mu),\vY_\delta(\mu))$ for any stable matching $\mu$.
\begin{proposition}\label{Prop_R_likely}
For any fixed $\delta > 0$ and $c \in (0,1/2)$, the constants $\underline{c}_1,\overline{c}_1$, and $c_2$ in \eqref{Eqn_def_underR1}-\eqref{Eqn_def_R2} can be appropriately chosen such that
\begin{equation}
    \Pp(\exists \mu\in\mathcal{S}, (\vX_\delta(\mu),\vY_\delta(\mu))\notin\mathcal{R}) \lesssim e^{-n^c}
\end{equation}
asymptotically as $n\to\infty$.
\end{proposition}

\begin{remark}
    It is helpful to compare the bounds \eqref{Eqn_def_underR1}-\eqref{Eqn_def_R2} to classic results in the setting with uniform preferences (cf. \citep{pittel1989average,pittel1992likely}): Namely, the optimal average rank is $\Theta(\log n)$, the pessimal average rank is $\Theta(n/\log n)$, and the product of the average ranks on the two side is asymptotic to $n$ in all stable matchings. Here, due to heterogeneity of preferences, we pay a small price of an extra constant or $(\log n)^{1/8}$ factor.
\end{remark}

We will defer the proof to Appendix~\ref{Append_weak_regular_scores}. %
In fact, we will establish even finer control over the truncated value vectors in stable matchings. For a matching $\mu$, we define $\vU(\mu) = F(\hat{\vX}(\mu))$ and $\vV(\mu) = F(\hat{\vY}(\mu))$, where the (standard exponential CDF) function $F(z)=1-e^{-z}$ is applied coordinate-wise to the renormalized value vectors. Through relating $\vX$ and $\vY$ to $\vU$ and $\vV$, we will specify a subregion $\mathcal{R}^\star \subseteq \mathcal{R}$ in which $p_\mu(\vx,\vy)$ can be well approximated by $q(\vx,\vy)$.\footnote{Technically, $\mathcal{R}^\star$ has to be defined in the context of a matching $\mu$, as $\vU$ and $\vV$. Here we drop the dependency for convenience. See Corollary~\ref{Cor_Rstar_likely} for the formal definition of $\mathcal{R}^\star(\mu)$.} We will see in Corollary~\ref{Cor_Rstar_likely} that with high probability no stable matchings $\mu$ have $(\vX_\delta(\mu),\vY_\delta(\mu))$ outside $\mathcal{R}^\star$, from which Proposition~\ref{Prop_R_likely} follows.

The conditions for $\mathcal{R}^\star$ are sufficiently strong to bound the functions $p_\mu$ and $q$ within an $\exp(o(n))$ factor of each other. This is formalized as follows.

\begin{restatable}{proposition}{propRatioPQHighProbeon}\label{Prop_ratio_p_q_high_prob}
For any $\delta>0$ and $c \in (0,1/2)$, there exists an absolute constant $\theta\in(0,\infty)$ such that the probability that a matching $\mu$ is stable with value vectors \emph{not} satisfying
\begin{equation}\label{Eqn_prop_ratio_p_q_high_tag}
    \frac{p_\mu(\vX_\delta(\mu),\vY_\delta(\mu))}{q(\vX_\delta(\mu),\vY_\delta(\mu))} \le \exp\left(\frac{\theta n}{(\log n)^{1/2}}\right)\tag{$\star$}
\end{equation}
is $\frac{\exp(-n^c)}{n!}$. 
In other words, with high probability, there exist no stable matchings $\mu$ whose post-truncation value vectors $\vX_\delta(\mu)$ and $\vY_\delta(\mu)$ satisfy \eqref{Eqn_prop_ratio_p_q_high_tag}.
\end{restatable}

Again, the proof of Proposition~\ref{Prop_ratio_p_q_high_prob} 
is deferred to Appendix~\ref{Append_weak_regular_scores}.

The reason for using $\delta$-truncated value vectors is that, when  approximating $p(\vx,\vy)$ to second order, there will be terms of $\|\vx\|_2^2$ and $\|\vy\|_2^2$, which are hard to control due to the heavy tail of the exponential distribution.%
\footnote{Note that  moment generating function does not exist for $X^2$ where $X\sim\Exp(1)$ so the classic Hoeffding- or Bernstein-type bounds fail to apply.} On the other hand, changing the values of a $\delta$ fraction should  affect the empirical CDF by at most $\delta$ in $\ell^\infty$ distance. Therefore, %
it suffices to show that  for small enough $\delta$ all stable partial matchings of size $n-\floor{\delta n}$ have values and ranks empirically distributed close to some exponential distribution.

The function $p_\mu(\vx,\vy)$
cannot be  approximated globally by  $q(\vx,\vy) = \exp(-n\vx^\top \vM\vy)$. 
However, we can find a region in $\R_+^n\times\R_+^n$ where $p_\mu(\vx,\vy)$ and $q(\vx,\vy)$ are close (uniformly for all stable matchings) in the sense that \eqref{Eqn_prop_ratio_p_q_high_tag} holds; Meanwhile, Proposition~\ref{Prop_ratio_p_q_high_prob} states that with high probability, no stable matchings will ever have $\delta$-truncated value vectors outside this region.

\subsection{A key reduction lemma}

Proposition~\ref{Prop_ratio_p_q_high_prob} allows us to study high probability behaviors in stable partial matchings obtained from truncating stable (full) matchings pretending that the conditional probability of stability were given by $q(\vx,\vy)$. Concretely, consider a fixed constant $\delta > 0$ and a region $\Omega \subseteq \R_+^n \times \R_+^n$ that defines an event on the (truncated) value vectors. If there exists a stable matching $\mu$ whose truncated value vectors $(\vX_\delta(\mu),\vY_\delta(\mu))\in\Omega$, the induced partial matching $\mu_\delta$ of size $n-\floor{\delta n}$ between $\cM'_{\mu,\delta}$ and $\mu(\cM'_{\mu,\delta})$ must also be stable with value vectors $\vX_{\cM'_{\mu,\delta}}(\mu_\delta) = \vX_\delta(\mu)$ and $\vY_{\mu(\cM'_{\mu,\delta})}(\mu_\delta) = \vY_\delta(\mu)$. Thus, we will end up either having a stable matching whose truncated value violates \eqref{Eqn_prop_ratio_p_q_high_tag}, or a stable partial matching of size $n-\floor{\delta n}$ whose value vectors (already truncated) lies in $\Omega\cap\mathcal{R}^\star$. By Proposition~\ref{Prop_ratio_p_q_high_prob}, the former event happens with probability $o(1)$. Therefore, we may focus on the second event, where a stable partial matching of size $n-\floor{\delta n}$ exists with value vectors in $\Omega\cap\mathcal{R}^\star$. This is summarized by the following Lemma, which will be a major tool in the remainder of the proof.%

\begin{lemma}\label{Lemma_reduction_to_q}
Let $\delta > 0$, $c \in (0,1/2)$, and $\Omega \subseteq \R_+^n \times \R_+^n$. Then,
\begin{multline}
    \Pp(\exists \mu\text{ stable}, (\vX_\delta(\mu),\vY_\delta(\mu))\in\Omega) \le e^{-n^c} + \exp\left(\Theta\Big(\frac{n}{(\log n)^{1/2}}\Big)\right) \cdot \\
    \sum_{\substack{\cM'\subseteq\cM,\cW'\subseteq\cW\\|\cM'|=|\cW'|=n-\floor{\delta n}}}\sum_{\substack{\mu':\cM'\to\cW'\\\text{bijection}}} \E\big[q(\vX_{\cM'}(\mu'), \vY_{\cW'}(\mu')) \cdot \mathbbm{1}_{\mathcal{R}\cap\Omega}(\vX_{\cM'}(\mu'), \vY_{\cW'}(\mu'))\big].
\end{multline}
\end{lemma}
\begin{proof}
Note that
\begin{multline}\label{Eqn_proof_reduction_to_q_1}
    \Pp(\exists \mu\text{ stable}, (\vX_\delta(\mu),\vY_\delta(\mu))\in\Omega) \le \Pp(\exists \mu\text{ stable}, (\vX_\delta(\mu),\vY_\delta(\mu))\notin\mathcal{R}^\star) \\
    + \Pp(\exists \mu\text{ stable}, (\vX_\delta(\mu),\vY_\delta(\mu))\in\mathcal{R}^\star\cap\Omega),
\end{multline}
where the first term is $e^{-n^c}$ by Corollary~\ref{Cor_Rstar_likely}.
Let $\mathcal{E}$ denote the event that there exists a {\em stable} partial matching $\mu'$ between $\cM'\subseteq\cM$ and $\cW'\subseteq\cW$ with $|\cM'|=|\cW'|=n-\floor{\delta n}$ where $(\vX_{\cM'}(\mu'), \vY_{\cW'}(\mu'))\in\mathcal{R}^\star\cap\Omega$. Clearly, the existence of a stable matching $\mu$ with $(\vX_\delta(\mu),\vY_\delta(\mu))\in\mathcal{R}^\star\cap\Omega$ implies $\mathcal{E}$.
Thus, the second term in \eqref{Eqn_proof_reduction_to_q_1} is bounded by
\begin{equation}
    \Pp(\mathcal{E})
    \le \sum_{\substack{\cM'\subseteq\cM,\cW'\subseteq\cW\\|\cM'|=|\cW'|=n-\floor{\delta n}}}\sum_{\substack{\mu':\cM'\to\cW'\\\text{bijection}}} \Pp(\mu'\text{ stable}, (\vX_{\cM'}(\mu'), \vY_{\cW'}(\mu'))\in\mathcal{R}^\star\cap\Omega).
\end{equation}
using union bound. For each $\cM',\cW'$, and $\mu'$ in the summation, we compute the above probability through conditioning on $\vX_{\cM'}(\mu')$ and $\vY_{\cW'}(\mu')$ as
\begin{align}
    \Pp(\mu'\text{ stable}&, (\vX_{\cM'}(\mu'), \vY_{\cW'}(\mu'))\in\mathcal{R}^\star\cap\Omega) \nonumber\\
    &= \E\big[\Pp\big(\mu'\text{ stable} \big| \vX_{\cM'}(\mu'), \vY_{\cW'}(\mu')\big); (\vX_{\cM'}(\mu'), \vY_{\cW'}(\mu'))\in\mathcal{R}^\star\cap\Omega\big] \nonumber\\
    &= \E\big[p_{\mu'}(\vX_{\cM'}(\mu'), \vY_{\cW'}(\mu')) \cdot\mathbbm{1}_{\mathcal{R}^\star\cap\Omega}(\vX_{\cM'}(\mu'), \vY_{\cW'}(\mu'))\big] \nonumber\\
    &\le \E\left[\exp\left(\Theta\Big(\frac{n}{(\log n)^{1/2}}\Big)\right)q(\vX_{\cM'}(\mu'), \vY_{\cW'}(\mu')) \cdot\mathbbm{1}_{\mathcal{R}^\star\cap\Omega}(\vX_{\cM'}(\mu'), \vY_{\cW'}(\mu'))\right] \nonumber\\
    &\le \exp\left(\Theta\Big(\frac{n}{(\log n)^{1/2}}\Big)\right) \E\left[q(\vX_{\cM'}(\mu'), \vY_{\cW'}(\mu')) \cdot\mathbbm{1}_{\mathcal{R}\cap\Omega}(\vX_{\cM'}(\mu'), \vY_{\cW'}(\mu'))\right],
\end{align}
which completes the proof.
\end{proof}
\begin{remark}
    The choice of the constant $c\in(0,1/2)$ affects the implicit constants in defining $\mathcal{R}$ and $\mathcal{R}^\star$. Once we have a target convergence of $e^{-n^c}$ for some $c\in(0,1/2)$, we will assume that $c$ is fixed in the rest of our discussion unless otherwise mentioned.
\end{remark}

Lemma~\ref{Lemma_reduction_to_q} will be a key tool in the proof to establish further likely behaviors of value (and rank) vectors. It will be a  recurring theme where we first identify a likely region $\Omega_\text{likely}$ for truncated value vectors of stable (full) matchings to fall in ($\mathcal{R}$ to start with), then rule out a bad event $\Omega_\text{bad}$ within $\Omega_\text{likely}$ by showing $\Omega_\text{likely}\cap\Omega_\text{bad}$ is unlikely for value vectors of any stable \emph{partial} matching, and apply Lemma~\ref{Lemma_reduction_to_q} to conclude that $\Omega_\text{bad}$ is unlikely for truncated value vectors of any stable matching and that $\Omega_\text{likely}\cap\Omega_\text{bad}$ can be used as the likely region moving forward.

Based on Lemma~\ref{Lemma_reduction_to_q}, it now suffices to consider partial matchings $\mu'$ of size $n-\floor{\delta n}$ between $\cM'$ and $\cW'$ and upper bound $\E\big[q(\vX_{\cM'}(\mu'), \vX_{\cW'}(\mu')) \cdot \mathbbm{1}_\Omega(\vX_{\cM'}(\mu'), \vX_{\cW'}(\mu'))\big]$.
From now on, we fix %
$\cM'$, $\cW'$, and $\mu'$, and make the dependency of $\vX_{\cM'}$ and $\vY_{\cW'}$ on $\mu'$ implicit when the context is clear.

\subsection{Estimating the (unconditional) stability probability}\label{Subsec_uncond_stable_prob}

Using concentration inequalities given in Lemma~\ref{Lemma_weighted_exp_chernoff} and \ref{Lemma_wgt_exp_cond_concentration}, we  derive the following upper bound, which essentially characterizes the (approximate) probability that a partial matching of size $n-\floor{\delta n}$ is stable with a probable value vector for the women (i.e., $Y_{\cW'}\in\mathcal{R}_1$).

\begin{restatable}{proposition}{propEqXYBound}\label{Prop_EqXY_bound}
For a fixed a partial matching $\mu'$ on $\cM'$ and $\cW'$ of size $n-\floor{\delta n}$,
\begin{equation}\label{Eqn_prop_EqXY_target_bound}
    \E[q(\vX_{\cM'},\vY_{\cW'}) \cdot \mathbbm{1}_{\mathcal{R}_2}(\vX_{\cM'}, \vY_{\cW'}) \cdot\mathbbm{1}_{\mathcal{R}_1}(\vY_{\cW'})] \le e^{o(n)+o_\delta(n)} \frac{(\delta n)!}{n!}\prod_{i\in\cM'} a_{i,\mu'(i)} b_{\mu'(i),i}.
\end{equation}
\end{restatable}

The proof of Proposition~\ref{Prop_EqXY_bound} will be deferred to Appendix~\ref{appendix_extra_proofs}, where we will develop intermediate results that characterize the typical behavior of $\vX_{\cM'}$ and $\vY_{\cW'}$ relative to each other (see Appendix~\ref{Append_proof_prop_eigenvec_of_M}).

Proposition~\ref{Prop_EqXY_bound} provides evidence that, heuristically, the expected number of stable partial matchings should be sub-exponential.
\begin{restatable}[Number of stable partial matchings]{corollary}{corSubExpNumOfStableMatch}
\label{Cor_subexp_num_stable_match}
Fix any $\delta > 0$ and $c\in(0,1/2)$. Let $N_\delta$ denote the number of stable partial matchings of size $n-\floor{\delta n}$ satisfying the condition in Corollary~\ref{Cor_Rstar_likely} (i.e., $\mathcal{R}^\star$) in a random instance of the matching market. Then, $\E[N_\delta] \le \exp(o_\delta(n))$ granted that $n$ is sufficiently large. Further, with probability at least $1-e^{-n^c}$, the condition $\mathcal{R}^\star$ is satisfied by all $\delta$-truncated stable matchings.
\end{restatable}

\begin{remark}
Corollary~\ref{Cor_subexp_num_stable_match} falls short of establishing a sub-exponential bound for the expected number of stable matchings in two aspects.
\begin{itemize}
    \item While stable matchings that violate $\mathcal{R}^\star$ (when truncated) will not exist with high probability, we have not yet proved a bound for the expected number of such stable matchings. We believe that this can be overcome with a refined analysis of deferred acceptance, which should lead to stronger results than Lemma~\ref{Lemma_DA_prop_num}. Note that all high probability results in Appendix~\ref{Append_weak_regular_scores} after this lemma come with an upper bound on the expected number of stable matchings under various conditions.
    \item In general, it is possible to have multiple, in the worst case $\floor{\delta n}!$, stable matchings that produces the same $\delta$-truncated stable partial matching.
\end{itemize}
We believe that a sub-exponential bound for the number of stable matchings is possible with a more refined  analysis.
\end{remark}

\subsection{Opportunity sets and an eigenspace property for the value vectors}\label{Subsec_proximity_eigensubsp}

Our next result states that value vectors in stable matchings are not only controlled in terms of their first and second moments, but also in a sense ``close'' to some constant vector, i.e., $t\mathbf{1}$ for some $t\in\R_+$, which are eigenvectors of $\vM$ corresponding to its maximal eigenvalue $\lambda_1(\vM)=1$.

Let us fix the women's values to be $\vY_{\cW'}=\vy$ and consider the implication for the men's outcome in any (partial) matching $\mu'$. For a man $\m_i$ with value $x_i$, the expected number of blocking pairs between him and the women, conditional on $x_i$ and $\vy$, is
\begin{equation*}
    \sum_{j\ne \mu(i)} (1-e^{-a_{ij} x_i}) (1-e^{-b_{ji} y_j}) \approx \sum_{j=1}^n a_{ij} b_{ji} x_i y_j = n (\vM \vy)_i x_i.
\end{equation*}
The next result suggests that, in a typical market, the burden of avoiding blocking pairs falls roughly equally on the men in the sense that the entries of $\vM \vY_{\cW'}$ are largely the same. 

\begin{restatable}{lemma}{PropEigenVecOfMHighProf}\label{Prop_eigenvec_of_M_high_prob}
Let $\mu'$ be a partial matching of size $n-\floor{\delta n}$ on $\cM'\subseteq\cM$ and $\cW'\subseteq\cW$. Fix any $\zeta > 0$, and let
\begin{equation}
    \Oeigz := \left\{(\vx,\vy)\in \R^n\times\R^n : \exists t\in\R_+, \sum_{i=1}^n \mathbbm{1}\left\{|(\vM \vy)_i-t| \ge \sqrt{\zeta} t \right\} \le \sqrt{\zeta} n\right\}.
\end{equation}

Then
\begin{equation}\label{Eqn_Prop_eigenvec_of_M_Expectation_is_small}
    \E\big[q(\vX_{\cM'}, \vY_{\cW'}) \cdot \mathbbm{1}_{\mathcal{R}\backslash\Oeig(\Theta(\delta)+\zeta)}(\vX_{\cM'}, \vY_{\cW'})\big]
    \le %
    \exp(o_\delta(n)-\Theta(\zeta^2 n)) \cdot \frac{(\delta n)!}{n!} \prod_{i\in\cM'}a_{i,\mu'(i)}b_{\mu'(i),i},
\end{equation}
with the implicit constants uniform over all $\cM',\cW'$, and $\mu'$.
\end{restatable}

The proof of Lemma~\ref{Prop_eigenvec_of_M_high_prob} is deferred to Appendix~\ref{Append_proof_prop_eigenvec_of_M}.

Let us observe the immediate corollary of this Lemma, the proof of which is similar to that of Lemma~\ref{Lemma_reduction_to_q} and deferred to Appendix~\ref{Append_proof_no_stable_outside_oeigz}.

\begin{restatable}{corollary}{CorNoStableOutsideOeigz}\label{Cor_no_stable_outside_Oeigz}
For $\delta > 0$ sufficiently small, there exists a choice of $\zeta =\zeta(\delta) > 0$ such that $\zeta\to 0$ as $\delta \to 0$ and that
\begin{equation}
    \Pp(\exists \mu\text{ stable}, (\vX_\delta(\mu),\vY_\delta(\mu))\notin\Oeigz) \lesssim e^{-n^c}
\end{equation}
asymptotically as $n\to\infty$.
\end{restatable}

Corollary~\ref{Cor_no_stable_outside_Oeigz} roughly states that, in a contiguous market, conditioning on the women's outcomes in a stable matching has an almost even impact on the men's values.

\section{Empirical distribution of values and ranks}\label{sec_distribution}

\subsection{Empirical distribution of values}

Knowing the eigenspace property of the value vectors allows us to characterize the empirical distribution of values.

\begin{restatable}{lemma}{propOempeLikelyForHappinessEmpDist}\label{Prop_Oempe_likely_for_happiness_emp_distr}
Fix any $\delta,\zeta > 0$. Let $\mu'$ be a partial matching of size $n-\floor{\delta n}$ on $\cM'\subseteq\cM$ and $\cW'\subseteq\cW$. For any $\eps > 0$, consider
\begin{equation}\label{Eqn_Prop_Oempe_likely_for_happiness_oempe_def}
    \Oempe := \left\{(\vx,\vy)\in \R_+^n\times\R_+^n : \exists \lambda\in\R_+, \big\|\Femp(\vx) - F_\lambda\|_\infty \le \eps + \Theta(\delta + \sqrt{\zeta})\right\}.
\end{equation}
Then
\begin{equation}
    \E\big[q(\vX_{\cM'}, \vY_{\cW'}) \cdot \mathbbm{1}_{\mathcal{R} \cap \Oeigz\backslash\Oempe}(\vX_{\cM'}, \vY_{\cW'})\big] \le \exp(o_\delta(n)-\Theta(\eps^2 n)) \cdot \frac{(\delta n)!}{n!} \prod_{i\in\cM'}a_{i,\mu'(i)}b_{\mu'(i),i},
\end{equation}
where again the implicit constants are uniform over all $\cM',\cW'$, and $\mu'$.
\end{restatable}

The proof formalizes the intuition that, conditional on stability of $\mu'$ and $\vY_{\cW'}=\vy$, the value $X_i$ for $i\in\cM'$ should behave approximately as $\Exp(\lambda_i)$ for some $\lambda_i = (1+ \Theta(\delta+\sqrt{\zeta})) \|\vy\|_1$.
The full proof is deferred to Appendix~\ref{Proof_prop_Oempe_likely}.
Hence, instead of looking for the optimal $\lambda$ that minimizes $\|\hat{\mathcal{F}}(\vX_{\cM'})-\mathcal{F}(\Exp(\lambda))\|_\infty$ in the definition \eqref{Eqn_Prop_Oempe_likely_for_happiness_oempe_def} of $\Oempe$, we may simply choose $\lambda = \|\vy\|_1$, which only differs from the right choice by at most a tolerable $\Theta(\sqrt{\zeta}+\delta)$ factor. In other words, if we define
\begin{equation*}
    \tOemp(\eps) := \left\{(\vx,\vy)\in \R_+^n\times\R_+^n : \big\|\Femp(\vx) - \mathcal{F}(\Exp(\|\vy\|_1))\|_\infty \le \eps + \Theta(\delta + \sqrt{\zeta})\right\},
\end{equation*}
albeit with a worse implicit constant in $\Theta(\delta+\sqrt{\zeta})$,
the same conclusion holds as in Lemma~\ref{Prop_Oempe_likely_for_happiness_emp_distr} with $\Oempe$ replaced by $\tOemp(\eps)$; that is,
\begin{equation}\label{Eqn_tilde_Oempe_likely_for_happi_emp_distr}
    \E\big[q(\vX_{\cM'}, \vY_{\cW'}) \cdot \mathbbm{1}_{\mathcal{R} \cap \Oeigz\backslash\tOemp(\eps)}(\vX_{\cM'}, \vY_{\cW'})\big] \le \exp(o_\delta(n)-\Theta(\eps^2 n)) \cdot \frac{(\delta n)!}{n!} \prod_{i\in\cM'}a_{i,\mu'(i)}b_{\mu'(i),i}.
\end{equation}

Using Lemma~\ref{Prop_Oempe_likely_for_happiness_emp_distr}, we now prove our first main theorem about the uniform limit of empirical distribution of men's (or women's) value in stable matchings.

\begin{theorem}[Empirical distribution of value]\label{Thm_main_happiness_dist_body}
Fix any $\eps>0$. Then
\begin{equation}\label{Eqn_happiness_dist_main_body_thm_whp}
    \Pp\bigg(\max_{\mu\in\mathcal{S}} \inf_{\lambda\in\R_+} \|\Femp(\vX(\mu))-F_\lambda\|_\infty > \eps\bigg) \lesssim e^{-n^c}
\end{equation}
asymptotically as $n\to\infty$.
In particular, the infimum over $\lambda$ in \eqref{Eqn_happiness_dist_main_body_thm_whp} can be replaced with the choice of $\lambda=\|\vY_\delta(\mu)\|_1$ for $\delta$ sufficiently small.
\end{theorem}
\begin{proof}
    Plugging \eqref{Eqn_tilde_Oempe_likely_for_happi_emp_distr} into Lemma~\ref{Lemma_reduction_to_q} and repeating the same arithmetic as in \eqref{Eqn_sum_expectation_Omega_zeta} and \eqref{Eqn_sum_expectation_Omega_zeta_summation_bound} immediately give
    \begin{equation}
        \Pp(\exists \mu\in\mathcal{S}, (\vX_\delta(\mu),\vY_\delta(\mu))\in \Oeigz\backslash\tilde{\Omega}_{\text{emp}}(\eps)) \le e^{-n^c} + \exp(o_\delta(n) - \Theta(\eps^2 n)) \lesssim e^{-n^c},
    \end{equation}
    granted that $\eps \ge \eps_0(\delta)$, where the function $\eps_0(\delta)\to 0$ as $\delta\to 0$.
    Corollary~\ref{Cor_no_stable_outside_Oeigz} implies that with probability at least $1-\Theta(e^{-n^c})$ there exists no stable matching $\mu$ with $(\vX_\delta(\mu),\vY_\delta(\mu))\notin \Oeigz$, and hence
    \begin{equation}\label{Eqn_proof_Thm_main_happiness_dist_tOemp}
        \Pp(\exists \mu\in\mathcal{S}, (\vX_\delta(\mu),\vY_\delta(\mu))\notin \tOemp(\eps/3)) \lesssim e^{-n^c},
    \end{equation}
    granted that $\eps_0(\delta) < \eps / 3$.
    By choosing $\delta$ (and hence also $\zeta=\zeta(\delta)$) sufficiently small so that the $\Theta(\delta+\sqrt{\zeta})$ term in the definition of $\tOemp$ is upper bounded by $\eps/3$, we ensure
    \begin{equation}
        \Pp(\exists \mu\in\mathcal{S}, \big\|\Femp(\vX_\delta(\mu)) - \mathcal{F}(\Exp(\|\vY_\delta(\mu)\|_1))\|_\infty \ge 2\eps/3) \lesssim e^{-n^c}.
    \end{equation}
    By further restricting $\delta$ to be sufficiently small, we may absorb the difference caused by the $\delta$-truncation on $\vX(\mu)$ into an extra term of $\Theta(\delta)\le \eps/3$, since $\|\Femp(\vX_\delta(\mu))-\Femp(\vX(\mu))\|_\infty \le \delta$. The theorem follows immediately.
\end{proof}

With essentially the same analysis as in Lemma~\ref{Prop_Oempe_likely_for_happiness_emp_distr} and Theorem~\ref{Thm_main_happiness_dist_body}, except for replacing the DKW inequality with Bernstein's inequality for empirical average, we can also deduce the fillowing result. The proof is omitted.

\begin{proposition}
For any fixed $\eps,\delta > 0$ and $0< c < 1/2$,
\begin{equation}
    \Pp\bigg(\max_{\mu\in\mathcal{S}} |n^{-1}\|\vX_\delta(\mu)\|_1 \|\vY_\delta(\mu)\|_1 - 1| > \eps\bigg) \lesssim e^{-n^c}.
\end{equation}
\end{proposition}

The effect of the $\delta$-truncation is nontrivial to remove because the sum of values can be sensitive to outliers, in particular given the heavy tail of the exponential distribution. We believe, however, that a refined analysis should suggest that $\sup_{\mu\in\mathcal{S}}\big|n^{-1}\|\vX(\mu)\|_1\|\vY(\mu)\|_1 - 1\big| \overset{p}{\to} 0$. This is the analogue of the ``law of hyperbola'' in \citet{pittel1992likely}.

\subsection{Empirical distribution of ranks}

Based on the previous discussion on the empirical distribution of value, we now extend the result to ranks and prove our second main theorem.

\begin{theorem}
[Empirical distribution of ranks]\label{Thm_main_rank_dist_body}
For any fixed $\eps>0$,
\begin{equation}\label{Eqn_rank_dist_main_thm_body_whp}
    \Pp\bigg(\max_{\mu\in\mathcal{S}} \inf_{\lambda\in\R_+} \|\Femp(\bm{\phi}^{-1}\circ\mathbf{R}(\mu))-F_\lambda\|_\infty > \eps\bigg) \lesssim e^{-n^c}
\end{equation}
asymptotically as $n\to\infty$,
where $\bm\phi$ is men's fitness vector.
As in Theorem~\ref{Thm_main_happiness_dist_body}, the infimum over $\lambda$ in \eqref{Eqn_happiness_dist_main_body_thm_whp} can be replaced with the choice of $\lambda=\|\vY_\delta(\mu)\|_1$ for $\delta$ sufficiently small.
\end{theorem}

Heuristically, we would expect the rank $R_i(\mu)$ for a man to be proportional to his value $X_i(\mu)$ when stability of $\mu$ and the values $X_i(\mu)=x_i$ and $\vY(\mu)=\vy$ are conditioned upon. Indeed, a woman $w_j$ with $j\ne\mu(i)$ stands ahead of $w_{\mu(i)}$ in the preference of man $m_i$ exactly when $X_{ij}<x_i$ and $Y_{ji}>y_{\mu(i)}$. As $X_{ij}$ and $Y_{ji}$ jointly follow the product distribution $\Exp(a_{ij})\otimes\Exp(b_{ji})$ conditional on the event that $X_{ij}<x_i$ and $Y_{ji}<y_{\mu(i)}$ do not simultaneously happen, the conditional probability that $\w_j\succeq_{\m_i} \w_{\mu(i)}$ is $\frac{(1-e^{-a_{ij}x_i})e^{-b_{ji}y_j}}{1-(1-e^{-a_{ij}x_i})(1-e^{-b_{ji}y_j})} \approx a_{ij} x_i$. By summing over $j\ne \mu(i)$, we should expect that the rank $R_i(\mu)$ to be in expectation close to $x_i \sum_{j\ne i} a_{ij} \approx x_i \phi_i$; further, as the sum of independent Bernoulli random variables, $R_i(\mu)$ should concentrate at its expectation, therefore leading to $R_i(\mu) \approx x_i \phi_i$ simultaneously for most $i\in [n]$.
This intuition is formalized in the proof, which is given in  Appendix~\ref{Append_proof_thm_main_rank}.

\section{Conclusion}

We studied the welfare structure in stable outcomes in large two-sided matching markets with logit-based preferences. Under a contiguity condition that prevents agents from disproportionately favoring or unfavoring other agents, we characterize stable and almost stable  matchings outcomes in terms of the empirical distribution of latent values and ranks.

In particular, our results suggest that the welfare of an agent in a stable matching can be decomposed into three parts: a global parameter that determines the trade-off between the two sides, a personal intrinsic fitness computed from the systematic scores, and an exogenous factor behaving as a standard exponential random variable. In other words, given the market structure (i.e., the systematic scores), the average rank (or value) of the men (or women) is essentially a sufficient statistic for the outcome distribution.

\bibliographystyle{plainnat}
\bibliography{ref}

\begin{thebibliography}{27}
\providecommand{\natexlab}[1]{#1}
\providecommand{\url}[1]{\texttt{#1}}
\expandafter\ifx\csname urlstyle\endcsname\relax
  \providecommand{\doi}[1]{doi: #1}\else
  \providecommand{\doi}{doi: \begingroup \urlstyle{rm}\Url}\fi

\bibitem[Agarwal and Somaini(2018)]{agarwal2018demand}
Nikhil Agarwal and Paulo Somaini.
\newblock Demand analysis using strategic reports: An application to a school
  choice mechanism.
\newblock \emph{Econometrica}, 86\penalty0 (2):\penalty0 391--444, 2018.

\bibitem[Ashlagi et~al.(2014)Ashlagi, Braverman, and
  Hassidim]{ashlagi2014stability}
Itai Ashlagi, Mark Braverman, and Avinatan Hassidim.
\newblock Stability in large matching markets with complementarities.
\newblock \emph{Operations Research}, 62\penalty0 (4):\penalty0 713--732, 2014.

\bibitem[Ashlagi et~al.(2017)Ashlagi, Kanoria, and
  Leshno]{ashlagi2017unbalanced}
Itai Ashlagi, Yash Kanoria, and Jacob~D Leshno.
\newblock Unbalanced random matching markets: {T}he stark effect of
  competition.
\newblock \emph{Journal of Political Economy}, 125\penalty0 (1):\penalty0
  69--98, 2017.

\bibitem[Ashlagi et~al.(2021)Ashlagi, Braverman, Saberi, Thomas, and
  Zhao]{ashlagi2020tiered}
Itai Ashlagi, Mark Braverman, Amin Saberi, Clayton Thomas, and Geng Zhao.
\newblock Tiered random matching markets: Rank is proportional to popularity.
\newblock In \emph{12th Innovations in Theoretical Computer Science Conference
  (ITCS)}, 2021.

\bibitem[Cai and Thomas(2022)]{cai2019short}
Linda Cai and Clayton Thomas.
\newblock The short-side advantage in random matching markets.
\newblock In \emph{SIAM Symposium on Simplicity in Algorithms (SOSA)}, 2022.

\bibitem[Doerr(2020)]{doerr2020probabilistic}
Benjamin Doerr.
\newblock Probabilistic tools for the analysis of randomized optimization
  heuristics.
\newblock In \emph{Theory of evolutionary computation}, pages 1--87. Springer,
  2020.

\bibitem[Dvoretzky et~al.(1956)Dvoretzky, Kiefer, and
  Wolfowitz]{dvoretzky1956asymptotic}
Aryeh Dvoretzky, Jack Kiefer, and Jacob Wolfowitz.
\newblock Asymptotic minimax character of the sample distribution function and
  of the classical multinomial estimator.
\newblock \emph{The Annals of Mathematical Statistics}, pages 642--669, 1956.

\bibitem[Feller(1971)]{feller1971introduction}
William Feller.
\newblock \emph{An introduction to probability theory and its applications,
  Volume 2}, volume~81.
\newblock John Wiley \& Sons, 1971.

\bibitem[Gale and Shapley(1962)]{gale1962college}
David Gale and Lloyd~S Shapley.
\newblock College admissions and the stability of marriage.
\newblock \emph{The American Mathematical Monthly}, 69\penalty0 (1):\penalty0
  9--15, 1962.

\bibitem[Gimbert et~al.(2019)Gimbert, Mathieu, and
  Mauras]{gimbert2019popularity}
Hugo Gimbert, Claire Mathieu, and Simon Mauras.
\newblock Two-sided matching markets with correlated random preferences have
  few stable pairs.
\newblock \emph{arXiv preprint arXiv:1904.03890}, 2019.

\bibitem[Hitsch et~al.(2010)Hitsch, Horta{\c{c}}su, and
  Ariely]{hitsch2010matching}
G{\"u}nter~J Hitsch, Ali Horta{\c{c}}su, and Dan Ariely.
\newblock Matching and sorting in online dating.
\newblock \emph{American Economic Review}, 100\penalty0 (1):\penalty0 130--163,
  2010.

\bibitem[Immorlica and Mahdian(2015)]{immorlica2015incentives}
Nicole Immorlica and Mohammad Mahdian.
\newblock Incentives in large random two-sided markets.
\newblock \emph{ACM Transactions on Economics and Computation (TEAC)},
  3\penalty0 (3):\penalty0 1--25, 2015.

\bibitem[Knuth et~al.(1990)Knuth, Motwani, and Pittel]{knuth1990stable}
Donald~E Knuth, Rajeev Motwani, and Boris Pittel.
\newblock Stable husbands.
\newblock \emph{Random Structures \& Algorithms}, 1\penalty0 (1):\penalty0
  1--14, 1990.

\bibitem[Knuth(1976)]{knuth1976mariages}
Donald~Ervin Knuth.
\newblock \emph{Mariages stables et leurs relations avec d'autres problemes
  combinatoires: introduction a l'analysis mathematique des algorithmes.}
\newblock Les Presses de l'Universit{\'e} de Montr{\'e}al, 1976.

\bibitem[Knuth(1997)]{knuth1997stable}
Donald~Ervin Knuth.
\newblock \emph{Stable marriage and its relation to other combinatorial
  problems: An introduction to the mathematical analysis of algorithms},
  volume~10.
\newblock American Mathematical Soc., 1997.

\bibitem[Kojima and Pathak(2009)]{kojima2009incentives}
Fuhito Kojima and Parag~A Pathak.
\newblock Incentives and stability in large two-sided matching markets.
\newblock \emph{American Economic Review}, 99\penalty0 (3):\penalty0 608--27,
  2009.

\bibitem[Mauras(2021)]{mauras2021two}
Simon Mauras.
\newblock Two-sided random matching markets: {E}x-ante equivalence of the
  deferred acceptance procedures.
\newblock \emph{ACM Transactions on Economics and Computation}, 9\penalty0
  (4):\penalty0 1--14, 2021.

\bibitem[McCullagh(2014)]{mccullagh2014asymptotic}
Peter McCullagh.
\newblock An asymptotic approximation for the permanent of a doubly stochastic
  matrix.
\newblock \emph{Journal of Statistical Computation and Simulation}, 84\penalty0
  (2):\penalty0 404--414, 2014.

\bibitem[Menzel(2015)]{menzel2015large}
Konrad Menzel.
\newblock Large matching markets as two-sided demand systems.
\newblock \emph{Econometrica}, 83\penalty0 (3):\penalty0 897--941, 2015.

\bibitem[Nevzorov(1986)]{nevzorov1986representations}
VB~Nevzorov.
\newblock Representations of order statistics, based on exponential variables
  with different scaling parameters.
\newblock \emph{Journal of Soviet Mathematics}, 33\penalty0 (1):\penalty0
  797--798, 1986.

\bibitem[Pittel(1989)]{pittel1989average}
Boris Pittel.
\newblock The average number of stable matchings.
\newblock \emph{SIAM Journal on Discrete Mathematics}, 2\penalty0 (4):\penalty0
  530--549, 1989.

\bibitem[Pittel(1992)]{pittel1992likely}
Boris Pittel.
\newblock On likely solutions of a stable marriage problem.
\newblock \emph{The Annals of Applied Probability}, 2\penalty0 (2):\penalty0
  358--401, 1992.

\bibitem[Pittel(2019)]{pittel2019likely}
Boris Pittel.
\newblock On likely solutions of the stable matching problem with unequal
  numbers of men and women.
\newblock \emph{Mathematics of Operations Research}, 44\penalty0 (1):\penalty0
  122--146, 2019.

\bibitem[Roth(2018)]{roth2018marketplaces}
Alvin~E Roth.
\newblock Marketplaces, markets, and market design.
\newblock \emph{American Economic Review}, 108\penalty0 (7):\penalty0 1609--58,
  2018.

\bibitem[Roth and Sotomayor(1992)]{roth1992two}
Alvin~E Roth and Marilda Sotomayor.
\newblock Two-sided matching.
\newblock \emph{Handbook of game theory with economic applications},
  1:\penalty0 485--541, 1992.

\bibitem[Schweitzer(1914)]{schweitzer1914egy}
P{\'a}l Schweitzer.
\newblock \emph{Egy egyenl{\H{o}}tlens{\'e}g az aritmetikai
  k{\"o}z{\'e}p{\'e}rt{\'e}kr{\H{o}}l ({I}nequality containing the arithmetic
  mean)}.
\newblock 1914.

\bibitem[Sinkhorn(1964)]{sinkhorn1964relationship}
Richard Sinkhorn.
\newblock A relationship between arbitrary positive matrices and doubly
  stochastic matrices.
\newblock \emph{The annals of mathematical statistics}, 35\penalty0
  (2):\penalty0 876--879, 1964.

\end{thebibliography}

\appendix

\section{Proofs of typical behaviors of scores in stable matching}\label{Append_weak_regular_scores}

Recall that for a matching $\mu$ with (latent) value vectors $\vX(\mu)$ and $\vY(\mu)$, we define $\vU(\mu) = F(\hat{\vX}(\mu))$ and $\vV(\mu) = f(\hat{\vX}(\mu))$ with the standard exponential CDF $F(z) = 1-e^{-z}$ applied component-wise to the renormalized values $\hat{X}_i(\mu) = X_i(\mu) a_{i,\mu(i)}$ and $\hat{Y}_i(\mu) = Y_i(\mu) b_{i,\mu^{-1}(i)}$ for $i\in[n]$, so that $U_i,V_i\sim\Unif([0,1])$ and are mutually independent due to the way the score matrices are generated.

\begin{lemma}\label{Lemma_DA_prop_num}
For any $c \in (0,1/2)$, there exists a constant $\theta_1> 0$ (depending on $c$ and $C$) such that in a random instance of the market, at least $\theta_1 n\ln n$ proposals are made during man-proposing deferred acceptance with probability $1 - \exp(-n^c)$.
\end{lemma}
\begin{proof}
This result follows from a standard analysis of the deferred acceptance algorithm executed as in \cite[Section~3]{ashlagi2020tiered}. In \cite{ashlagi2020tiered}, the preference model involves tiers, where fitness values among different tiers differ by at most a constant factor. It turns out that this bounded ratio of fitness is the only thing used in the proofs, and is also satisfied by our matching market under Assumption~\ref{Assumption_C_bounded}. The main steps of analysis are as follows.
\begin{enumerate}
    \item Consider (man-proposing) deferred acceptance \emph{with re-proposals}, where each time when man $i$ proposes, his proposal will go to woman $i$ with probability proportional to $a_{ij}$, independent of all previous proposals (and their acceptance/rejection). The total number $T$ of proposals in this process is equal in distribution to the number of draws in the coupon collector problem, and from standard concentration bounds we can show $\Pp(T \ge k n^{1+c}) \le e^{-n^c-2}$ for $k$ sufficiently large and $\Pp(T < \alpha n\ln n) \le \exp(-n^c-2)$ for $\alpha>0$ sufficiently small (see also \cite{doerr2020probabilistic}). The details mirror Appendices~A and B in \cite{ashlagi2020tiered}.
    \item Analogous to Lemmas~3.5 and 3.6 in \cite{ashlagi2020tiered}, we can show that, with probability $1-\exp(-n^c-1)$, no single man makes more than $\ell n^{2c}$ proposals during deferred acceptance with re-proposal for $\ell$ sufficiently large.
    \item Conditional on $T\ge \alpha n\ln n$ and the maximum number of proposals any man makes being at most $\ell n^{2c}$, the fraction of re-proposals should be no greater than $C\ell n^{2c-1}$ in expectation, since each proposal will be a duplicate of a previous proposal independently with probability at most $\frac{C\ell n^{2 n}}{n}$. It follows immediately from a binomial concentration that the (conditional) probability that the number of repeated proposals exceeds $T/2$ is exponentially small. Hence, the actual number of proposals during deferred acceptance (without re-proposals) is at least $\frac{\alpha}{2} n\ln n$ with probability $1-\exp(-n^c)$.
\end{enumerate}%
\end{proof}

\begin{lemma}\label{Lemma_U1_ge_ln_n}
For any $c \in (0,1/2)$, there exists a constant $\theta_1> 0$ (depending on $c$ and $C$) such that,
with probability $1-\exp(-n^c)$,
there exist no stable matchings $\mu$ where $\|\vU(\mu)\|_1 \le \theta_2\ln n$.
\end{lemma}
\begin{proof}
Note that $U_i(\mu) = 1-e^{-a_{i,\mu(i)} X_i(\mu)} \in [C^{-1}F(X_i(\mu)),C F(X_i(\mu))]$ since $a_{i,\mu(i)}\in[1/C, C]$. For any stable matching $\mu$, $\vX(\mu) \succeq \vX(\mu_\text{MOSM})$ where $\mu_\text{MOSM}$ is the man-optimal stable matching obtained from the man-proposing deferred acceptance, where all the $n$ men are matched with their optimal possible stable partner (and hence achieves best value) simultaneously, and hence
\begin{equation}
    \|\vU(\mu)\|_1 \ge \frac{1}{C} \|F(\vX(\mu))\|_1 \ge \frac{1}{C} \|F(\vX(\mu))\|_1 \ge \frac{1}{C^2} \|\vU(\mu_\text{MOSM})\|_1.
\end{equation}
Thus, it suffices to consider the event where $\|\vU(\mu_\text{MOSM})\|_1 \le \theta_2 C^2 \ln n$.
Without loss of generality, we may assume that $\mu_\text{MOSM}$ matches man $i$ with woman $i$ for $i\in[n]$. Denote by $R_i\in[n]$ and $X_i=X_{ii}\in\R_+$ the (random) rank of partner and the latent value in $\mu_\text{MOSM}$ for man $i\in[n]$, and let $U_i = F(a_{ii}X_i)$. By definition, $X_i\sim\Exp(a_{ii})$ and $R_i = \sum_{j\ne i} \mathbbm{1}\{X_{ij}/a_{ij} < X_i\}$.

We condition on a specific execution of the man-proposing deferred acceptance algorithm, i.e., on the sequence of proposals, which specifies the rank $R_i$ and an ordering over his top $R_i$ most preferred partners for each man. Notice that the specific value $X_{ij}$ only affect the execution through the ordering of proposals, and hence conditional on a particular ordering, the values of the men are independent. Further, the value $X_i$ conditional on $R_i$ and an ordering $w_{j_1} \succeq_{m_i} w_{j_2} \succeq_{m_i} w_{j_1} \succeq \cdots \succeq_{m_i} w_{j_{R_i}}$ is equal in distribution to $\tilde{X}_{(R_i)}$, i.e., the $R_i$-th order statistic of $(\tilde{X}_1,\ldots,\tilde{X}_n)$ with $\tilde{X}_j \sim \Exp(a_{ij})$, conditional on $\tilde{X}_{(k)} = \tilde{X}_{j_k}$ for all $k\in[R_i]$. By the representation of exponential order statistics given in \cite{nevzorov1986representations}, under such conditions,
\begin{equation}
    \tilde{X}_{(R_i)} \overset{d}{=} \sum_{t=1}^{R_i} \frac{Z_{i,t}}{\sum_{k=t}^n a_{i,j_k}}, %
\end{equation}
where $Z_{i,t}\sim\Exp(1)$ are independently sampled for $t\in[n]$. Conditional on $R_i$ and the sequence $j_1,\ldots,j_{R_i}$, we have
\begin{multline}
    U_i = F(a_{ii} X_i) \overset{d}{=} F\bigg(\sum_{t=1}^{R_i} \frac{a_{ii} Z_{i,t}}{\sum_{k=t}^n a_{i,j_k}} \bigg) \ge F\bigg( \frac{1}{n}\sum_{t=1}^{R_i} \frac{Z_{i,t}}{C^2} \bigg) \\
    \ge \frac{1}{n}\sum_{t=1}^{R_i}F(Z_{i,t}/C^2) \ge \frac{1}{C^2 n} \sum_{t=1}^{R_i} F(Z_{i,t}) = \frac{1}{C^2 n} \sum_{t=1}^{R_i} W_{i,t},
\end{multline}
where the second last inequality follows from Jensen's inequality (applied to the concave function $F$), and $W_{i,t}=F(Z_{i,t}) \sim \Unif([0,1])$ independently.
Thus, conditional on $\mathbf{R}$, we have $C^2 n\|\vU\|_1 \succeq \sum_{i=1}^n \sum_{t=1}^{R_i} W_{i,t}$, independent of the specific ordering (and the identity) of the proposals made. Therefore, we may marginalize over this ordering to get
\begin{equation}
    \Pp\bigg(\|\vU\|_1 \le \theta_2 C^2 \ln n\bigg| \mathbf{R}\bigg) \le \Pp\left( \sum_{i=1}^n \sum_{t=1}^{R_i} W_{i,t} \le \theta_2 C^4 n \ln n \middle| \mathbf{R}\right).
\end{equation}
Whenever $\|\mathbf{R}\|_1 \ge \theta_1 n \ln n$, the probability above is at most $\exp(-\Theta(n\ln n)) \ll \exp(-n^c)$ by Hoeffding's inequality granted that we choose $\theta_2 < \frac{\theta_1}{2C^4}$,
and our proof is complete as we marginalize over all possible realizations of $\mathbf{R}$ with $\|\mathbf{R}\|_1\ge \theta_1 n\ln n$.
\end{proof}

\begin{lemma}\label{Lemma_U1V1_le_O_n_ln}
For any $\kappa \ge 0$ and $\theta_3 > 0$, the expected number of stable matchings $\mu$ with $\|\vU(\mu)\|_1 \|\vV(\mu)\|_1 \ge \theta_3 n(\ln n)^{1/8}$ is upper bounded by $\exp(-\kappa n)$.
In particular, with high probability, no such stable matchings exists.
\end{lemma}
\begin{proof}
It suffices to show that the probability that any fixed matching $\mu$ is stable and satisfies $\|\vU(\mu)\|_1 \|\vV(\mu)\|_1 \ge \theta_3 n(\ln n)^{1/8}$ is upper bounded by $o(e^{-\kappa n}/n!)$ for any $\theta_3>0$. Write $I=[0,1]$ for the unit interval. Let
\begin{equation*}
    \Omega = \left\{(\vu,\vv)\in I^n\times I^n : \|\vu\|_1\|\vv\|_1 > \theta_3 n(\ln n)^{1/8}\right\}
\end{equation*}
and let
\begin{equation}
    P := \Pp(\mu\in\mathcal{S}, (\vU(\mu),\vV(\mu))\in \Omega) = \int_{\R^n_+\times \R^n_+} \pxy \cdot \mathbbm{1}_{\Omega}(F(\hat{\vx}), F(\hat{\vy})) \cdot \prod_{i=1}^n f(\hat{x}_i)f(\hat{y}_i) \dd\hat{\vx} \dd\hat{\vy},
\end{equation}
where $f(t)=e^{-t}$ denotes the standard exponential density function. We apply the simple bound \eqref{Eqn_naive_bound_pxy} on $\pxy$ to obtain
\begin{align}
    P &\le \int_{\R^n_+\times \R^n_+} \prod_{i\ne j} \Big(1 - \big(1-e^{-\hat{x}_i/C^2}\big)\big(1-e^{-\hat{y}_j/C^2}\big) \Big) \cdot \mathbbm{1}_{\Omega}(F(\hat{\vx}), F(\hat{\vy})) \cdot \prod_{i=1}^n f(\hat{x}_i)f(\hat{y}_i) \dd\hat{\vx} \dd\hat{\vy} \nonumber\\
    &= \int_{I^n\times I^n} \prod_{i\ne j} \Big(1 - \big(1-(1-u_i)^{1/C^2}\big)\big(1-(1-v_j)^{1/C^2}\big) \Big) \cdot \mathbbm{1}_{\Omega}(\vu, \vv) \dd\vu \dd\vv \nonumber\\
    &\le \int_{I^n\times I^n} \prod_{i\ne j} \Big(1 - \frac{1}{C^4} u_iv_j \Big) \cdot \mathbbm{1}_{\Omega}(\vu, \vv) \dd\vu \dd\vv \nonumber\\
    &\le \int_{I^n\times I^n} \exp\Big(-\frac{1}{C^4}(\|\vu\|_1\|\vv\|_1 - \vu\cdot\vv)\Big) \cdot \mathbbm{1}_{\Omega}(\vu, \vv) \dd\vu \dd\vv \nonumber\\
    &\le \int_{I^n\times I^n} \exp\Big(\frac{1}{C^4}(n - \|\vu\|_1\|\vv\|_1)\Big) \cdot \mathbbm{1}_{\Omega}(\vu, \vv) \dd\vu \dd\vv\label{Eqn_Lemma_weak_reg_U_apply_naive_bound},
\end{align}
where we use the basic facts that $1-z^{1/C^2}\ge (1-z)/C^2$ for all $z\in[0,1]$ and $C\ge 1$ and that $1+z\le e^z$ for all $z\in\R$.

Let $s = \|\vu\|_1$ and $t = \|\vv\|_1$. It is well known (e.g., \cite[Ch.~I,~Sec.~9]{feller1971introduction}) that the probability density of $S=\|\vU\|_1$ with $U_1,\ldots,U_n$ independent samples from $\Unif(I)$ is bounded above by $\frac{s^{n-1}}{(n-1)!}$. Hence,
\begin{equation}
    P \le \int_{s,t\in[0,n]\,:\, st\ge \theta_3 n(\ln n)^{1/8}} e^{(n-st)/C^4} \cdot \frac{n^2(st)^{n-1}}{(n!)^2} \dd s \dd t.
\end{equation}
Note that when $st \ge n (\ln n)^2$, we have $\exp\big((n-st)/C^4\big) \le \exp\big((n-n(\ln n)^2)/C^4) = o(\exp(-\kappa n))/n!$ and therefore the region $\{s,t\in[0,n]\,:\, st\ge n(\ln n)^2\}$ contributes a negligible amount to the integral. Hence,
\begin{align}\label{Eqn_Lemma_weak_reg_U_p2_final_bound}
    P &\le \frac{o(e^{-\kappa n})}{n!} + \int_{s,t\in[0,n]\,:\, \theta_3 n(\ln n)^{1/8} \le st\le n(\ln n)^2} e^{(n-st)/C^4} \cdot \frac{n^2(st)^{n-1}}{(n!)^2} \dd s \dd t \nonumber\\
    &\labelrel\le{Rel_proof_lemma_u1v1_1} \frac{o(e^{-\kappa n})}{n!} + n^2\cdot e^{(n-\theta_3 n(\ln n)^{1/8})/C^4} \cdot \frac{n^2(n(\ln n)^2)^{n-1}}{(n!)^2} \nonumber\\
    &\labelrel\le{Rel_proof_lemma_u1v1_2} \frac{o(e^{-\kappa n})}{n!} + \frac{1}{n!} \cdot \exp\Big(\frac{1}{C^4}(n-\theta_3 n(\ln n)^{1/8}) + 3\ln n + 2(n-1)\ln\ln n + n\Big) = \frac{o(e^{-\kappa n})}{n!},
\end{align}
where in step \eqref{Rel_proof_lemma_u1v1_1} we upper bound the integral by the product of the Lebesgue measure of its domain (bounded by $n^2$) and the supremum of its integrand, and in step \eqref{Rel_proof_lemma_u1v1_2} we invoke Stirling's approximation.
\end{proof}

\begin{corollary}\label{Cor_V1_le_n_over_ln}
For any constant $c\in(0,1/2)$,
there exists a constant $\theta_4> 0$ such that,
with probability $1-\exp(-n^c)$,
there exist no stable matchings $\mu$ with $\|\vU(\mu)\|_1 \ge \frac{\theta_4 n}{(\ln n)^{7/8}}$.
\end{corollary}
\begin{proof}
This follows immediately from Lemma~\ref{Lemma_U1_ge_ln_n} and \ref{Lemma_U1V1_le_O_n_ln}.
\end{proof}

\begin{proposition}\label{Prop_characterize_low_prob_events}
Let $\Omega\subset I^n$ be a (sequence of) regions in the $n$-dimensional hypercube. For $k\in\Z_+$, define interval $I_k = (2^{-k}n, 2^{-k+1}n]$. If 
\begin{equation}\label{Eqn_prop_characterize_low_prob_events_main_assumpt}
    \Pp_{\vW\sim\Exp(2^k)^{\otimes n}}(\vW\in \Omega, \|\vW\|_1\in I_k) \le \frac{g(n)}{e^{6n}C^{8n}}
\end{equation}
for some function $g(n)$ and uniformly for all $k\in\Z_+$, we can guarantee that the expected number of stable matchings $\mu$ with value vector $\vX(\mu)\in\R^n$ satisfying $\vU(\mu) \in \Omega$ is upper bounded by $g(n)+e^{-\Theta(n^2)}$; %
in particular, with high probability, no such stable matchings exist if $g(n) = o(1)$.
\end{proposition}
\begin{proof}
We focus on a fixed matching $\mu$ and by union bound, it suffices to show that
\begin{equation}
    \Pp(\mu\in\mathcal{S}, \vU(\mu)\in\Omega) = \frac{g(n)+ e^{-\Theta(n^2)}}{n!}
\end{equation}
under the condition of \eqref{Eqn_prop_characterize_low_prob_events_main_assumpt}.
The same chain of reasoning as in \eqref{Eqn_Lemma_weak_reg_U_apply_naive_bound} gives
\begin{align}
    P := \Pp(\mu\in\mathcal{S}, \vU(\mu) \in \Omega) &\le \int_{I^n\times I^n} \exp\Big(\frac{1}{C^4}(n - \|\vu\|_1\|\vv\|_1)\Big) \cdot \mathbbm{1}_{\Omega}(\vu) \dd\vu \dd\vv \nonumber\\
    &\le e^n \int_{I^n\times I^n} \exp( - \|\vu\|_1\|\vv\|_1/C^4)) \cdot \mathbbm{1}_{\Omega}(\vu) \dd\vu \dd\vv
\end{align}
since $C\ge 1$.
Observing that $\|\vU(\mu)\|_1=0$ and $\|\vV(\mu)\|_1=0$ are both probability zero events, we may split the domain of the integral above into sub-regions according to which intervals $\|\vU(\mu)\|_1$ and $\|\vV(\mu)\|_1$ fall into, and then bound the value of the integral within each sub-region. That is, with the help of the monotone convergence theorem to interchange summation with integral,
\begin{align}
    P &\le e^n \sum_{k,\ell=1}^\infty \int_{I^n\times I^n}\exp( - \|\vu\|_1\|\vv\|_1/C^4) \cdot \mathbbm{1}_{\Omega}(\vu) \mathbbm{1}_{I_k}(\|\vu\|_1) \mathbbm{1}_{I_\ell}(\|\vv\|_1) \dd\vu \dd\vv \nonumber\\
    &\le e^n \sum_{k,\ell=1}^\infty \int_{I^n\times I^n}\exp( - 2^{-k-l}n^2/C^4) \cdot \mathbbm{1}_{\Omega}(\vu) \mathbbm{1}_{I_k}(\|\vu\|_1) \mathbbm{1}_{I_\ell}(\|\vv\|_1) \dd\vu \dd\vv.
\end{align}
For all $k\in\Z_+$, we have $e^{2n}\exp(-2^k\|\vu\|_1)\ge 1$ whenever $\vu\in I^n$ with $\|\vu\|_1\in I_k$. Thus,
\begin{align}
    P &\le e^{5n} \sum_{k,\ell=1}^\infty \int_{I^n\times I^n} \exp( - 2^{-k-l}n^2/C^4) \cdot \mathbbm{1}_{\Omega}(\vu) \mathbbm{1}_{I_k}(\|\vu\|_1) \mathbbm{1}_{I_\ell}(\|\vv\|_1) \cdot \exp(-2^k\|\vu\|_1 - 2^\ell\|\vv\|_1) \dd\vu \dd\vv \nonumber\\
    &= e^{5n} \sum_{k,\ell=1}^\infty 2^{-(k+\ell)n} \exp( - 2^{-k-l}n^2/C^4) \E_{\vU\sim\Exp(2^k)^{\otimes n},\vV\sim\Exp(2^\ell)^{\otimes n}}\left[ \mathbbm{1}_{\Omega}(\vU)\mathbbm{1}_{I^n}(\vV)\mathbbm{1}_{I_k}(\|\vU\|_1) \mathbbm{1}_{I_\ell}(\|\vV\|_1)\right].
\end{align}
Observe that all the terms with $k+\ell \ge n$ combined contribute at most
\begin{multline}
    e^{5n}\sum_{k+\ell \ge n} 2^{-(k+\ell)n} = e^{5n} \left( \sum_{k=1}^{n-1} 2^{-kn} \sum_{\ell=n-k}^\infty 2^{-\ell n} + \sum_{k=n}^\infty 2^{-kn} \sum_{\ell=1}^\infty 2^{-\ell n} \right) \\
    = e^{5n}\cdot 2^{-n^2-n} \cdot\frac{n(1-2^{-n})+1}{(1-2^{-n})^2} = \frac{e^{-\Theta(n^2)}}{n!},
\end{multline}
which is negligible. Therefore, we only need to consider $\bO(n^2)$ terms and
\begin{align}
    P
    &\le \frac{e^{-\Theta(n^2)}}{n!} + n^2 e^{5n} \max_{k,\ell\in\Z_+} \frac{\E_{\vU\sim\Exp(2^k)^{\otimes n},\vV\sim\Exp(2^\ell)^{\otimes n}}\left[ \mathbbm{1}_{\Omega}(\vU) \mathbbm{1}_{I^n}(\vV) \mathbbm{1}_{I_k}(\|\vU\|_1) \mathbbm{1}_{I_\ell}(\|\vV\|_1)\right]}{2^{(k+\ell)n} \exp( 2^{-k-l}n^2/C^4 )} \nonumber\\
    &\le \frac{e^{-\Theta(n^2)}}{n!} + e^{6n} \max_{k\in\Z_+} \frac{\Pp_{\vU\sim\Exp(2^k)^{\otimes n}}(\vU\in \Omega, \|\vU\|_1\in I_k)}{2^{(k+\ell)n} \exp( 2^{-k-l}n^2/C^4 )}.
\end{align}
Let $\alpha = 2^{-(k+\ell)}$. When $\alpha\le C^8/n$, the denominator of the second term can be bounded as $\alpha^{-n}e^{\alpha n^2/C^4} \ge n^n/C^{8n} \ge n!/C^{8n}$; when $\alpha > C^8/n$, let $\tau = n\alpha > C^8$ and we have $\alpha^{-n}e^{\alpha n^2/C^4} = \frac{n^n}{\tau^n} e^{n\tau/C^4} = n^n \exp\big(n(\tau/C^4 - \ln \tau)\big) \ge n^n \exp\big(n(C^4 - 4\ln C)\big) \ge n!$. Thus, the denominator is bounded below by $n!/C^{8n}$ and
\begin{equation}\label{Eqn_proof_X1_O_V1_last_reusable_bound}
    P \le \frac{e^{-\Theta(n^2)}}{n!} + \frac{e^{6n} C^{8n}}{n!} \max_{k\in\Z_+} \Pp_{\vU\sim\Exp(2^k)^{\otimes n}}(\vU\in \Omega, \|\vU\|_1\in I_k) \le \frac{1}{n!}(e^{-\Theta(n^2)} + g(n)).
\end{equation}
The claim follows immediately.
\end{proof}

\begin{lemma}\label{Lemma_bdd_x_delta_infty_norm}
For any fixed $\delta > 0$ and $\kappa > 0$, there exists an absolute constant $\theta_5$ (depending on $\delta,\kappa$, and $C$) such that, with probability $1-\exp(-\kappa n)$,
there exist no stable matchings $\mu$ with $\|\vX_\delta(\mu)\|_\infty \le \theta_5$. (Recall that $\vX_\delta(\mu)$ is the value vector of the $(1-\delta)$-partial matching obtained from $\mu$ that excludes the least happy $\delta/2$ fraction of men and women.)
\end{lemma}
\begin{proof}
Since $\hat{\vX}$ and $\vX$ differ by at most a factor of $C$ component-wise, we have
\begin{equation}\label{Eqn_proof_Lemma_bdd_x_delta_infty_norm_translate_to_U_quantile}
    \|\vX_\delta(\mu)\|_\infty \le X_{(n-\floor{\delta n/2})}(\mu) \le C \hat{X}_{(n-\floor{\delta n/2})}(\mu) = -C\log\big(1-U_{(n-\floor{\delta n/2})}(\mu)\big).
\end{equation}
Thus, it suffices to bound the upper $\delta/2$ quantile $U_{(n-\floor{\delta n/2})}(\mu)$ away from $1$. 

Let $\Omega = \{\vu\in I^n:\vu_{(n-\floor{\delta n/2})} > 1-e^{-s}\}$ for some $s\ge 1$ that we will specified later. Then $\vW\in\Omega$ implies that $\sum_{i=1}^n \mathbbm{1}_{(1-e^{-s},1]}(W_i) \ge \delta n/2$. For $W_i\sim\Exp(2^k)$, we have
\begin{equation*}
    \Pp(W_i\in (1-e^{-s},1]) = \int_{1-e^{-s}}^1 2^k e^{-2^k t} dt \le e^{-s}\cdot 2^ke^{-2^{k-1}} \le e^{-s}.
\end{equation*}
Thus, $\sum_{i=1}^n \mathbbm{1}_{(1-e^{-s},1]}(W_i)$ is stochastically dominated by a $\Binom(n, e^{-s})$ random variable, and as a result
\begin{multline*}
    \Pp_{\vW\sim\Exp(2^k)^{\otimes n}}(\vW\in\Omega, \|\vW\|_1\in I_k) \le \Pp_{\vW\sim\Exp(2^k)^{\otimes n}}(\vW\in\Omega)  \\
    \le \Pp_{\vW\sim\Exp(2^k)^{\otimes n}}\left(\sum_{i=1}^n \mathbbm{1}_{(1-e^{-s},1]}(W_i) \ge \frac{\delta n}{2}\right)
    \le \Pp_{Z\sim\Binom(n, e^{-s})}\bigg(Z\ge \frac{\delta n}{2}\bigg) \le \exp(-n D(\delta/2 \| e^{-s})).
\end{multline*}
Since $D(\delta/2 \| z)\to\infty$ as $z\to 0$, it suffices to take $s$ sufficiently large to guarantee $D(\delta/2\|e^{-s}) > \kappa + 6 + 8 \log C$. Proposition~\ref{Prop_characterize_low_prob_events} then guarantees with probability $1-\exp(-\kappa n)$ that no stable matchings $\mu$ have $\vU(\mu) \in \Omega$, and as a result of \eqref{Eqn_proof_Lemma_bdd_x_delta_infty_norm_translate_to_U_quantile}, with at least the desired probability, no stable matchings $\mu$ should have $\|\vX_\delta(\mu)\|_\infty > \theta_5 := C s$.
\end{proof}

\begin{lemma}\label{Lemma_X1_le_O_U1}
For any $\delta > 0$ and $\kappa > 0$, there exists an absolute constant $\theta_6$ (again depending on $\delta,\kappa$, and $C$) such that, with probability $1-\exp(-\kappa n)$, there exist no stable matchings $\mu$ with $\|\vX_\delta(\mu)\|_1 \ge \theta_6 \|\vU(\mu)\|_1$.
\end{lemma}
\begin{proof}
Take $\theta_6 = C \sup_{0<x\le C\theta_5} \frac{x}{F(x)} = \frac{C^2\theta_5}{1-e^{-C\theta_5}}$, where $\theta_5$ is the constant in Lemma~\ref{Lemma_bdd_x_delta_infty_norm}. Assume that $\|\vX_\delta(\mu)\|_\infty \le \theta_5$ for all stable matchings $\mu$ since the probability otherwise is at most $\exp(-\kappa n)$ as desired. Note that for each $i$ in the support of $\vX_\delta(\mu)$ (i.e., $(X_\delta)_i(\mu) > 0$), we have
\begin{equation}
    \hat{X}_i(\mu) \le C X_i(\mu) = C (X_\delta)_i(\mu) \le C\theta_5,
\end{equation}
and subsequently
\begin{equation}
    U_i(\mu) = F(\hat{X}_i(\mu)) \ge \frac{C\hat{X}_i(\mu)}{\theta_6} \ge \frac{X_i(\mu)}{\theta_6} = \frac{(X_\delta)_i(\mu)}{\theta_6},
\end{equation}
and this final inequality is trivial for any $i$ not in the support of $\vX_\delta(\mu)$. The claim then follows immediately.
\end{proof}

\begin{lemma}\label{Lemma_Xdelta2sq_le_O_V1sq}
For any fixed $\delta > 0$ and $\kappa > 0$, there exists an absolute constant $\theta_7$ (depending on $\delta,\kappa$, and $C$) such that, with probability $1-\exp(-\kappa n)$,
there exist no stable matchings $\mu$ with $\|\vX_\delta(\mu)\|_2^2 \ge \theta_7 \|\vU(\mu)\|_1^2 /n$.
\end{lemma}
\begin{proof}
Taking advantage of Lemma~\ref{Lemma_bdd_x_delta_infty_norm}, let us assume that $X_{(n-\floor{\delta n/2})}(\mu) \le \theta_5$ is satisfied simultaneously by all stable matchings $\mu$ (see the proof, in particular \eqref{Eqn_proof_Lemma_bdd_x_delta_infty_norm_translate_to_U_quantile}, for more details); the event otherwise has probability bounded by $\exp(-\kappa n)$.

Notice that
\begin{equation*}
    \|\vX_\delta(\mu)\|_2^2 \le \sum_{i=1}^{n-\floor{\delta n/2}} X_{(i)}(\mu)^2 \le C^2 \sum_{i=1}^{n-\floor{\delta n/2}} \hat{X}_{(i)}(\mu)^2 \le \theta_6' \sum_{i=1}^{n-\floor{\delta n/2}} U_{(i)}(\mu)^2,
\end{equation*}
where the $(i)$ subscript denotes the $i$-th (lower) order statistics (and in particular, $\hat{X}_{(i)}(\mu)$ is the $i$-th smallest entry of $\hat{\vX}(\mu)$) with $\theta_6' = C^2 \left(\frac{\theta_5}{F(\theta_5)}\right)^2$. Now it suffices to compare $\sum_{i=1}^{n-\floor{\delta n/2}} U_{(i)}(\mu)^2$ with $\|\vU(\mu)\|_1^2 /n$.

Consider $\Omega := \{\vw \in I^n : \sum_{i=1}^{n-\floor{\delta n/2}} w_{(i)}^2 \ge \gamma \|\vw\|_1^2/n\}$ for some $\gamma\in\R_+$ to be specified. By Proposition~\ref{Prop_characterize_low_prob_events}, it suffices to show that for some appropriate value of $\gamma$ we have
\begin{equation*}
    \Pp_{\vW\sim\Exp(2^k)^{\otimes n}}(\vW\in\Omega, \|\vW\|_1\in I_k) \le e^{-(\kappa+6)n}C^{-8n}
\end{equation*}
for all $k\in\Z_+$. Observe that
\begin{align*}
    \Pp_{\vW\sim\Exp(2^k)^{\otimes n}}(\vW\in\Omega, \|\vW\|_1\in I_k) &= \Pp_{\vW\sim\Exp(2^k)^{\otimes n}}\bigg(\sum_{i=1}^{n-\floor{\delta n/2}} W_{(i)}^2 \ge \frac{\gamma \|\vW\|_1^2}{n}, 2^{-k}n<\|\vW\|_1\le 2^{-k+1}n\bigg) \\
    &\le \Pp_{\vW\sim\Exp(2^k)^{\otimes n}}\bigg(\sum_{i=1}^{n-\floor{\delta n/2}} W_{(i)}^2 \ge \gamma 2^{-2k} n\bigg) \\
    &\le \Pp_{\vW\sim\Exp(1)^{\otimes n}}\bigg(\sum_{i=1}^{n-\floor{\delta n/2}} W_{(i)}^2 \ge \gamma n\bigg) \\
    &\le \Pp_{\vW\sim\Exp(1)^{\otimes n}}(W_{(n-\floor{\delta n/2})} \ge \sqrt{\gamma}) \\
    &\le \Pp_{Z\sim\Binom(n, e^{-\sqrt{\gamma}})}\left(Z \ge \frac{\delta n}{2}\right).
\end{align*}
By the large deviation bound for binomial distribution, choosing $\gamma$ sufficiently large such that $D(\delta/2 \| e^{-\sqrt{\gamma}}) > \kappa + 6 + 8\log C$ ensures that this probability is $o(e^{-(\kappa+6)n}C^{-8n})$. This finishes the proof with the choice of $\theta_7 = \theta_6' \gamma$.
\end{proof}
\begin{remark}
    This is the only part of our analysis that relies on the $\delta$-truncation of values. Without the truncation, $\frac{1}{n}\|\vW\|_2^2$ would concentrate poorly -- in fact not even having a finite mean -- for $\vW\sim\Exp(1)^{\otimes n}$.
\end{remark}

\begin{corollary}\label{Cor_X_delta_truncate_at_C_over_2_is_small}
For any constant $c \in (0, 1/2)$, there exists a constant $\theta_8> 0$ such that,
with probability $1-\exp(-n^c)$,
there exist no stable matchings $\mu$ with $\sum_{i=1}^n (X_\delta)_i \mathbbm{1}_{[2/C,\infty)}\big((X_\delta)_i\big) \ge \theta_8 \|\vU\|_1/(\ln n)^{7/8}$.
\end{corollary}
\begin{proof}
Notice that
\begin{equation*}
    \|\vX_\delta\|_2^2 \ge \sumiton (X_\delta)_i^2 \mathbbm{1}_{[2/C,\infty)}\big((X_\delta)_i\big) \ge \frac{2}{C} \sumiton (X_\delta)_i \mathbbm{1}_{[2/C,\infty)}\big((X_\delta)_i\big).
\end{equation*}
The statement then follows from Lemma~\ref{Lemma_Xdelta2sq_le_O_V1sq} and Corollary~\ref{Cor_V1_le_n_over_ln}.
\end{proof}

\begin{lemma}\label{Lemma_Xdelta1_ge_U1}
For any fixed $\delta > 0$ and $\kappa > 0$, there exists an absolute constant $\theta_9$ (depending on $\delta,\kappa$, and $C$) such that, with probability $1-\exp(-\kappa n)$,
there exist no stable matchings $\mu$ with $\|\vX_\delta(\mu)\|_1 \le \theta_9 \|\vU(\mu)\|_1$.
\end{lemma}
\begin{proof}
Since $\vX \succeq \vU$ component-wise, we have $\|\vX_\delta(\mu)\|_1 \ge \|\vU_\delta(\mu)\|_1$. Thus, it suffices to consider the condition $\|\vU_\delta(\mu)\|_1 \le \theta_9 \|\vU(\mu)\|_1$.

Consider $\Omega := \{\vw \in I^n : \exists S\subseteq[n], |S|=n-\floor{\delta n},\sum_{i\in S} w_i \le \alpha \|\vw\|_1\}$ for some $\alpha\in\R_+$ to be specified. By union bound, for any $k\in\Z_+$,
\begin{align*}
    \Pp_{\vW\sim\Exp(2^k)^{\otimes n}}(\vW\in\Omega, \|\vW\|_1\in I_k) &\le \binom{n}{\floor{\delta n}} \Pp_{\vW\sim\Exp(2^k)^{\otimes n}}\bigg(\sum_{i=1}^{n-\floor{\delta n}} w_i \le \alpha \|\vW\|_1\le 2^{-k+1}\alpha n\bigg) \\
    &= \binom{n}{\floor{\delta n}} \Pp_{\vW\sim\Exp(1)^{\otimes n}}\bigg(\sum_{i=1}^{n-\floor{\delta n}} w_i \le 2\alpha n\bigg) \\
    &= \exp(-h(\delta) n + o(n)) \cdot \bigg(\frac{2\alpha e}{1-\delta}\bigg)^{n-\floor{\delta n}},
\end{align*}
where in the last step we use Stirling's approximation to bound the first factor and standard (lower) concentration of $\Exp(1)$ to bound the probability term (e.g., see Lemma~\ref{Lemma_weighted_exp_chernoff}).
For $\alpha$ sufficiently small, e.g., $\alpha < \exp\big(\frac{1}{1-\delta}(h(\delta)-\kappa - 6-8\ln C)-h(\delta)\big)$, we have
\begin{equation*}
    \Pp_{\vW\sim\Exp(2^k)^{\otimes n}}(\vW\in\Omega, \|\vW\|_1\in I_k) \le e^{-(\kappa+6)n}C^{-8n}
\end{equation*}
for all $k\in\Z_+$. Invoking Lemma~\ref{Prop_characterize_low_prob_events} concludes the proof with $\theta_9 = \alpha$.
\end{proof}

\begin{corollary}\label{Cor_Xdelta1_ge_ln_n}
    For any constant $c \in (0, 1/2)$, there exists a constant $\theta_{10}> 0$ such that,
    with probability $1-\exp(-n^c)$,
    there exist no stable matchings $\mu$ with $\|\vX_\delta(\mu)\|_1 \le \theta_{10} \ln n$.
\end{corollary}
\begin{proof}
    This follows from Lemmas~\ref{Lemma_U1_ge_ln_n} and \ref{Lemma_Xdelta1_ge_U1}, with $\theta_{10} = \theta_2\theta_9$.
\end{proof}

The following Corollary combines all the previous into the typical behavior of value vectors in stable matchings.

\begin{corollary}\label{Cor_Rstar_likely}
Define $\mathcal{R}^\star(\mu)\subseteq \R_+^n \times \R_+^n$, in the context of a matching $\mu$, to be the set of all pairs of vectors $(\vx,\vy)\in \R_+^n \times \R_+^n$ that satisfy all of the following conditions:
\begin{align}
    \theta_2 \ln n \le \|\vu\|_1&,\|\vv\|_1 \le \frac{\theta_4 n}{(\ln n)^{7/8}}, \label{Eqn_def_Rstar_1} \\
    \|\vu\|_1\|\vv\|_1 &\le \theta_3 n(\ln n)^{1/8}, \label{Eqn_def_Rstar_2} \\
    \|\vx_\delta\|_1 \le \theta_6 \|\vu\|_1 
    &\text{ and } \|\vy_\delta\|_1 \le \theta_6 \|\vv\|_1, \label{Eqn_def_Rstar_3} \\
    \|\vx_\delta\|_2^2 \le \frac{\theta_7 \|\vu\|_1^2}{n} 
    &\text{ and } \|\vy_\delta\|_2^2 \le \frac{\theta_7 \|\vv\|_1^2}{n}, \label{Eqn_def_Rstar_4} \\
    \sum_{i=1}^n (x_\delta)_i \mathbbm{1}_{[2/C,\infty)}\big((x_\delta)_i\big) \le \frac{\theta_8 \|\vu\|_1}{(\ln n)^{7/8}} 
    &\text{ and } \sum_{i=1}^n (y_\delta)_i \mathbbm{1}_{[2/C,\infty)}\big((y_\delta)_i\big) \le \frac{\theta_8 \|\vv\|_1}{(\ln n)^{7/8}}, \label{Eqn_def_Rstar_5} \\
    \|\vx_\delta\|_1,\|\vy_\delta\|_1 &\ge \theta_{10} \ln n, \label{Eqn_def_Rstar_6}
\end{align}
where $u_i = F(a_{i,\mu(i)}x_i)$ and $v_j = F(b_{j,\mu^{-1}(j)}y_j)$ for $i,j\in[n]$; $\vx_\delta$ and $\vy_\delta$ denote the truncated version of $\vx$ and $\vy$;
$\theta_2,\theta_3, \theta_6, \theta_7, \theta_4, \theta_8,\theta_{10}\in\R_+$ are absolute constants (independent of $\mu$) chosen appropriately as in Lemmas~\ref{Lemma_U1_ge_ln_n}, \ref{Lemma_U1V1_le_O_n_ln}, 
 \ref{Lemma_X1_le_O_U1}, \ref{Lemma_Xdelta2sq_le_O_V1sq}, and Corollaries \ref{Cor_V1_le_n_over_ln}, \ref{Cor_X_delta_truncate_at_C_over_2_is_small}, and \ref{Cor_Xdelta1_ge_ln_n}. Then, for any $c\in(0,1/2)$, with probability $1-\exp(-n^c)$, $(\vX(\mu),\vY(\mu)) \in \mathcal{R}^\star(\mu)$ for all stable matchings $\mu$.
\end{corollary}
The proof simply summarizes the aforementioned Lemmas and Corollaries and shall be omitted.

\propRatioPQHighProbeon*
\begin{proof}
Note that $1-e^{-tx} \ge \left(tx - \frac{t^2 x}{2}\right)$ for all $x, t \ge 0$. In particular, $1-e^{-tx} \ge \left(tx - \frac{t^2 x}{2}\right) \mathbbm{1}_{[0, 2/t]}(x) \ge 0$. Using this to approximate $p(\vx,\vy)$ gives
\begin{align}
    p(\vx,\vy) &= \prod_{\substack{i\ne j}} \left(1 -  \big(1-e^{-a_{ij}x_i}\big)\big(1-e^{-b_{ji}y_j}\big) \right) \nonumber\\
    &\le \prod_{\substack{i\ne j}} \left(1 - \mathbbm{1}_{[0, 2/a_{ij}]}(x_i) \mathbbm{1}_{[0, 2/b_{ji}]}(y_j) \bigg(a_{ij}x_i - \frac{a_{ij}^2}{2}x_i^2\bigg)\bigg(b_{ji}y_j-\frac{b_{ji}^2}{2}y_j^2\bigg) \right) \nonumber\\
    &\le \exp\left( -\sum_{i\ne j} \mathbbm{1}_{[0, 2/C]}(x_i) \mathbbm{1}_{[0, 2/C]}(y_j) \bigg(a_{ij}x_i - \frac{a_{ij}^2}{2}x_i^2\bigg)\bigg(b_{ji}y_j-\frac{b_{ji}^2}{2}y_j^2\bigg) \right).
\end{align}
Taking logarithm for simplicity and expanding the expression above gives
\begin{align}
    \ln p(\vx,\vy) & \le - \sum_{i\ne j} \Bigg( a_{ij}b_{ji} x_i y_j - \big(\mathbbm{1}_{(2/C,\infty)}(x_i) + \mathbbm{1}_{(2/C,\infty)}(y_j)\big) a_{ij} b_{ji} x_i y_j \nonumber \\
    &\qquad\qquad - \mathbbm{1}_{[0, 2/C]}(x_i) \mathbbm{1}_{[0, 2/C]}(y_j) \bigg( a_{ij}^2 b_{ji} x_i^2 y_j + a_{ij} b_{ji}^2 x_i y_j^2
    \bigg) \Bigg) \nonumber \\
    &\le
        -\sum_{i, j=1}^n a_{ij}b_{ji} x_i y_j
          + \sum_{i=1}^n C^2 x_i y_i \nonumber \\
    &\qquad\qquad + \sum_{i, j=1}^n \Bigg( C^2 \big(\mathbbm{1}_{(2/C,\infty)}(x_i) + \mathbbm{1}_{(2/C,\infty)}(y_j)\big) x_i y_j + C^3 \bigg( x_i^2 y_j + x_i y_j^2 \bigg) \Bigg). %
\end{align}
Notice that $-\ln q(\vx,\vy) = \sum_{i, j=1}^n a_{ij}b_{ji}\frac{x_i}{a_{ii}}\frac{y_j}{b_{jj}}$. Thus,
\begin{multline}\label{Eqn_proof_ratio_p_q_high_prob_diff_pq}
    \ln \frac{p(\vx,\vy)}{q(\vx,\vy)} \le C^2 \vx^\top \vy + C^2\left(\|\vx\|_1 \sumiton\mathbbm{1}_{(2/C,\infty)}(y_j) y_j + \|\vy\|_1 \sumiton\mathbbm{1}_{(2/C,\infty)}(x_i) x_i\right) \\
    + C^3 \left(\|\vx\|_2^2 \|\vy\|_1 + \|\vx\|_1 \|\vy\|_2^2\right).
\end{multline}

In light of Corollary~\ref{Cor_Rstar_likely}, it suffices to upper bound $\ln\frac{p(\vx_\delta,\vy_\delta)}{q(\vx_\delta,\vy_\delta)}$ by $cn/(\ln n)^{1/2}$ for all $(\vx,\vy)\in\mathcal{R}^\star(\mu)$ and for all $\mu$. To simplify notation, we will make the dependency on $\mu$ implicit in the rest of the proof.
By Cauchy-Schwarz inequality, the first term in \eqref{Eqn_proof_ratio_p_q_high_prob_diff_pq}, up to a factor of $C^2$, is at most
\begin{equation*}
    \|\vx_\delta\|_2 \|\vy_\delta\|_2 \le  \frac{\theta_7\|\vu\|_1 \|\vv\|_1}{n} \le \theta_3\theta_7(\ln n)^{1/8} = o\left(\frac{n}{(\ln n)^{1/2}}\right)
\end{equation*}
by \eqref{Eqn_def_Rstar_4} and \eqref{Eqn_def_Rstar_2}.%
The middle term in \eqref{Eqn_proof_ratio_p_q_high_prob_diff_pq}, up to a factor of $2C^2$, is at most%
\begin{equation*}
    \|\vx_\delta\|_1 \sumiton\mathbbm{1}_{(2/C,\infty)}((y_\delta)_j) (y_\delta)_j \le \theta_6\|\vu\|_1 \frac{\|\vv\|_1}{(\ln n)^{7/8}} \le \theta_3\theta_6\frac{n}{(\ln n)^{3/4}}
\end{equation*}
by \eqref{Eqn_def_Rstar_3}, \eqref{Eqn_def_Rstar_5}, and \eqref{Eqn_def_Rstar_2}.
Finally, the last term, up to a factor of $2C^2$, is upper bounded by
\begin{multline*}
    \|\vx_\delta\|_2^2 \|\vy_\delta\|_1 = \frac{\|\vx_\delta\|_2^2}{\|\vu\|_1^2} \frac{\|\vy_\delta\|_1}{\|\vv\|_1} \frac{1}{\|\vv\|_1} (\|\vu\|_1\|\vv\|_1)^2 \\
    \le \frac{\theta_7}{n} \cdot \theta_6 \cdot \frac{1}{\theta_2\ln n} \cdot \theta_3^2 n^2(\ln n)^{1/4} = \frac{\theta_7\theta_6\theta_3^2}{\theta_2}\frac{n}{(\ln n)^{3/4}}
\end{multline*}
due to \eqref{Eqn_def_Rstar_4}, \eqref{Eqn_def_Rstar_3}, \eqref{Eqn_def_Rstar_1}, and \eqref{Eqn_def_Rstar_2}. Putting these together gives the proposition.
\end{proof}

\section{Some concentration inequalities}

\subsection{Concentration for independent non-identically distributed exponential random variables}\label{appendix_concentration_for_nonid_exp_rvs}

\begin{restatable}{lemma}{lemmaWeightedExpChernoff}\label{Lemma_weighted_exp_chernoff}
Let $\mathbf{u}\in\R^n_+$ be a vector with $\|\mathbf{u}\|_1=n$ and let $\mathbf{Z}$ be a random vector with independent $\Exp(1)$ components. Then for any $t\in[0,1)$, we have
\begin{equation}
    \mathbb{P}(\mathbf{u}\cdot \mathbf{Z} \le tn) \le (te^{1-t})^n \prod_{i=1}^n u_i^{-1} \le (te)^n \prod_{i=1}^n u_i^{-1}.
\end{equation}
In fact, the upper bound given by the second inequality holds trivially when $t\ge 1$ and is invariant under simultaneous scaling of $u$ and $t$.

Further, when $1/K \le u_i\le K$ for some constant $K\ge 1$, we have
\begin{equation}
    \mathbb{P}(\mathbf{u}\cdot \mathbf{Z} \le tn) \ge e^{-\bO(n^{2/3})} (te^{1-Kt})^n \prod_{i=1}^n u_i^{-1}.
\end{equation}
For $t=o(1)$, this lower bound becomes
\[
    e^{-\bO(n^{2/3})+(1-K)tn} (te^{1-t})^n \prod_{i=1}^n u_i^{-1} = e^{o(n)} (te)^n \prod_{i=1}^n u_i^{-1},
\]
indicating that the upper bound is tight up to a factor of $e^{o(n)}$.
In particular, when $t=\bO(n^{-1/3})$, the gap is $e^{\bO(n^{2/3})}$.
\end{restatable}
\begin{proof}
First we establish the upper bound. Directly applying Chernoff's method on $\vu\cdot \vZ$, we have
\begin{equation}
    \mathbb{P}(\mathbf{u}\cdot \mathbf{Z} \le tn) \le \inf_{\lambda \ge 0} \frac{\E[\exp(-\lambda \mathbf{u}\cdot \mathbf{Z})]}{\exp(-\lambda tn)} = \inf_{\lambda\ge 0} e^{\lambda tn} \prod_{i=1}^n \frac{1}{1+\lambda u_i}.
\end{equation}
Taking $\lambda=1/t - 1$ (which is the minimizer when $\mathbf{u}=\mathbf{1}$) gives
\begin{equation}\label{Eqn_lemma_weighted_exp_chernoff_3}
    \mathbb{P}(\vu\cdot \vX \le tn) \le e^{n-tn} \prod_{i=1}^n \frac{t}{t+(1-t)u_i}.
\end{equation}
Notice that $\mathbf{u} \mapsto \sum_{i=1}^n \log u_i$ is a concave function on $\R_+^n$, and hence
\[
    \prod_{i=1}^n (t+(1-t)u_i) \ge \left(\prod_{i=1}^n u_i\right)^{1-t} \ge \prod_{i=1}^n u_i
\]
since $\prod_{i=1}^n u_i \le \big(n^{-1}\sumiton u_i\big)^n=1$.
Plugging the above inequality into \eqref{Eqn_lemma_weighted_exp_chernoff_3} gives the desired upper bound.

Now we establish the tightness of the bound under the additional assumption that $1/K \le u_i \le K$ for all $i\in[n]$.
Consider independent random variables $W_i\sim \Exp(u_i R/t)$ for $i=1,\cdots,n$ with $R=1+n^{-1/3}$, so that by Chebyshev's inequality
\[
    q_n := \mathbb{P}(\mathbf{u}\cdot \mathbf{W}\le tn) = \mathbb{P}_{T\sim \Gamma(n,1)}(T \le nR) \ge 1 - n^{-1/3}.
\] 
For convenience, we similarly write
\[
    p_n := \mathbb{P}(\vu\cdot \mathbf{Z}\le tn)
\]
and write the (joint) distributions of $\mathbf{Z}$ and $\mathbf{W}$ as $P_n = \Exp(1)^{\otimes n}$ and $Q_n = \bigotimes_{i=1}^n \Exp(u_i R/t)$, respectively. Applying the data processing inequality to the channel $\mathcal{C}$ that maps $\mathbf{\zeta}\in\R^n$ to $\mathbbm{1}\{\mathbf{u}\cdot \mathbf{\zeta} \le tn\}$ gives
\begin{multline}\label{Eqn_lemma_weighted_exp_chernoff_tightness_data_proc_ineq}
    D(Q_n\|P_n) \ge D(\mathcal{C}(Q_n) \| \mathcal{C}(P_n)) = q_n \log\frac{q_n}{p_n} + (1-q_n) \log\frac{1-q_n}{1-p_n} \\
    = - q_n \log p_n + (1-q_n) \log(1-p_n) + (q_n\log q_n + (1-q_n)\log(1-q_n)),
\end{multline}
where $D(\cdot\|\cdot)$ denotes the Kullback-Leibler (KL) divergence between two probability distributions. A direct computation gives
\begin{align}
    D(Q_n\|P_n) &= \sum_{i=1}^n \left(\frac{t}{Ru_i} - 1 - \log\frac{t}{Ru_i}\right) \nonumber\\
    &\le \sum_{i=1}^n \left(\frac{Kt}{R} - 1 - \log\frac{t}{Ru_i}\right) \nonumber\\
    &= -n + R^{-1} Ktn - n \log t + n\log R + \sum_{i=1}^n \log u_i \\
    &\le -n + Ktn - n \log t + n^{2/3} + \sum_{i=1}^n \log u_i.\label{Eqn_lemma_weighted_exp_chernoff_tightness_compute_kl}
\end{align}
Combining this with \eqref{Eqn_lemma_weighted_exp_chernoff_tightness_data_proc_ineq} and letting $n\to\infty$ gives
\begin{equation}
    -n + Ktn - n \log t + n^{2/3} + \sum_{i=1}^n \log u_i \ge - (1-O(n^{-1/3}))\log p_n + o(1),
\end{equation}
where we used the fact that
$\log(1-p_n)\to 0$ (due to our upper bound). Exponentiating both sides gives the desired lower bound for $p_n$.
\end{proof}

As a consequence, we have the following lemma.

\begin{restatable}{lemma}{lemmaWgtExpCondConcentration}\label{Lemma_wgt_exp_cond_concentration}
Let $\mathbf{u},\mathbf{v}\in\R^n_+$ be two vectors with bounded components, namely $\|\mathbf{u}\|_1=\|\mathbf{v}\|_1=n$ and $1/K \le u_i,v_i\le K$ for some fixed $K\ge 1$. For independent $Z_1,\cdots,Z_n\sim \Exp(1)$, we have
\begin{equation}\label{Eqn_main_Lemma_wgt_exp_cond_concentration}
    \mathbb{P}\left(\left|\frac{\mathbf{u}\cdot \mathbf{Z}}{t \mathbf{u}\cdot \mathbf{v}^{-1}} - 1\right| >\zeta \;\middle|\; \mathbf{v}\cdot \mathbf{Z} < tn\right) \le \exp(-\Theta(n\zeta^2))
\end{equation}
for $t=o(1)$ and and any fixed constant $\zeta>0$, where $\mathbf{v}^{-1}$ denotes the component-wise inverse of vector $\mathbf{v}$.
\end{restatable}

Notice that this result is invariant under simultaneous scaling of vector $\mathbf{u}$, $\mathbf{v}$, and $t$. Essentially, we only need $tn/\|\mathbf{v}\|_1 = o(1)$ and bounded ratios between among the entries of $\mathbf{u}$ and $\mathbf{v}$. Further, the result remains unchanged if $Z_i\sim\Exp(c_i)$ independently with $c_i$'s bounded on some $[1/K', K']$; the $c_i$'s can simply be absorbed into $\vu$ and $\vv$.

\begin{proof}

We first prove the concentration bound for the lower tail.

Writing
\begin{multline}
    \mathbb{P}(u\cdot x < (1-\zeta) tu\cdot v^{-1} | v\cdot x < tn) = \frac{\mathbb{P}(u\cdot x < (1-\zeta) tu\cdot v^{-1}, v\cdot x < tn)}{\mathbb{P}(v\cdot x < tn)} \\
    \le \frac{\mathbb{P}((\lambda u+(1-\lambda)v)\cdot x < (1-\zeta)\lambda t u\cdot v^{-1} + (1-\lambda)tn)}{\mathbb{P}(v\cdot x < tn)}
\end{multline}
for some $\lambda > 0$ to be determined later, the previous Lemma bounds the numerator by
\[
    t^n \left(1-\lambda + \frac{(1-\zeta)\lambda u\cdot v^{-1}}{n}\right)^n e^{n - (1-\zeta)\lambda t u\cdot v^{-1} - (1-\lambda)tn} \prod_{i=1}^n \frac{1}{\lambda u_i + (1-\lambda)v_i}.
\]

For lower bounding the denominator $\mathbb{P}(v\cdot x<tn)$, Lemma~\ref{Lemma_weighted_exp_chernoff} indicates that for $t=o(n)$, the denominator is well approximated by $t^n e^{n-tn} \prod_i v_i^{-1}$, up to an error of $e^{o(n)}$. Taking the ratio between the two quantities gives
\begin{equation}\label{Eqn_proof_wgt_exp_cond_concentration_lower_p_ratio_up_to_e^o(n)_1}
    \left(1-\lambda + \frac{(1-\zeta)\lambda u\cdot v^{-1}}{n}\right)^n e^{\lambda tn-(1-\zeta)\lambda t u\cdot v^{-1}} \prod_{i=1}^n \frac{1}{\lambda u_i v_i^{-1} + 1-\lambda}
\end{equation}

Focus on the quantity $\prod_{i=1}^n (\lambda u_i v_i^{-1} + 1-\lambda)$. We use the following claim for a bound on the gap between the arithmetic and geometric means.

\begin{claim}
For $z\in\R_+^n$ with $\bar{z} = n^{-1}\sum_{i=1}^n z_i$, the function $f:[0,1]\to\R$ given by $f(\alpha) = \sum_{i=1}^n \log(\bar{z} + \alpha (z_i-\bar{z}))$ is concave with a maximum at $\alpha=0$. (Indeed, the function $z\mapsto \sum_{i=1}^n \log z_i$ is concave on $\R_+^n$.) Hence,
\begin{equation}
    0\le f(0) - f(1) \le -f'(1) = -\sum_{i=1}^n \frac{z_i-\bar{z}}{z_i} = -n + \bar{z}\sum_{i=1}^n\frac{1}{z_i}.
\end{equation}
Exponentiating both sides gives
\begin{equation}
    \bar{z}^n \prod_{i=1}^n z_i^{-1} \le \exp\left( -n + \bar{z}\sum_{i=1}^n\frac{1}{z_i} \right).
\end{equation}
\end{claim}

Applying the above claim to the product in \eqref{Eqn_proof_wgt_exp_cond_concentration_lower_p_ratio_up_to_e^o(n)_1} with $z_i = \lambda u_i v_i^{-1} + 1-\lambda$, we obtain
\begin{equation}
    \prod_{i=1}^n \frac{1}{\lambda u_i v_i^{-1} + 1-\lambda} \le \left(1-\lambda + \frac{\lambda u\cdot v^{-1}}{n}\right)^{-n} \exp\left( -n +\sum_{i=1}^n\frac{1-\lambda + \lambda u\cdot v^{-1}/n}{1-\lambda + \lambda u_i v_i^{-1}}\right)
\end{equation}
Hence, the conditional probability of interest is upper bounded, up to $e^{o(n)}$, by the following expression
\begin{equation}
    \left(1-\lambda + \frac{(1-\zeta)\lambda u\cdot v^{-1}}{n}\right)^n e^{\lambda tn-(1-\zeta)\lambda t u\cdot v^{-1}} \left(1-\lambda + \frac{\lambda u\cdot v^{-1}}{n}\right)^{-n} \exp\left( -n +\sum_{i=1}^n\frac{1-\lambda + \lambda u\cdot v^{-1}/n}{1-\lambda + \lambda u_i v_i^{-1}}\right).
\end{equation}

Denote the negative logarithm of the $n$-th root of the quantity above by $\underline{\chi}(\lambda)$. That is,
\begin{equation}
    \mathbb{P}(u\cdot x < (1-\zeta) tu\cdot v^{-1} | v\cdot x < tn)
    \le \inf_{\lambda>0}e^{o(n) - n \underline{\chi}(\lambda)} = \exp\Big(o(n) - n \sup_{\lambda>0}\underline{\chi}(\lambda)\Big)
\end{equation}
for any $\lambda > 0$ with
\begin{multline}
    \underline{\chi}(\lambda) := -\log\left(1-\lambda + \frac{(1-\zeta)\lambda u\cdot v^{-1}}{n}\right) - \lambda t + \frac{(1-\zeta)\lambda t u\cdot v^{-1}}{n} +  \log\left(1-\lambda + \frac{\lambda u\cdot v^{-1}}{n}\right) \\
    + 1 - \frac{1}{n}\sum_{i=1}^n\frac{1-\lambda + \lambda u\cdot v^{-1}/n}{1-\lambda + \lambda u_i v_i^{-1}}.\label{Eqn_proof_wgt_exp_cond_concentration_lower_neg_log_prob_def}
\end{multline}
The $o(n)$ factor is of lower order, and it suffices to show that there exists some $\lambda$ such that $\underline{\chi}(\lambda)=\Theta(\zeta^2)$.
For $\lambda$ sufficiently small (e.g., $\lambda \le K^{-2}/2$, recalling that $u_i,v_i\in[1/K,K]$), we may approximate the logarithm function near its zero and obtain
\begin{equation}
    \log\left(1-\lambda + \frac{\lambda u\cdot v^{-1}}{n}\right) \ge \frac{\lambda u\cdot v^{-1}}{n} -\lambda - \left(\frac{\lambda u\cdot v^{-1}}{n} -\lambda\right)^2.
\end{equation}
Then the two log terms in \eqref{Eqn_proof_wgt_exp_cond_concentration_lower_neg_log_prob_def} combined can be bounded below by
\begin{equation}
    \lambda - \frac{(1-\zeta)\lambda u\cdot v^{-1}}{n} + \frac{\lambda u\cdot v^{-1}}{n} -\lambda - \left(\frac{\lambda u\cdot v^{-1}}{n} -\lambda\right)^2 = \zeta \lambda \frac{u\cdot v^{-1}}{n} - \lambda^2 \left(\frac{u\cdot v^{-1}}{n}-1\right)^2.
\end{equation}
With $u_i,v_i\in[1/K,K]$, a naive lower bound is the following
\begin{equation}\label{Eqn_proof_wgt_exp_cond_concentration_lower_neg_log_prob_lower_bound}
    \underline{\chi}(\lambda) \ge \zeta \lambda K^{-2} - \lambda^2 K^4 - \lambda t + (1-\zeta)\lambda t K^{-2}
    + 1 - \frac{(2-2\lambda+K^2\lambda+K^{-2}\lambda)^2}{4(1-\lambda+K^2\lambda)(1-\lambda+K^{-2}\lambda)},
\end{equation}
where the summation at the end of \eqref{Eqn_proof_wgt_exp_cond_concentration_lower_neg_log_prob_def} is bounded using Schweitzer's inequality \cite{schweitzer1914egy} for the ratio between arithmetic and harmonic means, stating
\[
    \frac{1}{n}\sum_{i=1}^n\frac{\bar{z}}{z_i} \le \frac{(a+b)^2}{4ab}
\]
for $z\in\R^n$ with bounded components $0< a\le z_i\le b$. Further, we observe
\begin{equation}
    1 - \frac{(2-2\lambda+K^2\lambda+K^{-2}\lambda)^2}{4(1-\lambda+K^2\lambda)(1-\lambda+K^{-2}\lambda)} = -\frac{(4K^2(K^{-2}-1)^2+(K-K^{-1})^4)}{4(1-\lambda+K^2\lambda)(1-\lambda+K^{-2}\lambda)}\lambda^2 \ge - 3K^4\lambda^2.
\end{equation}
Taking $\lambda = \zeta K^{-6}/8$ in \eqref{Eqn_proof_wgt_exp_cond_concentration_lower_neg_log_prob_lower_bound} yields
\begin{equation}
    \underline{\chi}\left(\frac{1}{8}\zeta K^{-6}\right) \ge \frac{1}{16}\zeta^2 K^{-8} - \frac{1}{8}t\zeta K^{-6} \ge \Theta(\zeta^2),
\end{equation}
hence finishing our proof for the lower tail.

The proof for the upper tail follows a similar structure.
Writing
\begin{multline*}
    \mathbb{P}(u\cdot x > (1+\zeta) tu\cdot v^{-1} | v\cdot x < tn) = \frac{\mathbb{P}(u\cdot x > (1+\zeta) tu\cdot v^{-1}, v\cdot x < tn)}{\mathbb{P}(v\cdot x < tn)} \\
    \le \frac{\mathbb{P}((-\lambda u+ (1+\lambda)v)\cdot x < -(1+\zeta)\lambda t u\cdot v^{-1} + (1+\lambda) tn)}{\mathbb{P}(v\cdot x < tn)}
\end{multline*}
for some $0 < \lambda < K^{-2}$ (so that $-\lambda u+ (1+\lambda)v\in\R_+^n$) to be determined later, Lemma~\ref{Lemma_weighted_exp_chernoff} and \ref{Lemma_weighted_exp_chernoff} together imply that the ratio is, up to $e^{o(n)}$,
\begin{equation}\label{Eqn_proof_wgt_exp_cond_concentration_upper_p_ratio_up_to_e^o(n)_1}
    \left(1+\lambda - \frac{(1+\zeta)\lambda u\cdot v^{-1}}{n}\right)^n e^{(1+\zeta)\lambda t u\cdot v^{-1} - \lambda tn} \prod_{i=1}^n \frac{1}{1+\lambda -\lambda u_i v_i^{-1}}
\end{equation}

As in the proof of lower tail bound, the product term can be bounded by
\begin{equation}
    \prod_{i=1}^n (1+\lambda-\lambda u_i v_i^{-1}) \le
    \left(1+\lambda - \frac{\lambda u\cdot v^{-1}}{n}\right)^{-n} \exp\left( -n +\sum_{i=1}^n\frac{1+\lambda - \lambda u\cdot v^{-1}/n}{1+\lambda - \lambda u_i v_i^{-1}}\right),
\end{equation}
giving an upper bound, again up to $e^{o(n)}$, of
\begin{equation}
    \left(1+\lambda - \frac{(1+\zeta)\lambda u\cdot v^{-1}}{n}\right)^n e^{(1+\zeta)\lambda t u\cdot v^{-1}-\lambda tn} \left(1+\lambda - \frac{\lambda u\cdot v^{-1}}{n}\right)^{-n} \exp\left( -n +\sum_{i=1}^n\frac{1+\lambda - \lambda u\cdot v^{-1}/n}{1+\lambda - \lambda u_i v_i^{-1}}\right)
\end{equation}
for the conditional probability of interest.

Denote the negative logarithm of the $n$-th root of the quantity above by $\overline{\chi}(\lambda)$. That is,
\begin{equation}
    \mathbb{P}(u\cdot x > (1+\zeta) tu\cdot v^{-1} | v\cdot x < tn)
    \le \inf_{0<\lambda<K^{-2}} e^{o(n) - n \overline{\chi}(\lambda)}
\end{equation}
for any $\lambda \in(0,K^{-2})$ with
\begin{multline}
    \overline{\chi}(\lambda) := -\log\left(1+\lambda - \frac{(1+\zeta)\lambda u\cdot v^{-1}}{n}\right) + \lambda t -\frac{(1+\zeta)\lambda t u\cdot v^{-1}}{n} + \log\left(1+\lambda - \frac{\lambda u\cdot v^{-1}}{n}\right) \\
    + 1 - \frac{1}{n}\sum_{i=1}^n\frac{1+\lambda - \lambda u\cdot v^{-1}/n}{1+\lambda - \lambda u_i v_i^{-1}}.
\end{multline}
Again, it suffices to prove that for some choice of $\lambda$ we have $\overline{\chi}(\lambda)=\Theta(\zeta^2)$.
With similar arithmetic as in the proof for the lower tail, we observe that for $\lambda$ sufficiently small (e.g., $\lambda \le K^{-2}/2$)
\begin{equation}\label{Eqn_proof_wgt_exp_cond_concentration_upper_neg_log_prob_lower_bound}
    \overline{\chi}(\lambda) \ge \zeta \lambda K^{-2} + \lambda t - (1+\zeta)\lambda t K^2 - 4 \lambda^2 K^4.
\end{equation}
Again, taking $\lambda = \zeta K^{-6}/8$ in \eqref{Eqn_proof_wgt_exp_cond_concentration_upper_neg_log_prob_lower_bound} gives the desired lower bound of $\Theta(\zeta^2)$ for $\sup_{0<\lambda<K^{-2}}\overline{\chi}(\lambda)$ and thus finishes our proof.
\end{proof}

\subsection{A generalized DKW inequality for independent and nearly identically distributed random variables}

\begin{lemma}\label{Lemma_dkw_non_identical}
Let $X_i$, $i=1,\cdots,n$ be independent random variables each with (non-identical) distribution function $G_i$, and assume that there exists a constant $\delta>0$ and a distribution $F$ such that $\|G_i-F\|_\infty \leq \delta$ uniformly across all $i=1,\cdots,n$. Let $\hat{G}$ be the empirical distribution function of $\{X_i\}_{i=1}^n$. Then
\begin{equation}\label{Eqn_dkw_non_identical_lemma}
    \mathbb{P}(\|\hat{G}-F\|_\infty > 2\delta + \epsilon) < 4\exp(-2n\epsilon^2/9).
\end{equation}
\end{lemma}
\begin{proof}
Let $U_i = G_i(X_i)$ so that $U_1,\cdots,U_n$ are i.i.d. uniform on $[0,1]$, and denote their empirical distribution function by $\hat{J}$. Let $Y_i=F^{-1}(U_i)=F^{-1}(G_i(X_i))$ so that $Y_1,\cdots,Y_n$ are i.i.d. each with distribution function $F$, and denote their empirical distribution function by $\hat{F}$. Notice that
\begin{align*}
    \|\hat{G}-F\|_\infty &= \sup_{x\in\mathbb{R}} \left| n^{-1} \sum_{i=1}^n I_{(-\infty,x)}(X_i) - F(x)\right|\\
    &= \sup_{x\in\mathbb{R}} \left| n^{-1} \sum_{i=1}^n I_{(-\infty,x)}(Y_i) - F(x) + n^{-1} \sum_{i=1}^n \left(I_{(-\infty,x)}(Y_i) - I_{(-\infty,x)}(X_i)\right) \right|\\
    &\leq \sup_{x\in\mathbb{R}} \left| n^{-1} \sum_{i=1}^n I_{(-\infty,x)}(Y_i) - F(x)\right| + \sup_{x\in\mathbb{R}}\left(n^{-1} \sum_{i=1}^n \left|I_{(-\infty,x)}(Y_i) - I_{(-\infty,x)}(X_i) \right|\right)\\
    &= \|\hat{F}-F\|_\infty + \sup_{x\in\mathbb{R}}A(x).
\end{align*}
From the classic result of DKW inequality \cite{dvoretzky1956asymptotic} applied to $\hat{F}$ and $F$, we know that
\begin{equation}\label{Eqn_proof_dkw_non_identical_bound_1}
    \mathbb{P}(\|\hat{F}-F\|_\infty > \epsilon/3) < 2\exp(-2n\epsilon^2/9).
\end{equation}
For the second supremum of in the sum above, we now consider $U_i$, $i=1,\cdots,n$ as the underlying random variables. Each term in the summation in $A$ contributes 1 to the sum if and only if
\[F^{-1}(U_i)=Y_i<x\leq X_i=G_i^{-1}(U_i) \;\text{ or }\; G_i^{-1}(U_i)=X_i<x\leq Y_i=F^{-1}(U_i),\]
or alternatively
\[F(x)\wedge G_i(x) \leq U_i \leq F(x)\vee G_i(x),\footnote{We may safely ignore the case where the two sides are equal, as it happens with probability zero.}\]
where $\wedge$ and $\vee$ denote the operators of taking the minimum and maximum, respectively. Hence,
\begin{align*}
    A(x) &= n^{-1} \sum_{i=1}^n \left|I_{(-\infty,x)}(Y_i) - I_{(-\infty,x)}(X_i) \right|\\
    &=  n^{-1} \sum_{i=1}^n I_{(F(x)\wedge G_i(x),F(x)\vee G_i(x))}(U_i)\\
    &\leq  n^{-1} \sum_{i=1}^n I_{(\bigwedge_j G_j(x)\wedge F(x),\bigvee_j G_j(x)\vee F(x))}(U_i)\\
    &= \hat{J}(M(x)) - \hat{J}(m(x)),
\end{align*}
where $M$ and $m$ denote the maximum and minimum across $F$ and $G_i$, $i=1,\cdots,n$, respectively. By our assumption that $\|G_i-F\|_\infty \leq \delta$ across all $i$, we have that
\[
    0\leq M(x) - m(x)\leq 2\delta
\]
for all $x\in\mathbb{R}$. Noticing that the true distribution function $J$ of $U_i$, $i=1,\cdots,n$ is the identity function on $[0,1]$, we have
\begin{align*}
    A(x) &= \hat{J}(M(x)) - \hat{J}(m(x))\\
    &\leq \left|\hat{J}(M(x)) - J(M(x))\right| + \left|\hat{J}(m(x)) - J(m(x))\right| + \left|J(M(x)) - J(m(x))\right|\\
    &\leq 2\|\hat{J}-J\|_\infty + 2\delta
\end{align*}
on $\mathbb{R}$ uniformly. Therefore, applying DKW inequality again, we see that
\begin{equation}\label{Eqn_proof_dkw_non_identical_bound_2}
    \mathbb{P}(\sup A > 2\delta + 2\epsilon/3) \leq \mathbb{P}(\|\hat{J}-J\|_\infty > \epsilon/3) < 2\exp(-2n\epsilon^2/9).
\end{equation}
Combining \eqref{Eqn_proof_dkw_non_identical_bound_1} and \eqref{Eqn_proof_dkw_non_identical_bound_2} yields the desired bound in \eqref{Eqn_dkw_non_identical_lemma}.
\end{proof}

\section{Additional proofs}\label{appendix_extra_proofs}

\subsection{Proof of Corollary~\ref{Cor_subexp_num_stable_match}}

In this section, we prove Corollary~\ref{Cor_subexp_num_stable_match}, which is restated below for convenience. We will assume Proposition~\ref{Prop_EqXY_bound}, whose proof is deferred to Appendix~\ref{Append_proof_prop_EqXY}.

\corSubExpNumOfStableMatch*
\begin{proof}
The last part is simply Corollary~\ref{Cor_Rstar_likely}.

For the first part, observe that for each $\cM'\subseteq\cM$ and $\cW'\subseteq\cW$ with $|\cM'|=|\cW'|=n-\floor{\delta n}$ and partial matching $\mu'$ between $\cM'$ and $\cW'$,
\begin{multline}
    \Pp(\mu'\text{ is stable and satisfies }\mathcal{R}^*) \le e^{o(n)} \E[q(\vX_{\cM'},\vY_{\cW'}) \cdot \mathbbm{1}_{\mathcal{R}^*}(\vX_{\cM'},\vY_{\cW'})] \le \\
    e^{o(n)} \E[q(\vX_{\cM'},\vY_{\cW'}) \cdot \mathbbm{1}_{\mathcal{R}_2}(\vX_{\cM'}, \vY_{\cW'}) \cdot\mathbbm{1}_{\mathcal{R}_1}(\vY_{\cW'})] \le e^{o_\delta(n)} \frac{(\delta n)!}{n!}\prod_{i\in\cM'} a_{i,\mu'(i)} b_{\mu'(i),i}
\end{multline}
by Proposition~\ref{Prop_EqXY_bound}. Summing over $\cM'$, $\cW'$, and $\mu'$ bounds the expected number of such stable partial matchings above by
\begin{align}
    \E[N_\delta] &\le \sum_{\substack{\cM'\subseteq\cM,\cW'\subseteq\cW\\|\cM'|=|\cW'|=n-\floor{\delta n}}}\sum_{\substack{\mu':\cM'\to\cW'\\\text{bijection}}} e^{o_\delta(n)} \frac{(\delta n)!}{n!}\prod_{i\in\cM'} a_{i,\mu'(i)} b_{\mu'(i),i} \nonumber \\
    &\labelrel={Step1_num_stable} \frac{1}{\floor{\delta n}!} \sum_{\substack{\mu:\cM\to\cW\\\text{bijection}}} \sum_{\substack{\cM'\subseteq\cM\\|\cM'|=n-\floor{\delta n}}} e^{o_\delta(n)} \frac{(\delta n)!}{n!}\prod_{i\in\cM'} a_{i,\mu(i)} b_{\mu(i),i} \nonumber \\
    &\labelrel\le{Step2_num_stable} \sum_{\substack{\mu:\cM\to\cW\\\text{bijection}}} \sum_{\substack{\cM'\subseteq\cM\\|\cM'|=n-\floor{\delta n}}} e^{o_\delta(n)} \frac{1}{n!}\cdot C^{2\floor{\delta n}} \prod_{i\in\cM} a_{i,\mu(i)} b_{\mu(i),i} \nonumber \\
    &\labelrel\le{Step3_num_stable} e^{o_\delta(n)} \binom{n}{\floor{\delta n}} \cdot \frac{1}{n!} \Perm(\vA\circ \vB^\top) \nonumber \\
    &\labelrel\le{Step4_num_stable} e^{o_\delta(n)}, \nonumber
\end{align}
where in \eqref{Step1_num_stable} we use an alternative counting of partial matchings by counting sub-matchings of size $n-\floor{\delta n}$ in full matchings and then deduplicate by a factor of $\floor{\delta n}!$; in \eqref{Step2_num_stable} we use the boundedness assumption on the components of $\vA$ and $\vB$; in \eqref{Step3_num_stable} we merge $C^{2\floor{\delta n}}$ into $e^{o_\delta(n)}$; and finally in \eqref{Step4_num_stable} we merge $\binom{n}{\floor{\delta n}}=\exp(h(\delta) n + o(n))$ into $e^{o_\delta(n)}$ and bound the permanent term by $n^n \Perm(\vM) \le \Theta(n!)$ using the moderate deviation property of $\vM$ \citep[Sec.~3]{mccullagh2014asymptotic}.
\end{proof}

\subsection{Proof of Lemma~\ref{Prop_eigenvec_of_M_high_prob}}\label{Append_proof_prop_eigenvec_of_M}

In this section, we prove Lemma~\ref{Prop_eigenvec_of_M_high_prob}, which is restated below for convenience.

\PropEigenVecOfMHighProf*

To prepare for the proof of Lemma~\ref{Prop_eigenvec_of_M_high_prob}, let us denote the expectation in \eqref{Eqn_Prop_eigenvec_of_M_Expectation_is_small} by $E$, and express it as
\begin{align}
    E &= \int_{0}^\infty \Pp\big(q(\vX_{\cM'}, \vY_{\cW'}) \cdot \mathbbm{1}_{\mathcal{R}\backslash\Oeigz}(\vX_{\cM'}, \vY_{\cW'}) > s\big) \dd s \nonumber\\
    &= \int_0^1 \Pp\big(\exp(-n\vX_{\cM'}^\top \vM \vY_{\cW'}) > s, (\vX_{\cM'}, \vY_{\cW'}) \in \mathcal{R} \backslash \Oeigz\big) \dd s \nonumber \\
    &= \int_0^\infty \Pp\big(\vX_{\cM'}^\top \vM \vY_{\cW'} < t, (\vX_{\cM'}, \vY_{\cW'}) \in \mathcal{R}_2 \backslash \Oeigz, \vX_{\cM'}\in\mathcal{R}_1, \vY_{\cW'}\in\mathcal{R}_1\big) \cdot n e^{-nt}\dd t \nonumber \\
    &= \int_0^\infty \Pp\big(\vX_{\cM'}^\top \vM \vY_{\cW'} < \bar{t}, (\vX_{\cM'}, \vY_{\cW'}) \notin \Oeigz, \vX_{\cM'}\in\mathcal{R}_1, \vY_{\cW'}\in\mathcal{R}_1\big) \cdot n e^{-nt}\dd t,
\end{align}
where $\bar{t} := t \wedge (c_2(\log n)^{1/8})$.
If we can find two families of regions $\Omega_1(\zeta;s),\Omega_2(\zeta;s)\subseteq\R_+^n\times\R_+^n$ such that $\Oeigz \supseteq \Omega_1(\Theta(\zeta);s)\cap\Omega_2(\Theta(\zeta);s)$ for all $0<s<c_2(\log n)^{1/8}$, by union bound and relaxing the requirement that $\vX_{\cM'}$ (resp. $\vY_{\cW'}$) is in $\mathcal{R}_1$, we will obtain
\begin{align}
    E
    &\le \int_0^\infty \Pp\big(\vX_{\cM'}^\top \vM \vY_{\cW'} < \bar{t}, (\vX_{\cM'}, \vY_{\cW'}) \notin \Omega_1(\Theta(\zeta);\bar{t}), \vX_{\cM'}\in\mathcal{R}_1, \vY_{\cW'}\in\mathcal{R}_1\big) \cdot n e^{-nt}\dd t \nonumber \\
    &\qquad + \int_0^\infty \Pp\big(\vX_{\cM'}^\top \vM \vY_{\cW'} < \bar{t}, (\vX_{\cM'}, \vY_{\cW'}) \notin \Omega_2(\Theta(\zeta);\bar{t}), \vX_{\cM'}\in\mathcal{R}_1, \vY_{\cW'}\in\mathcal{R}_1\big) \cdot n e^{-nt}\dd t \nonumber\\
    &\le \int_0^\infty \Pp\big(\vX_{\cM'}^\top \vM \vY_{\cW'} < \bar{t}, (\vX_{\cM'}, \vY_{\cW'}) \notin \Omega_1(\Theta(\zeta);\bar{t}), \vX_{\cM'}\in\mathcal{R}_1\big) \cdot n e^{-nt}\dd t \nonumber \\
    &\qquad + \int_0^\infty \Pp\big(\vX_{\cM'}^\top \vM \vY_{\cW'} < \bar{t}, (\vX_{\cM'}, \vY_{\cW'}) \notin \Omega_2(\Theta(\zeta);\bar{t}), \vY_{\cW'}\in\mathcal{R}_1\big) \cdot n e^{-nt}\dd t. \nonumber
\end{align}
Rewriting the probabilities through conditioning and further relaxing the requirement gives
\begin{align}
    E
    &\le \int_0^\infty \Pp\big((\vX_{\cM'}, \vY_{\cW'}) \notin \Omega_1(\Theta(\zeta);\bar{t}) \big| \vX_{\cM'}^\top \vM \vY_{\cW'} < \bar{t}, \vX_{\cM'}\in\mathcal{R}_1\big) \nonumber \\
    &\qquad\qquad \cdot \Pp\big(\vX_{\cM'}^\top \vM \vY_{\cW'} < \bar{t}, \vX_{\cM'}\in\mathcal{R}_1\big) \cdot n e^{-nt}\dd t  \nonumber \\
    &\qquad + \int_0^\infty \Pp\big((\vX_{\cM'}, \vY_{\cW'}) \notin \Omega_2(\Theta(\zeta);\bar{t}) \big| \vX_{\cM'}^\top \vM \vY_{\cW'} < \bar{t}, \vY_{\cW'}\in\mathcal{R}_1\big) \nonumber \\
    &\qquad\qquad \cdot \Pp\big(\vX_{\cM'}^\top \vM \vY_{\cW'} < \bar{t}, \vY_{\cW'}\in\mathcal{R}_1\big) \cdot n e^{-nt}\dd t.  %
    \label{Eqn_decompose_E_exp_qxy_on_Omega_eig_zeta}
\end{align}
Due to the symmetry between the two integrals, it then suffices to bound one of the two integrals (e.g., the latter) by showing
\begin{equation}\label{Eqn_unif_concentrat_MX_cond_on_XE}
    \sup_{0 < t < c_2(\log n)^{1/8}} \Pp\big((\vX_{\cM'}, \vY_{\cW'}) \notin \Omega_2(\Theta(\zeta);t) \big| \vX_{\cM'}^\top \vM \vY_{\cW'} < t, \vX_{\cM'}\in\mathcal{R}_1\big) \le \exp(-\Theta(\zeta^2 n))
\end{equation}
and 
\begin{equation}\label{Eqn_EexpXMY_over_Y_in_R1}
    \int_0^\infty \Pp\big(\vX_{\cM'}^\top \vM \vY_{\cW'} < \bar{t}, \vY_{\cW'}\in\mathcal{R}_1\big) \cdot n e^{-nt}\dd t
    \le e^{o(n)+o_\delta(n)} \frac{(\delta n)!}{n!} \prod_{i\in\cM'} a_{i,\mu'(i)} b_{\mu'(i),i},
\end{equation}
from which the desired upper bound immediately follows. Recognizing that
\begin{equation*}
    \int_0^\infty \Pp\big(\vX_{\cM'}^\top \vM \vY_{\cW'} < \bar{t}, \vY_{\cW'}\in\mathcal{R}_1\big) \cdot n e^{-nt}\dd t = \E\big[q(\vX_{\cM'},\vY_{\cW'})\cdot\mathbbm{1}_{\mathcal{R}_2}(\vX_{\cM'},\vY_{\cW'})\cdot\mathbbm{1}_{\mathcal{R}_1}(\vY_{\cW'})\big],
\end{equation*}
we reduce \eqref{Eqn_EexpXMY_over_Y_in_R1} to Proposition~\ref{Prop_EqXY_bound}. Our road map is to first find the desirable choices for $\Omega_1$ and $\Omega_2$ and establish \eqref{Eqn_unif_concentrat_MX_cond_on_XE}, and then prove Proposition~\ref{Prop_EqXY_bound} in Appendix~\ref{Append_proof_prop_EqXY}. Note that Proposition~\ref{Prop_EqXY_bound} is in fact independent of our choice of $\Omega_1$ and $\Omega_2$, but we will develop useful intermediate results to prepare for its proof.

Concretely, we consider events $\Omega_1$ and $\Omega_2$ as follows:
\begin{equation}
    \Omega_1(\zeta;t) := \left\{(\vx,\vy)\in\R_+^n\times\R_+^n : \max_{i\in[n]}\left|(1-\delta)\frac{\vM_{i,\cdot}\cdot \vy}{t \vM_{i,\cdot} \cdot (\vM^\top\vx)^{-1}_{\cW'}} - 1\right| >\zeta\right\},
\end{equation}
\begin{equation}
    \Omega_2(\zeta;t) := \left\{(\vx,\vy)\in\R_+^n\times\R_+^n : \max_{j\in[n]}\left|(1-\delta)\frac{\vM_{\cdot,j}\cdot \vx}{t \vM_{\cdot,j} \cdot (\vM\vy)^{-1}_{\cM'}} - 1\right| >\zeta\right\},
\end{equation}
where $\vM_{i,\cdot}$ and $\vM_{\cdot,j}$ denote the $i$-th row and the $j$-th column of $\vM$, respectively; inverse is applied coordinate-wise on vectors; and $\vv_{S}$ denotes the $n$-dimensional vector obtained by zeroing out the $i$-th component $v_i$ of $\vv\in\R_+^n$ for all $i\in[n]\backslash S$ (with this operation performed after coordinate-wise inverse). We first verify the following lemma.

\begin{lemma}\label{lemma_Omega1_and_Omega2_suggests_Oeig}
There exist absolute constants $\zeta_0,\delta_0 > 0$ and $k_1, k_2 > 0$ such that for all $\zeta\in(0,\zeta_0)$, $\delta\in(0,\delta_0)$, and $t>0$ we have
\begin{equation}
    \Omega_1(\zeta;t)\cap\Omega_2(\zeta;t) \subseteq \Oeig(k_1\delta + k_2\zeta).
\end{equation}
\end{lemma}
\begin{proof}
Let $\vd = \vM^\top \vx$ and $\ve = \vM\vy$.
Under the event that $(\vx,\vy)\in\Omega_1(\zeta;t)\cap\Omega_2(\zeta;t)$,
we have
\begin{align}
    \frac{1}{d_j} &= \frac{1}{\vM_{\cdot,j}\cdot \vx} \labelrel\le{Step_use_Omega2_Lemma_Omega12_gives_Oeig} \frac{1-\delta}{(1-\zeta) t\vM_{\cdot,j} \cdot \ve^{-1}_{\cM'}} \labelrel\le{Step_bdd_comp_Lemma_Omega12_gives_Oeig} \frac{1-\delta}{(1-\zeta) t(1+2C^2\delta)\vM_{\cdot,j} \cdot \ve^{-1}} \nonumber \\
    &\labelrel\le{Step_Jensen_Lemma_Omega12_gives_Oeig} \frac{1-\delta}{(1-\zeta) t (1+2C^2\delta)} \sumiton m_{ij} e_i = \frac{1-\delta}{(1-\zeta) t (1+2C^2\delta)} \sumiton m_{ij}\vM_{i,\cdot}\cdot \vy \nonumber \\
    &\labelrel\le{Step_use_Omega1_Lemma_Omega12_gives_Oeig} \frac{1-\delta}{(1-\zeta) t (1+2C^2\delta)} \sumiton m_{ij} \left( \frac{1+\zeta}{1-\delta} t \vM_{i,\cdot} \cdot \vd^{-1}_{\cW'}\right) \le \frac{1+\zeta}{(1-\zeta)(1+2C^2\delta)} (\vM^\top\vM\vd^{-1})_j,
\end{align}
where \eqref{Step_use_Omega2_Lemma_Omega12_gives_Oeig} uses the definition of $\Omega_2(\zeta;t)$; \eqref{Step_bdd_comp_Lemma_Omega12_gives_Oeig} uses the fact that $\vM$ and $\ve$ both have bounded ratios (at most $C$) among their entries and assumed $\delta < 1/2$; \eqref{Step_Jensen_Lemma_Omega12_gives_Oeig} is due to Jensen's inequality (or equivalently, harmonic-mean-arithmetic-mean inequality); and \eqref{Step_use_Omega1_Lemma_Omega12_gives_Oeig} uses the definition of $\Omega_1(\zeta;t)$.

Recall our assumption that $\vM$ has entries bounded on $[1/(Cn), C/n]$. It is straightforward to verify that for any vector $\vv\in\R_+^n$ with $\bar{v}=\frac{1}{n}\sumiton v_i$, we have $\max_{i\in[n]} (\vM\vv)_i - \bar{v} \le \max_{i\in[n]} v_i - \frac{1}{C n} \cdot n(\max_{i\in[n]} v_i - \bar{v}) - \bar{v} = (1-C^{-1}) (\max_{i\in[n]} v_i - \bar{v})$. In the case of $\vv = \vd^{-1}$, this implies that
\begin{multline}\label{Eqn_chain_bound_d_i_star_and_harmonic_mean}
    (1-C^{-1})^2 \big(d_{i^*}^{-1} - \bar{d}_{(H)}^{-1}\big) \ge \max_{i\in[n]}(\vM^\top\vM\vd^{-1})_i - \bar{d}_{(H)}^{-1} \\
    \ge (\vM^\top\vM\vd^{-1})_{i^*} - \bar{d}_{(H)}^{-1} \ge (1+2C^2\delta) \frac{1-\zeta}{1+\zeta}d_{i^*}^{-1} - \bar{d}_{(H)}^{-1},
\end{multline}
where $i^*=\argmin_{i\in[n]} d_i$ and $\bar{d}_{H} = \big(n^{-1}\sumiton d_i^{-1}\big)^{-1}$ is the harmonic mean of $d_1,\ldots,d_n$. Solving \eqref{Eqn_chain_bound_d_i_star_and_harmonic_mean} gives
\begin{equation}
    \frac{ d_{i^*}^{-1} - \bar{d}_{(H)}^{-1} }{ \bar{d}_{(H)}^{-1} } \le \Theta(\delta) + \frac{2\zeta}{1 - \zeta - (1-C^{-2})^2(1+\zeta)} \le \Theta(\delta + \zeta)
\end{equation}
with hidden constants independent of $\delta$ and $\zeta$, granted that $\zeta$ is sufficiently small. Hence, for all but $\sqrt{\delta+\zeta} n$ indices $i\in[n]$, we have $1-\Theta(\sqrt{\delta+\zeta}) \le \frac{\bar{d}_{(H)}}{d_i} \le 1 + \Theta(\delta+\zeta)$, implying that $(\vx,\vy)\in\Oeig(\Theta(\delta+\zeta))$.
\end{proof}

Let $\vD = \vM^\top \vX_{\cM'}$ and $\vE = \vM\vY_{\cW'}$. Note that $\vD$ and $\vE$ both have bounded ratios among their components due to the bounded ratio assumption on $\vM$, and in addition $\|\vD\|_1=\|\vX_{\cM'}\|_1$ and $\|\vE\|_1=\|\vY_{\cW'}\|_1$. By Lemma~\ref{Lemma_wgt_exp_cond_concentration}, whenever $t\le c_2(\log n)^{1/8}$, we have for each column $\vM_{\cdot,j}$ of $\vM$, $j=1,\ldots,n$,
\begin{equation}
    \mathbb{P}\left(\left|(1-\delta)\frac{\vM_{\cdot,j}\cdot \vX_{\cM'}}{t \vM_{\cdot,j} \cdot \vE^{-1}_{\cM'}} - 1\right| >\zeta \;\middle|\; \vX_{\cM'} \cdot \vE < t, \|\vE\|_1\ge \underline{c}_1 \log n\right) \le \exp(-\Theta(n\zeta^2)),
\end{equation}
where we note that the effective dimension of $\vX_{\cM'}$ is $n-\floor{\delta n}$ instead of $n$.
By a union bound over $j\in[n]$, this gives
\begin{equation}
    \mathbb{P}\left(\max_{j\in[n]}\left|(1-\delta)\frac{\vM_{\cdot,j}\cdot \vX_{\cM'}}{t \vM_{\cdot,j} \cdot \vE^{-1}_{\cM'}} - 1\right| >\zeta, \|\vE\|_1\ge \underline{c}_1 \log n \;\middle|\; \vX_{\cM'} \cdot \vE < t \right) \le \exp(-\Theta(n\zeta^2)),
\end{equation}
which is simply \eqref{Eqn_unif_concentrat_MX_cond_on_XE}.

\subsection{Proof of Corollary~\ref{Cor_no_stable_outside_Oeigz}}\label{Append_proof_no_stable_outside_oeigz}

We now prove Corollary~\ref{Cor_no_stable_outside_Oeigz}, restated below.
\CorNoStableOutsideOeigz*

\begin{proof}
    Summing over all partial matchings with size $n-\floor{\delta n}$ gives
    \begin{multline}\label{Eqn_sum_expectation_Omega_zeta}
        \sum_{\substack{\cM'\subseteq\cM,\cW'\subseteq\cW\\|\cM'|=|\cW'|=n-\floor{\delta n}}}\sum_{\substack{\mu':\cM'\to\cW'\\\text{ bijection}}} \E\big[q(\vX_{\cM'}(\mu'), \vY_{\cW'}(\mu')) \cdot \mathbbm{1}_{\mathcal{R}\backslash\Oeigz}(\vX_{\cM'}, \vY_{\cW'})\big] \\
        \le \exp(o_\delta(n)-\Theta(\zeta^2 n)) \cdot \frac{(\delta n)!}{n!}  \sum_{\substack{\cM'\subseteq\cM,\cW'\subseteq\cW\\|\cM'|=|\cW'|=n-\floor{\delta n}}}\sum_{\substack{\mu':\cM'\to\cW'\\\text{bijection}}} \prod_{i\in\cM'} (n m_{i,\mu'(i)}).
    \end{multline}
    To bound the summation, %
    notice that
    \begin{align}
        \sum_{\substack{\cM'\subseteq\cM,\cW'\subseteq\cW\\|\cM'|=|\cW'|=n-\floor{\delta n}}}\sum_{\substack{\mu':\cM'\to\cW'\\\text{bijection}}} \prod_{i\in\cM'} (n m_{i,\mu'(i)})
        &= \frac{1}{(\floor{\delta n})!} \sum_{\substack{\mu:\cM\to\cW\\\text{bijection}}} \sum_{\substack{\cM'\subseteq\cM\\|\cM'|=n-\floor{\delta n}}} \prod_{i\in\cM'} (n m_{i,\mu'(i)}) \nonumber\\
        &\le \frac{1}{(\floor{\delta n})!} \sum_{\substack{\mu:\cM\to\cW\\\text{bijection}}} \binom{n}{\floor{\delta n}} C^{\floor{\delta n}} \prod_{i\in\cM} (n m_{i,\mu(i)}) \nonumber\\
        &= \frac{e^{o_\delta(n)}}{(\delta n)!} n^n \Perm(\vM).\label{Eqn_sum_expectation_Omega_zeta_summation_bound}
    \end{align}
    Under the assumption that the bistochastic matrix $\vM$ is of moderate deviation (cf. \cite[Section~3]{mccullagh2014asymptotic}), we know that $n^n \Perm(\vM) = O( n! )$. Hence, the quantity in \eqref{Eqn_sum_expectation_Omega_zeta} is bounded by $\exp(o_\delta(n)-\Theta(\zeta^2 n))$. Invoking Lemma~\ref{Lemma_reduction_to_q} finishes the proof.
\end{proof}

\subsection{Proof of Proposition~\ref{Prop_EqXY_bound}}\label{Append_proof_prop_EqXY}

In this section, we present the proof of Proposition~\ref{Prop_EqXY_bound}, restated below.

\propEqXYBound*

Denote the target expectation by $E$ and express it as an integral of tail probability
\begin{equation}\label{Eqn_proof_prop_6_6_goal}
    E = \int_0^\infty \Pp(\vX_{\cM'}^\top\vM\vY_{\cW'} < \bar{t}, \vY_{\cW'} \in \mathcal{R}_1) \cdot ne^{-nt} \dd t,
\end{equation}
where $\bar{t} = t \wedge (c_2(\log n)^{1/8})$. It suffices to upper bound probabilities of the form $\Pp(\vX_{\cM'}^\top\vM\vY_{\cW'} < \bar{t}, \vY_{\cW'} \in \mathcal{R}_1)$ for all $t\in(0,c_2(\log n)^{1/8}$. We will go one step further and prove a stronger result by relaxing the $\vY_{\cW'} \in \mathcal{R}_1$ condition, which will eventually translate to a bound on $\E[q(\vX_{\cM'},\vY_{\cW'}) \cdot \mathbbm{1}_{\mathcal{R}_2}(\vX_{\cM'}, \vY_{\cW'})]$.

\begin{lemma}\label{Lemma_x_y_not_both_small_cond_on_XMY_and_R1}
There exists a positive constant $\beta$ such that, for any $t\in(0, c_2(\log n)^{1/8})$,
\begin{equation}
    \Pp\left(\|\vX_{\cM'}\|_1 \le \beta \sqrt{tn}, \|\vY_{\cW'}\|_1 \le \beta \sqrt{tn} \middle| \vX_{\cM'}^\top\vM \vY_{\cW'} < t\right) \le 0.1
\end{equation}
for $n$ sufficiently large.
\end{lemma}
\begin{proof}
Let $p$ denote the target probability. We have
\begin{align}
    p &= \Pp\left(\|\vX_{\cM'}\|_1 \le \beta \sqrt{tn} \middle| \|\vY_{\cW'}\|_1 \le \beta \sqrt{tn}, \vX_{\cM'}^\top\vM \vY_{\cW'} < t\right) \nonumber \\
    &\qquad\cdot \Pp\left(\|\vY_{\cW'}\|_1 \le \beta \sqrt{tn} \middle| \vX_{\cM'}^\top\vM \vY_{\cW'} < t\right) \nonumber \\
    &\le \Pp\left(\|\vX_{\cM'}\|_1 \le \beta \sqrt{tn} \middle| \|\vY_{\cW'}\|_1 \le \beta \sqrt{tn}, \vX_{\cM'}^\top\vM \vY_{\cW'} < t\right) \nonumber \\
    &= \frac{\Pp\left(\|\vX_{\cM'}\|_1 \le \beta \sqrt{tn}, \vX_{\cM'}^\top\vM \vY_{\cW'} < t \middle| \|\vY_{\cW'}\|_1 \le \beta \sqrt{tn}\right)}{\Pp\left(\vX_{\cM'}^\top\vM \vY_{\cW'} < t \middle| \|\vY_{\cW'}\|_1 \le \beta \sqrt{tn}\right)} \nonumber \\
    &\le \frac{\Pp\left(\|\vX_{\cM'}\|_1 \le \beta \sqrt{tn} \middle| \|\vY_{\cW'}\|_1 \le \beta \sqrt{tn}\right)}{\Pp\left(\vX_{\cM'}^\top\vM \vY_{\cW'} < t \middle| \|\vY_{\cW'}\|_1 \le \beta \sqrt{tn}\right)} \nonumber \\
    &\le \frac{\Pp\left(\|\vX_{\cM'}\|_1 \le \beta \sqrt{tn}\right)}{\Pp\left(\|\vX_{\cM'}\|_1 \le \sqrt{tn}/(C\beta)\right)} = \Pp\left(\|\vX_{\cM'}\|_1 \le \beta \sqrt{tn} \middle| \|\vX_{\cM'}\|_1 \le (C\beta)^{-1}\sqrt{tn}\right),
\end{align}
where the last inequality follows from the independence between $\vX$ and $\vY$ and the fact that $n\vX_{\cM'}^\top\vM \vY_{\cW'}\le C\|\vX_{\cM'}\|_1 \|\vY_{\cW'}\|_1$.
By choosing $\beta = (2C)^{-1/2}$, the upper bound becomes
\begin{equation*}
    \Pp\left(\|\vX_{\cM'}\|_1 \le \beta \sqrt{tn} \middle| \|\vX_{\cM'}\|_1 \le 2\beta\sqrt{tn}\right).
\end{equation*}
A direct invocation of Lemma~\ref{Lemma_wgt_exp_cond_concentration} implies an $\exp(-\Theta(n))$ upper bound for this probability.
\end{proof}

\begin{lemma}\label{Lemma_high_prob_small_xnorm_cond_XMY}
There exists a positive constant $\gamma$ such that, for any $t\in(0, c_2(\log n)^{1/8})$,
\begin{equation}\label{Eqn_Lemma_high_prob_small_xnorm_cond_XMY_goal}
    \Pp\left(\|\vY_{\cW'}\|_1 \ge \gamma n (\log n)^{-7/8}
    \middle| \vX_{\cM'}^\top\vM \vY_{\cW'} < t\right) \le 0.1
\end{equation}
\end{lemma}

\begin{proof}
    To free ourselves from always carrying the notation for the partial matching, let us observe that, once we relinquish the condition on the bistochasticity of $\vM$, it becomes irrelevant that $\mu'$ is a partial matching between $\cM'\subseteq\cM$ and $\cW'\subseteq\cW$ (instead of a complete one between $\cM$ and $\cW$), since the difference in the market size $|\cM|=|\cW|=n$ and $|\cM'|=|\cW'|=n-{\delta n}$ does not affect the final asymptotics in the Lemma. Hence, it suffices to establish a version of \eqref{Eqn_Lemma_high_prob_small_xnorm_cond_XMY_goal} with $\vX_{\cM'}$ and $\vY_{\cW'}$ replaced by non-truncated value vectors $\vX$ and $\vY$ in a complete (instead of partial) matching $\mu$, as long as we do not rely on bistochasticity of $\vM$.

    In the simplified notation, let $\mathbf{Z} = \va \circ \vX$ and $\mathbf{W} = \vb \circ \vY$ with $\va=(a_{i,\mu(i)})_{i\in[n]}$ and $\vb=(b_{j,\mu^{-1}(j)})_{j\in[n]}$, so that $\vZ,\vW\sim \Exp(1)^n$ and are independent. Moreover, let $R=\|\vZ\|_1$ and $\vU=R^{-1}\vZ$ so that, as is well known, $R\sim\Gamma(n, 1)$, $\vU\sim\Unif(\Delta_{n-1})$, and $R$ and $\vU$ are independent. Similarly, let $S=\|\vW\|_1$, and $\vV=S^{-1}\vW$. Then
    \begin{equation*}
        \vX^\top \vM\vY = \vZ^\top \diag(\va^{-1})\vM \diag(\vb^{-1}) \vW = RS \vU^\top \tilde{\vM} \vV,
    \end{equation*}
    where $\tilde{\vM}:=\diag(\va^{-1})\vM \diag(\vb^{-1})$ again has entries bounded on $[1/(C^2n),C^2/n]$. Since $\|\vY\|_1 = \Theta(S)$, it suffice to find a positive constant $\gamma$ such that
    \begin{equation}
        \Pp\left(S \ge \gamma n (\log n)^{-7/8}
    \middle| RS \vU^\top \tilde{\vM} \vV < t\right) < 0.1
    \end{equation}
    for all $n$ sufficiently large and $t\in(0, c_2(\log n)^{1/8})$.
    Note that $1/(C^2n)\le \vU^\top \tilde{\vM} \vV \le C^2/n$ a.s. By conditional on all possible values of $\vU^\top \tilde{\vM} \vV$, it suffices to show that for all $t'\in(0, c_2C^2(\log n)^{1/8})$
    \begin{equation}
        \Pp\left(S \ge \gamma n (\log n)^{-7/8}
    \middle| RS < t' n\right) < 0.1 \label{Eqn_proof_Lemma_high_prob_small_xnorm_cond_XMY_reduced_goal_RS}
    \end{equation}
    asymptotically.
    
    First, we write $\Pp(S \ge \gamma n (\log n)^{-7/8}, RS < t' n)$ as
    \begin{equation*}
        \Pp\big(S \ge \gamma n (\log n)^{-7/8}, RS < t' n\big) = \int_{\gamma n (\log n)^{-7/8}}^\infty G(t'n/s) g(s) ds,
    \end{equation*}
    where $g(x)=\frac{x^{n-1}e^{-x}}{(n-1)!}$ is the probability density function of $\Gamma(n,1)$ and $G$ is the corresponding CDF. Since $t'n/s \ll n$, we may use Lemma~\ref{Lemma_weighted_exp_chernoff} to upper bound $G(t'n/s)$, giving
    \begin{align}
        \Pp\big(S \ge \gamma n (\log n)^{-7/8}, RS < t' n\big) &\le \int_{\gamma n (\log n)^{-7/8}}^\infty \left(\frac{t'e}{s}\right)^n \frac{s^{n-1}e^{-s}}{(n-1)!} ds \nonumber \\
        &= \frac{(t'e)^n}{(n-1)!} \int_{\gamma n (\log n)^{-7/8}}^\infty  \frac{e^{-s}}{s} ds \nonumber \\
        &\le \frac{(t'e)^n}{(n-1)!} e^{-\gamma n (\log n)^{-7/8}}. \label{Eqn_proof_Lemma_high_prob_small_xnorm_cond_XMY_upper_bound_numerator}
    \end{align}
    
    Next, we lower bound $\Pp(RS < t' n)$ by
    \begin{equation*}
        \Pp\left(n^{1/2} \le S \le n^{2/3}, RS < t' n\right) = \int_{n^{1/2}}^{n^{2/3}} G(t'n/s) g(s) ds.
    \end{equation*}
    Note that $t'n/s \le \bO(n^{2/3})$ for all $s\in[n^{1/2},n^{2/3}]$. Using the lower bound in Lemma~\ref{Lemma_weighted_exp_chernoff}, we have
    \begin{align}
        \Pp\left(n^{1/2} \le S \le n^{2/3}, RS < t' n\right) &\ge e^{-\bO(n^{2/3})} \int_{n^{1/2}}^{n^{2/3}} \left(\frac{t'e}{s}\right)^n \frac{s^{n-1}e^{-s}}{(n-1)!} ds \nonumber \\
        &= e^{-\bO(n^{2/3})} \frac{(t'e)^n}{(n-1)!} \int_{n^{1/2}}^{n^{2/3}} \frac{e^{-s}}{s} ds \nonumber \\
        &\ge \frac{(t'e)^n}{(n-1)!} n^{-2/3} e^{-\bO(n^{2/3})}. \label{Eqn_proof_Lemma_high_prob_small_xnorm_cond_XMY_lower_bound_denominator}
    \end{align}
    
    Comparing \eqref{Eqn_proof_Lemma_high_prob_small_xnorm_cond_XMY_upper_bound_numerator} with \eqref{Eqn_proof_Lemma_high_prob_small_xnorm_cond_XMY_lower_bound_denominator} establishes \eqref{Eqn_proof_Lemma_high_prob_small_xnorm_cond_XMY_reduced_goal_RS} and hence finishes the proof.
\end{proof}

\begin{remark}
Note that this lemma should be treated only as a technical result about the typical behavior of $\vX_{\cM'}$ and $\vY_{\cW'}$ when $q(\vX_{\cM'},\vY_{\cW'})=\exp(-n\vX_{\cM'}^\top\vM\vY_{\cW'})$ is large, and should not be confused with any attempt to bound the number of stable (partial) matchings with women's total values in a certain range. For example, one might hope to replace $\gamma n (\log n)^{-7/8}$ with $\gamma n(\log n)^{-1}$ in the proof to conclude that stable matchings with $\|\vY_\delta\|_1\in[n^{1/2},n^{2/3}]$ are over $e^{n^{2/3}}$ times more common than those with $\|\vY_\delta\|_1\ge \Omega(n(\log n)^{-1})$. This, however, is generally not true as we know in the classic case with uniformly random preferences. To see why this fact is not contradictory to our proof, recall from Proposition~\ref{Prop_ratio_p_q_high_prob} (see Section~\ref{sec_prep_proof} and Appendix~\ref{Append_weak_regular_scores}) that $q(\vX_{\cM'},\vY_{\cW'})$ is only a good approximation to $p_{\mu'}(\vX_{\cM'},\vY_{\cW'})$ when, among other conditions, $\|\vY_{\cW'}\|_1 \le \Theta(n(\log n)^{-7/8})$; even then, the approximation is only valid up to an $e^{o(n)}$ factor. As the ratio between \eqref{Eqn_proof_Lemma_high_prob_small_xnorm_cond_XMY_upper_bound_numerator} and \eqref{Eqn_proof_Lemma_high_prob_small_xnorm_cond_XMY_lower_bound_denominator} is only $e^{o(n)}$, the quality of approximation is insufficient for us to rule out the possibility for a (partial) stable matching to have $\|\vY_{\cW'}\|_1\ge\Theta(n(\log n)^{-1})$: the man-optimal stable matching obtained from the man-proposing deferred acceptance algorithm will be such an example.
\end{remark}

\begin{corollary}\label{Cor_x_not_small_when_y_large_cond_on_XMY_and_R1}
There exists a positive constant $\gamma'$ such that, for any $t\in(0, c_2(\log n)^{1/8})$,
\begin{equation}
    \Pp\left(\|\vX_{\cM'}\|_1 \le \gamma' t (\log n)^{7/8}, \|\vY_{\cW'}\|_1 \ge \beta \sqrt{tn}
    \middle| \vX_{\cM'}^\top\vM \vY_{\cW'} < t\right) \le 0.2
\end{equation}
for $n$ sufficiently large, where $\beta$ is the constant appearing in Lemma~\ref{Lemma_x_y_not_both_small_cond_on_XMY_and_R1}.
\end{corollary}
\begin{proof}
    Note that $\|\vM\vY_{\cW'}\|_1 = \|\vY_{\cW'}\|_1 \gtrsim \sqrt{tn} \gg t$ for $t\lesssim (\log n)^{1/8}$. For any $\vy$ supported on coordinates indexed by $\cW'$ with $t \ll \|\vy\|_1 \le \gamma n (\log n)^{-7/8}$, Lemma~\ref{Lemma_wgt_exp_cond_concentration} implies
    \begin{equation}
        \Pp\Big(\|\vX_{\cM'}\|_1 \le 0.9 \frac{t}{n-\floor{\delta n}}\big\|(\vM\vy)_{\cM'}^{-1}\big\|_1 \Big| \vX_{\cM'}^\top\vM \vY_{\cW'} < t, \vY_{\cW'}=\vy\Big) \le e^{-\Theta(n)}.
    \end{equation}
    Plugging in $\big\|(\vM\vy)_{\cM'}^{-1}\big\|_1 \ge (n-\floor{\delta n}) \frac{n}{C\|\vy\|_1} \ge \frac{n-\floor{\delta n}}{C\gamma}(\log n)^{7/8}$ gives
    \begin{equation}
        \Pp\big(\|\vX_{\cM'}\|_1 \le 0.9(C\gamma)^{-1} t(\log n)^{7/8} \big| \vX_{\cM'}^\top\vM \vY_{\cW'} < t, \vY_{\cW'}=\vy\big) \le e^{-\Theta(n)}.
    \end{equation}
    Marginalizing over all relevant values of $\vy$ implies
    \begin{equation}
        \Pp\big(\|\vX_{\cM'}\|_1 \le \gamma' t(\log n)^{7/8}, \beta \sqrt{tn} \le \|\vY_{\cW'}\|_1 \le \gamma n (\log n)^{-7/8} \big| \vX_{\cM'}^\top\vM \vY_{\cW'} < t\big) \le e^{-\Theta(n)}
    \end{equation}
    with $\gamma'=0.9(C\gamma)^{-1}$. Combining this with Lemma~\ref{Lemma_high_prob_small_xnorm_cond_XMY} completes the proof.
\end{proof}

\begin{corollary}
For any $t\in(0, c_2(\log n)^{1/8})$ and $n$ sufficiently large,
\begin{equation}
    \Pp\left(\|\vX_{\cM'}\|_1 \wedge \|\vY_{\cW'}\|_1 \ge \gamma' t (\log n)^{7/8}, \vX_{\cM'}^\top\vM \vY_{\cW'} < t\right) \ge \frac{1}{2} \Pp\left( \vX_{\cM'}^\top\vM \vY_{\cW'} < t\right).
\end{equation}
\end{corollary}
\begin{proof}
This is a direct consequence of Lemma~\ref{Lemma_x_y_not_both_small_cond_on_XMY_and_R1} and Corollary~\ref{Cor_x_not_small_when_y_large_cond_on_XMY_and_R1}
\end{proof}

We are now ready to state the proof of Proposition~\ref{Prop_EqXY_bound}.

\begin{proof}[Proof of Proposition~\ref{Prop_EqXY_bound}]
    Define events
    \begin{equation*}
        \begin{aligned}[c]
            A_1(t)&: \|\vX_{\cM'}\|_1 \ge \gamma' t(\log n)^{1/8},\\
            B_1(t)&: (\vX_{\cM'},\vY_{\cW'})\in\Omega_1(\zeta;t),
        \end{aligned}
        \qquad 
        \begin{aligned}[c]
            A_2(t)&: \|\vY_{\cW'}\|_1 \ge \gamma' t(\log n)^{1/8},\\
            B_2(t)&: (\vX_{\cM'},\vY_{\cW'})\in\Omega_2(\zeta;t),
        \end{aligned}
    \end{equation*}
    where $\Omega_1$ and $\Omega_2$ are defined in Appendix~\ref{Append_proof_prop_eigenvec_of_M}, and $\zeta$ is to be specified later.
    We have
    \begin{align}
        \frac{1}{2}\Pp(\vX_{\cM'}^\top\vM \vY_{\cW'} < t) &\le \Pp(\vX_{\cM'}^\top\vM \vY_{\cW'} < t, A_1(t), A_2(t)) \nonumber \\
        &\le \Pp(\vX_{\cM'}^\top\vM \vY_{\cW'} < t, A_1(t), A_2(t), B_1(t), B_2(t)) \nonumber \\
        &\qquad + \Pp(\vX_{\cM'}^\top\vM \vY_{\cW'} < t, A_1(t), B_1(t)^c) + \Pp(\vX_{\cM'}^\top\vM \vY_{\cW'} < t, A_2(t), B_2(t)^c) \nonumber \\
        &\le \Pp(\vX_{\cM'}^\top\vM \vY_{\cW'} < t, A_1(t), A_2(t), B_1(t), B_2(t)) \nonumber \\
        &\qquad + \Pp(B_1(t)^c | \vX_{\cM'}^\top\vM \vY_{\cW'} < t, A_1(t)) \cdot \Pp(\vX_{\cM'}^\top\vM \vY_{\cW'} < t) \nonumber \\
        &\qquad + \Pp(B_2(t)^c | \vX_{\cM'}^\top\vM \vY_{\cW'} < t, A_2(t)) \cdot \Pp(\vX_{\cM'}^\top\vM \vY_{\cW'} < t).
    \end{align}
    
    For any fixed $\delta>0$ and $\zeta=o_\delta(1)$, $\Pp(B_1(t)^c | \vX_{\cM'}^\top\vM \vY_{\cW'} < t, A_1(t)) \to 0$ by Lemma~\ref{Lemma_wgt_exp_cond_concentration}; in particular, we may assume that $\Pp(B_1(t)^c | \vX_{\cM'}^\top\vM \vY_{\cW'} < t, A_1(t))$ (and by symmetry $\Pp(B_2(t)^c | \vX_{\cM'}^\top\vM \vY_{\cW'} < t, A_2(t))$) is at most $1/8$.
    Thus,
    \begin{equation}
        \Pp(\vX_{\cM'}^\top\vM \vY_{\cW'} < t) \le 4 \Pp(\vX_{\cM'}^\top\vM \vY_{\cW'} < t, A_1(t), A_2(t), B_1(t), B_2(t)).
    \end{equation}
    
    By Lemma~\ref{lemma_Omega1_and_Omega2_suggests_Oeig}, $B_1(t)$ and $B_2(t)$ together imply that $n\vX_{\cM'}^\top\vM \vY_{\cW'} = (1 + o_{\delta}(1)) \|\vX_{\cM'}\|_1\|\vY_{\cW'}\|_1$. Further, along with the events $\vX_{\cM'}^\top\vM \vY_{\cW'} < t$, $A_1(t)$, and $A_2(t)$, they imply $t(\log n)^{1/8} \lesssim \|\vY_{\cW'}\|_1 \lesssim n(\log n)^{-1/8}$. Hence,
    \begin{align}
        \Pp(\vX_{\cM'}^\top\vM &\vY_{\cW'} < t, A_1(t), A_2(t), B_1(t), B_2(t)) \nonumber \\
        &\le \Pp\left(\|\vX_{\cM'}\|_1\|\vY_{\cW'}\|_1 \le \frac{nt}{1+o_{\delta,\zeta}(1)}, t(\log n)^{1/8} \le \|\vY_{\cW'}\|_1 \le \Theta(n(\log n)^{-1/8})\right) \nonumber \\
        &\le \E\left[\Pp\left(\|\vX_{\cM'}\|_1\|\vY_{\cW'}\|_1 \le \frac{nt}{1+o_{\delta,\zeta}(1)}, t(\log n)^{1/8} \le \|\vY_{\cW'}\|_1 \le \Theta(n(\log n)^{-1/8}) \middle| \|\vY_{\cW'}\|_1\right)\right] \nonumber \\
        &\le e^{o_{\delta}(n)} \bigg(\frac{ent}{n-\floor{\delta n}}\bigg)^{n-\floor{\delta n}} \prod_{i\in\cM'} a_{i,\mu'(i)} \nonumber \\
        &\qquad \cdot \E\left[\|\vY_{\cW'}\|_1^{-n+\floor{\delta n}}; t(\log n)^{1/8} \le \|\vY_{\cW'}\|_1 \le \Theta(n(\log n)^{-1/8})\right]. \label{Eqn_proof_prop_EqXY_P_XMY_and_A1A2B1B2}
    \end{align}
    
    It is straightforward, albeit a bit tedious, to explicitly bound the expectation term in \eqref{Eqn_proof_prop_EqXY_P_XMY_and_A1A2B1B2} by
    \begin{equation*}
        (\Theta(\log n) - \log t) e^{n-\floor{\delta n}} \prod_{i\in\cM'} b_{\mu'(i),i},
    \end{equation*}
    again using Lemma~\ref{Lemma_weighted_exp_chernoff}. Carrying out the integral over $t$ in \eqref{Eqn_proof_prop_6_6_goal} finishes the proof.
\end{proof}

\subsection{Proof of Lemma~\ref{Prop_Oempe_likely_for_happiness_emp_distr}}\label{Proof_prop_Oempe_likely}

In this section, we present the proof of Lemma~\ref{Prop_Oempe_likely_for_happiness_emp_distr}, restated below.

\propOempeLikelyForHappinessEmpDist*
\begin{proof}
In light of Proposition~\ref{Prop_EqXY_bound}, it suffices to show that for all $\vy \in \text{Proj}_y(\mathcal{R}\cap\Oeigz)$, we have
\begin{equation}\label{Eqn_want_ratio_EOempe_EqXY}
    \frac{\E\big[q(\vX_{\cM'}, \vY_{\cW'}) \cdot \mathbbm{1}_{\mathcal{R} \cap \Oeigz\backslash\Oempe}(\vX_{\cM'}, \vY_{\cW'})  \;|\; \vY_{\cW'} = \vy \big]}{\E\big[q(\vX_{\cM'}, \vY_{\cW'}) \cdot \mathbbm{1}_{\mathcal{R}_2}(\vX_{\cM'}, \vY_{\cW'}) \;|\; \vY_{\cW'} = \vy \big]} \le \exp(-\Theta(\eps^2 n)).
\end{equation}
It then follows that
\begin{align}
    \E\big[q(&\vX_{\cM'}, \vY_{\cW'}) \cdot \mathbbm{1}_{\mathcal{R} \cap \Oeigz\backslash\Oempe}(\vX_{\cM'}, \vY_{\cW'})\big] \nonumber \\
    &= \E\big[\E\big[q(\vX_{\cM'}, \vY_{\cW'}) \cdot \mathbbm{1}_{\mathcal{R} \cap \Oeigz\backslash\Oempe}(\vX_{\cM'}, \vY_{\cW'})\big| \vY_{\cW'}\big]\big] \nonumber \\
    &\le \E\big[\exp(-\Theta(\eps^2 n)) \cdot \E[q(\vX_{\cM'}, \vY_{\cW'}) \cdot \mathbbm{1}_{\mathcal{R}_2}(\vX_{\cM'}, \vY_{\cW'}) \;|\; \vY_{\cW'}] \cdot \mathbbm{1}_{\text{Proj}_y(\mathcal{R}\cap\Oeigz)}(\vY_{\cW'})\big] \nonumber \\
    &\le \exp(-\Theta(\eps^2 n)) \cdot \E\big[\E[q(\vX_{\cM'}, \vY_{\cW'}) \cdot \mathbbm{1}_{\mathcal{R}_2}(\vX_{\cM'}, \vY_{\cW'}) \;|\; \vY_{\cW'}] \cdot \mathbbm{1}_{\mathcal{R}_1}(\vY_{\cW'})\big] \nonumber \\
    &\le \exp(-\Theta(\eps^2 n)) \cdot \E[q(\vX_{\cM'}, \vY_{\cW'}) \cdot \mathbbm{1}_{\mathcal{R}_2}(\vX_{\cM'}, \vY_{\cW'}) \cdot \mathbbm{1}_{\mathcal{R}_1}(\vY_{\cW'})],
\end{align}
and Proposition~\ref{Prop_EqXY_bound} immediately implies the desired bound.

To show \eqref{Eqn_want_ratio_EOempe_EqXY}, notice that the quotient in the left-hand side is simply
\begin{equation}\label{Eqn_proof_distr_happi_equiv_emp_tail_bound_for_nearly_iid_X}
    \Pp_{\vX\sim\bigotimes_{i=1}^n \Exp(a_{i,\mu'(i)}+n(\vM\vy)_i)}\big((\vX_{\cM'},\vy)\in\mathcal{R}\cap\Oeigz\backslash\Oempe\big).
\end{equation}
Recall that $(\vX_{\cM'})_i = X_i$ for $i \in \cM'$ and $(\vX_{\cM'})_i = 0$ for $i\notin\cM'$. For any $\vy \in \text{Proj}_y(\mathcal{R}\cap\Oeigz)$, there must exist $\hat y\in\R_+$ such that for all but at most $\sqrt{\zeta} n$ indices $i\in[n]$ we have $|(\vM \vy)_i-\hat y| \le \sqrt{\zeta} \hat y$. In other words, under the distribution $\vX\sim\bigotimes_{i=1}^n \Exp(a_{i,\mu'(i)}+n(\vM\vy)_i)$, for all but at most $(\delta + \sqrt{\zeta}) n$ indices $i\in[n]$, we have $n\hat y X_i\sim \Exp\big(\lambda_i\big)$ for some
\begin{equation*}
    \lambda_i = \frac{a_{i,\mu'(i)}}{n\hat y}+\frac{(\vM\vy)_i}{\hat y} = 1 + \Theta(\sqrt{\zeta}) + \Theta(1/\log n),
\end{equation*}
where we used the fact that $n\hat y\ge\Theta(\|\vy\|_1) \ge \Theta(\log n)$ as implied by $\vy\in\text{Proj}_y(\mathcal{R}\cap\Oeigz)$. The generalized Dvoretzky–Kiefer–Wolfowitz (DKW) inequality (see Lemma~\ref{Lemma_dkw_non_identical}) for independent and nearly-identically distributed random variables implies that the probability \eqref{Eqn_proof_distr_happi_equiv_emp_tail_bound_for_nearly_iid_X} is upper bounded by
\begin{equation}
    \Pp_{\vX\sim\bigotimes_{i=1}^n \Exp(a_{i,\mu'(i)}+n(\vM\vy)_i)}\big(\big\|\Femp(\vx) - F_{n\hat y}\|_\infty > \eps + \Theta(\delta + \sqrt{\zeta})\big) \le \exp(-\Theta(\eps^2 n)),
\end{equation}
which finishes the proof.
\end{proof}

\subsection{Proof of Theorem~\ref{Thm_main_rank_dist_body}}\label{Append_proof_thm_main_rank}

Heuristically, we would expect the rank $R_i$ for a man to be proportional to his value $X_i(\mu)$. We will see below that this is approximately the case when $x_i\ll 1$. There are, however, going to be some $x_i$ of constant order, making it hard for us to say anything exact about $R_i$. But as we will soon see, for all but a $o(1)$ fraction of the $n$ men, we will indeed have $x_i = o(1)$. As we are concerned with the empirical distribution, such small fraction becomes negligible in the limit and can be simply ignored. This heuristics is formalized in the next Lemma.

\begin{lemma}\label{Lemma_probable_happiness_majority}
Fix any $\delta > 0$. Let $\mu'$ be a partial matching of size $n-\floor{\delta n}$ on $\cM'\subseteq\cM$ and $\cW'\subseteq\cW$. For any $0<\xi<\rho < 1$, consider $\Otailab$ defined as
\begin{equation}
    \bigg\{(\vx,\vy)\in \R_+^n\times\R_+^n : \sumiton \mathbbm{1}\Big\{n x_i (\vM\vy)_i\notin (F^{-1}(\xi/2)/2, F^{-1}(1-\xi/2))\Big\} \le \floor{\delta n} + \rho (n-\floor{\delta n})\bigg\}.
\end{equation}
Then
\begin{equation}\label{Eqn_Lemma_probable_happiness_majority_main_bound}
    \E\big[q(\vX_{\cM'}, \vY_{\cW'}) \cdot \mathbbm{1}_{\mathcal{R} \backslash\Otailab}(\vX_{\cM'}, \vY_{\cW'})\big] \le \exp(o_\delta(n)-\Theta(D(\rho\|\xi)n)) \cdot \frac{(\delta n)!}{n!} \prod_{i\in\cM'}a_{i,\mu'(i)}b_{\mu'(i),i},
\end{equation}
where $D(q\|p)$ denotes the KL divergence from $\Bern(p)$ to $\Bern(q)$.
\end{lemma}
\begin{proof}
The proof entirely mirrors that of Lemma~\ref{Prop_Oempe_likely_for_happiness_emp_distr}. It suffices to show that for all $\vy\in\text{Proj}_y(\mathcal{R})$ we have
\begin{equation}\label{Eqn_want_ratio_EOtailab_EqXY}
    \frac{\E\big[q(\vX_{\cM'}, \vY_{\cW'}) \cdot \mathbbm{1}_{\mathcal{R} \backslash\Otailab}(\vX_{\cM'}, \vY_{\cW'})  \;|\; \vY_{\cW'} = \vy \big]}{\E\big[q(\vX_{\cM'}, \vY_{\cW'}) \;|\; \vY_{\cW'} = \vy \big]} \le \exp(-\Theta(D(\rho\|\xi)n)).
\end{equation}
The quotient is simply
\begin{equation}\label{Eqn_proof_probable_happiness_majority_rewrite_x_distr_cond_on_y}
    \Pp_{\vX\sim\bigotimes_{i=1}^n \Exp(a_{i,\mu'(i)}+n(\vM\vy)_i)}\big((\vX_{\cM'},\vy)\in\mathcal{R}\backslash\Otailab\big) \le \Pp\big((\vX_{\cM'},\vy)\notin\Otailab\big),
\end{equation}
where we will have $\vX\sim\prod_{i=1}^n \Exp(a_{i,\mu'(i)}+n(\vM\vy)_i)$ for the rest of this proof.
Note that under this specified distribution, $\Big(\big(a_{i,\mu'(i)}+n(\vM\vy)_i\big)X_i\Big)_{i\in\cM'}$ are $n-\floor{\delta n}$ i.i.d. samples from $\Exp(1)$, each falling outside the interval $(F^{-1}(\xi/2), F^{-1}(1-\xi/2))$ with probability precisely $\xi$. Hence,
\begin{multline}\label{Eqn_proof_probable_happiness_majority_hoeffding_for_renormalized_x}
    \Pp\bigg(\sum_{i\in\cM'} \mathbbm{1}\Big\{X_i\big(a_{i,\mu'(i)}+n(\vM\vy)_i\big)\notin (F^{-1}(\xi/2), F^{-1}(1-\xi/2))\Big\} \le \rho (n-\floor{\delta n})\bigg) \\
    = \Pp_{Z\sim\Binom(n-\floor{\delta n}, \xi)}(Z > \rho(n-\floor{\delta n})) \le \exp(-D(\rho\|\xi)(n-\floor{\delta n}))
\end{multline}
by the Hoeffding bound for binomial distribution. Since $n(\vM\vy)_i \le a_{i,\mu'(i)}+n(\vM\vy)_i \le 2n(\vM\vy)_i$ across all $i\in\cM'$ for $\vy\in\mathcal{R}_1$ and $n$ sufficiently large, the probability \eqref{Eqn_proof_probable_happiness_majority_hoeffding_for_renormalized_x} upper bounds $\Pp\big((\vX_{\cM'},\vy)\notin\Otailab\big)$. This establishes \eqref{Eqn_want_ratio_EOtailab_EqXY} and concludes the proof.
\end{proof}

By fixing some small $\delta,\rho$ and choosing $\xi$ sufficiently small, we can make $D(\rho\|\xi)$ arbitrarily large and obtain the following Corollary.

\begin{corollary}\label{Cor_no_stable_match_tilde_Otailab}
For any $0 < \delta,\rho < 1/2$, there exists a choice of $\xi > 0$ such that %
\begin{equation}
    \Pp(\exists \mu\in\mathcal{S}, (\vX_\delta(\mu),\vY_\delta(\mu))\notin\tOtailab) \lesssim e^{-n^c}
\end{equation}
asymptotically as $n\to\infty$,
where $\tOtailab$ is defined as
\begin{equation}
    \bigg\{ (\vx,\vy)\in \R_+^n\times\R_+^n : \sumiton \mathbbm{1}\Big\{x_i\notin \Big(F^{-1}(\xi/2)\frac{(\log n)^{7/8}}{C^2\overline{c}_1 n}, 2F^{-1}(1-\xi/2)\frac{C^2}{\underline{c}_1 \log n}\Big)\Big\} \le (\delta + \rho) n \bigg\}.
\end{equation}
That is, with high probability, no stable matchings $\mu$ have more than $\delta+\rho$ fraction of the men's post-truncation values outside an interval $(\Theta(n^{-1}(\log n)^{7/8}), \Theta(1/\log n))$.
\end{corollary}
\begin{proof}
Observe that $\mathcal{R}\backslash\tOtailab \subseteq \mathcal{R}\backslash\Otailab$ by our definition of $\mathcal{R}$ and the boundedness assumption on the entries of $\vM$. Again, invoking Lemma~\ref{Lemma_reduction_to_q} using inequality \eqref{Eqn_Lemma_probable_happiness_majority_main_bound} in Lemma~\ref{Lemma_probable_happiness_majority} yields the $\Theta(e^{-n^c})$ upper bound on $\Pp((\vX_\delta(\mu),\vY_\delta(\mu))\notin\tOtailab)$.
\end{proof}

\begin{remark}
Recall that in a market with uniform preferences, the man-optimal (and woman-pessimal) stable matching realizes an average rank of $\Theta(\log n)$ for men and $\Theta(n/\log n)$ for women. Under the heuristics that values multiplied by $n$ roughly correspond to ranks (which we will formalize below), Lemma~\ref{Lemma_probable_happiness_majority} nicely matches our expectation that even in the most extreme cases, few individuals will strike a rank better (smaller) than $\Theta((\log n)^{7/8})$ or worse (larger) than $\Theta(n/\log n)$. The lower bound can be refined to $\Theta(\log n/\log \log n)$ with a more careful analysis.
\end{remark}

Now let us consider a specific partial matching $\mu'$ of size $n-\floor{\delta n}$ between $\cM'$ and $\cW'$ and condition on $\mu'$ being stable with value vectors $(\vX_{\cM'}, \vY_{\cW'})=(\vx,\vy)\in\tOtailab$. That is, there exists a subset $\bar{\cM}'\subseteq\cM'$ with $|\bar{\cM}'|\ge (1-\delta-\rho)n$ such that $\Theta(n^{-1}(\log n)^{7/8}) \le (X_{\cM'})_i \le \Theta(1/\log n)$ for all $i\in\bar{\cM}'$. By symmetry, we may further assume that there exists $\bar{\cW}'\subseteq\cW'$ with $|\bar{\cW}'|\ge (1-\delta-\rho)n$ such that $\Theta(n^{-1}(\log n)^{7/8}) \le (Y_{\cW'})_j \le \Theta(1/\log n)$ for all $j\in\bar{\cW}'$. We want to show that, for $i\in\bar{\cM}'$, the \emph{pre-truncation} rank $R_i$ of man $m_i$ (i.e., over the entire market, including the $\floor{\delta n}$ women outside $\cM'$) is well characterized by his value $X_{i,\,u'(i)}$ in the matching, up to some proper scaling. From now on, we will consider some $i\in\bar{\cM}'$ with value $X_{i,\,u'(i)}=x_i$, and write
\begin{equation}\label{Eqn_def_eqn_rank_i}
    R_i = 1 + \sum_{j\ne \mu'(i)}\mathbbm{1}_{[0,x_i]}(X_{ij}).
\end{equation}
The condition that $\mu'$ is stable requires $(X_{ij}, Y_{ji}) \notin [0, x_i]\times[0,y_j]$ for all $j\in\cW'\backslash\{\mu'(i)\}$. Thus, for all $j\in\cW'\backslash\{\mu'(i)\}$,
\begin{multline}
    \Pp(X_{ij} \le x_i | \mu'\text{ stable}, (\vX_{\cM'})_i=x_i, (\vY_{\cW'})_j=y_j) = \frac{\Pp(X_{ij} \le x_i,Y_{ji} > y_j)}{1 - \Pp(X_{ij} \le x_i,Y_{ji} \le y_j)} \\
    = \frac{(1-e^{-a_{ij}x_i})e^{-b_{ji}y_j}}{1-(1-e^{-a_{ij}x_i})(1-e^{-b_{ji}y_j})},
\end{multline}
and for all $j\in\cW\backslash\cW'$ (so $(\vY_{\cW'})_j=0$),
\begin{equation}
    \Pp(X_{ij} \le x_i | \mu'\text{ stable}, (\vX_{\cM'})_i=x_i) = 1-e^{-a_{ij}x_i}.
\end{equation}
Define
\begin{equation*}
    p_{ij} = \begin{cases}
    1 & \quad \text{ when }j=\mu'(i), \\
    \frac{(1-e^{-a_{ij}x_i})e^{-b_{ji}y_j}}{1-(1-e^{-a_{ij}x_i})(1-e^{-b_{ji}y_j})} & \quad \text{ when } j\in\cW'\backslash\{\mu'(i)\}, \\
    1-e^{-a_{ij}x_i} & \quad\text{ when } j\in\cW\backslash\cW',
    \end{cases}
\end{equation*}
and $I_{ij}\sim \Bern(p_{ij})$ independently for $i\in[n]$ so that $R_i = \sumjton I_{ij}$ conditional on $(\vX_{\cM'})_i = X_i=x_i$.
Note that for any $j\ne \mu'(i)$ and $j\notin \bar{\cW}'$, we always have
\begin{equation}\label{Eqn_relate_rank_to_happi_p_ij_upperbound}
    p_{ij} \le 1-e^{-a_{ij}x_i} \le a_{ij}x_i
\end{equation}
and
\begin{multline}\label{Eqn_relate_rank_to_happi_p_ij_lowerbound}
    p_{ij} \ge \frac{(1-e^{-a_{ij}x_i})e^{-b_{ji}y_j}}{1-(1-e^{-a_{ij}x_i})(1-e^{-b_{ji}y_j})} \ge (1-e^{-a_{ij}x_i})e^{-b_{ji}y_j} \\
    \ge e^{-\Theta(\frac{1}{\log n})}\bigg(1 - \Theta\Big(\frac{1}{\log n}\Big)\bigg) a_{ij}x_i = (1-o(1))a_{ij}x_i.
\end{multline}
For $j \in \bar{\cW}'\backslash\{\mu'(i)\}$, $p_{ij}$ admits the same upper bound \eqref{Eqn_relate_rank_to_happi_p_ij_upperbound} and the trivial lower bound of zero.
Hence, conditional on
\begin{equation}\label{Eqn_condition_stable_with_vx_vy}
    \mu'\text{ stable and }(\vX_{\cM'},\vY_{\cW'})=(\vx,\vy)\in\tOtailab\cap \tOemp(\eps) \cap\mathcal{R} \tag{$\dagger$}
\end{equation}
for any (fixed) $\xi,\rho,\eps> 0$,
we have the stochastic dominance
\begin{equation}
    1 + \sum_{j\notin \bar{\cW}'\cup\{\mu'(i)\}} \underline{I}_{ij} \preceq R_i \preceq 1 + \sum_{j\ne \mu'(i)} \overline{I}_{ij},
\end{equation}
where $\underline{I}_{ij}\sim\Bern((1-o(1))a_{ij}x_i)$ and $\overline{I}_{ij}\sim\Bern(a_{ij}x_i)$. Since $i\in\bar{\cM}'$ by our assumption and thus $\Theta((\log n)^{7/8}/n)\le x_i \le \Theta(1/\log n)$, the expectation of $R_i/x_i$ can be upper bounded by
\begin{equation}
    \E\bigg[\frac{R_i}{x_i} \bigg| \eqref{Eqn_condition_stable_with_vx_vy}\bigg] \le \frac{1}{x_i} + \sum_{j\ne\mu'(i)}a_{ij} = (1+o(1))\sumjton a_{ij}
\end{equation}
and lower bounded by
\begin{equation}
    \E\bigg[\frac{R_i}{x_i} \bigg| \eqref{Eqn_condition_stable_with_vx_vy}\bigg] \ge (1-o(1))\sum_{j\ne\bar{\cW}'}a_{ij} = (1-\Theta(\delta))\sumjton a_{ij}.
\end{equation}
Similarly, we may bound the variance of $R_i/x_i$ by
\begin{equation}
    \Var\bigg(\frac{R_i}{x_i} \bigg| \eqref{Eqn_condition_stable_with_vx_vy}\bigg) \le \sum_{j\ne\mu'(i)}a_{ij} (1-a_{ij}x_i) \le \sumjton a_{ij}.
\end{equation}
Hence, we have
\begin{equation}\label{Eqn_final_E_Var_bound_for_rank_happiness_ratio}
    1-\Theta(\delta) \le \E\bigg[\frac{R_i}{x_i \sumjton a_{ij}} \bigg| \eqref{Eqn_condition_stable_with_vx_vy}\bigg] \le 1+o(1) \enspace \text{ and } \enspace \Var\bigg(\frac{R_i}{x_i \sumjton a_{ij}} \bigg| \eqref{Eqn_condition_stable_with_vx_vy}\bigg) \le \Theta(n^{-1}),
\end{equation}
with these quantities conditionally independent for all $i\in\bar{\cM}'$ and the hidden constants depending only on $C$, implying concentration of $R_i$ around $x_i\sumjton a_{ij}$ in the following sense.

\begin{proposition}\label{Prop_most_have_good_rank_happi_ratio}
Conditional on \eqref{Eqn_condition_stable_with_vx_vy}, for any fixed $\theta > 0$ and $\delta,\rho,\gamma\in(0,1/2)$, we have
\begin{equation}
    \Pp\left(\sumiton \mathbbm{1}_{(\theta + \Theta(\delta), \infty)}\bigg(\bigg|\frac{R_i}{x_i\sumjton a_{ij}} - 1\bigg|\bigg) \ge (\delta+\rho+\gamma)n \middle| \eqref{Eqn_condition_stable_with_vx_vy}\right)
    \lesssim \Pp_{N\sim\Poi(\Theta(\theta^{-2}))}(N\ge \gamma n)
    \le e^{-\Theta(\gamma n)}.
\end{equation}
\end{proposition}
\begin{proof}
By Chebyshev's inequality and \eqref{Eqn_final_E_Var_bound_for_rank_happiness_ratio}, $\Pp\big(\big|\frac{R_i}{x_i\sumjton a_{ij}} - \E\big[\frac{R_i}{x_i \sumjton a_{ij}} \big| \eqref{Eqn_condition_stable_with_vx_vy}\big]\big| \ge \theta \big| \eqref{Eqn_condition_stable_with_vx_vy}\big) \le \Theta\big((n\theta^2)^{-1}\big)$ for all $i\in\bar{\cM}'$. Hence, by conditional independence of the ranks, $\sum_{i\in\bar{\cM}'} \mathbbm{1}_{(\theta + \Theta(\delta), \infty)}\big(\big|\frac{R_i}{x_i\sumjton a_{ij}} - 1\big|\big)$ is stochastically dominated by $\Binom\big(n, \Theta\big((n\theta^2)^{-1}\big)\big)$, which converges to $\Poi(\Theta(\theta^{-2}))$ in total variance. The Proposition follows from the well known tail bound for $N\sim\Poi(\lambda)$ that $\Pp(N\ge \lambda + t) \le \exp\big(-\frac{t^2}{2(\lambda+t)}\big)$, which implies $\Pp_{N\sim\Poi(\Theta(\theta^{-2}))}(N\ge \gamma n) \le \exp\big(-\frac{(\gamma n - \Theta(\theta^{-2}))^2}{2\gamma n}\big) \lesssim \exp(-\gamma n/2)$.
\end{proof}

\begin{proof}[Proof of Theorem~\ref{Thm_main_rank_dist_body}]
Note that in Corollary~\ref{Cor_no_stable_match_tilde_Otailab} and Corollary~\ref{Cor_no_stable_match_tilde_Otailab}, $\rho$ can be chosen arbitrarily small once $\delta$ is fixed. In particular, we may always guarantee $\rho \le \delta$. Similarly, we may assume $\theta \le \delta$ in Proposition~\ref{Prop_most_have_good_rank_happi_ratio}. Thus,
\begin{equation}
    \Pp\left(\sumiton \mathbbm{1}_{(\Theta(\delta), \infty)}\bigg(\bigg|\frac{R_i}{x_i w_i} - 1\bigg|\bigg) \ge (2\delta+\gamma)n \middle| \eqref{Eqn_condition_stable_with_vx_vy}\right)
    \le e^{-\Theta(\gamma n)},
\end{equation}
where $w_i = \sumjton a_{ij}$ is the fitness value of man $m_i$.
Marginalizing over all pairs of relevant value vectors $(\vx,\vy)\in\tOtailab\cap\tOemp(\eps)\cap\mathcal{R}$ in the condition \eqref{Eqn_condition_stable_with_vx_vy}, we obtain
\begin{equation}
    \Pp\left(\Eratio(\delta,\gamma) \middle| (\vX_{\cM'},\vY_{\cW'})\in\tOtailab\cap\tOemp(\eps)\cap\mathcal{R}\right)
    \le e^{-\Theta(\gamma n)},
\end{equation}
where $\Eratio(\delta,\gamma)$ denotes the undesirable event that $\sumiton \mathbbm{1}_{(\Theta(\delta), \infty)}\big(\big|\frac{R_i}{(\vX_{\cM'})_i w_i} - 1\big|\big) \ge (2\delta+\gamma)n$ for the partial matching $\mu'$.
By Proposition~\ref{Prop_EqXY_bound}, 
\begin{multline}
    \Pp((\vX_{\cM'},\vY_{\cW'})\in\tOtailab\cap\Oempe\cap\mathcal{R}) \le \Pp(\mu'\text{ stable}, (\vX_{\cM'},\vY_{\cW'})\in\mathcal{R}) \\
    \le e^{o(n)+o_\delta(n)} \frac{(\delta n)!}{n!}\prod_{i\in\cM'} a_{i,\mu'(i)} b_{\mu'(i),i}.
\end{multline}
By choosing $\gamma=\gamma(\delta)=o_\delta(1)$ sufficiently large (relative to $\delta$) and following a similar computation as in Lemma~\ref{Lemma_reduction_to_q} and Corollary~\ref{Cor_subexp_num_stable_match}, we can ensure that with probability $1-\Theta(e^{-n^c})$ there exists no stable partial matching $\mu'$ where both $\Eratio(\delta,\gamma)$ and $(\vX_{\cM'},\vY_{\cW'})\in \tOemp(\eps_0(\delta))\cap\mathcal{R}$ happen, where the function $\eps_0$ is defined in the proof of \ref{Thm_main_happiness_dist_body}. Notice that by repeated uses of the triangle inequality, 
\begin{equation}
    (\vX_{\cM'},\vY_{\cW'})\in \tOemp(\eps_0(\delta)), \Eratio(\delta,\gamma)^c \enspace \Longrightarrow \enspace \|\Femp(\mathbf{w}^{-1}\circ\mathbf{R}(\mu'))-F_\lambda\|_\infty \le \Theta(\delta) + \gamma(\delta) + \eps_0(\delta) = o_\delta(1)
\end{equation}
for the choice of $\lambda = \|\vY_{\cW'}\|_1$. Combining this with \eqref{Eqn_proof_Thm_main_happiness_dist_tOemp} and Proposition~\ref{Prop_R_likely}, we conclude that with probability $1-\Theta(e^{-n^c})$, all stable matchings $\mu\in\mathcal{S}$ induce $\delta$-truncated stable partial matchings $\mu_\delta$ with $\|\Femp(\mathbf{w}^{-1}\circ\mathbf{R}(\mu_\delta))-F_{\lambda(\mu)}\|_\infty = o_\delta(1)$, where $\lambda(\mu)=\|\vY_{\delta}(\mu)\|_1$. The $\delta$-truncation affects the distance by at most $\delta$, which can be absorbed into the $o_\delta(1)$ upper bound. Thus, by choosing $\delta$ sufficiently small relative to any fixed $\eps > 0$, we complete our proof of Theorem~\ref{Thm_main_rank_dist_body}.
\end{proof}

\subsection{Proofs of Theorem~\ref{Thm_dist_body_approx_stable} and Corollary~\ref{Cor_body_imbalance}}

\ThmMainApproxStable*
\begin{proof}
    There are $\binom{n}{\alpha n}^2 = \exp(2 h_b(\alpha) n + O(\ln n))$ sub-markets of size at least $(1-\alpha) n$, where $h_b(p) = -p \log p - (1-p) \log (1-p)$ is the binary entropy function. Under Assumption~\ref{Assumption_C_bounded} for the whole market, each of such sub-markets also satisfies Assumption~\ref{Assumption_C_bounded}. Fix any $\eps > 0$. By Theorem~\ref{Thm_main_happiness_dist_body}, each of such sub-markets can only contain a stable matching with men's empirical distribution of value deviating from  the family of exponential distributions by at least $\eps/2$ in Kolmogorov-Smirnov distance with probability at most $1-\exp(-n^c)$ for any fixed $c\in(0,1/2)$. Whenever $\alpha < n^{-\eta}$ for some $\eta > 1/2$, we have $h_b(\alpha) n < n^{1-\eta}$. Choosing $c \in (1-\eta, 1/2)$ and applying union bound over all relevant sub-markets gives the first part of \eqref{Eqn_happiness_dist_approx_stable}, since the additional $\alpha$ fraction of the market affects the empirical distribution by at most $\alpha\to 0$. The second part follows analogously from Theorem~\ref{Thm_main_rank_dist_body}.
\end{proof}

\CorImbalanceMarket*
\begin{proof}
    The proof is entirely the same as the proof of Theorem~\ref{Thm_dist_body_approx_stable}. The union bound covers all sub-markets of size $n-k$, that is, consisting all the men and a subset of the women. The rest is the same.
\end{proof}

\end{document}